\documentclass[11pt,draftclsnofoot,onecolumn]{IEEEtran}

\usepackage{amsmath,amsfonts,amssymb,amsthm}
\usepackage{mathrsfs}
\usepackage{comment,blkarray}
\usepackage{multirow,bigdelim}
\usepackage{mathrsfs}
\usepackage{cite}
\usepackage[ruled,commentsnumbered, vlined]{algorithm2e}
\usepackage{tikz}
\usepackage{verbatim}
\usepackage{graphicx}
\usepackage{subfigure}		
\usepackage{booktabs}	
\usepackage{xcolor}
\usepackage{colortbl}													

\includecomment{itwfull}
\excludecomment{removeEX4}
\excludecomment{itw2016}
\excludecomment{journalonly}

\usepackage{url}
\usepackage{subfigure}
\usepackage{multicol}

\usepackage{ifpdf}															%
\ifpdf																		%
	\usepackage{hyperref}
\else																		%
\fi	
\usetikzlibrary{arrows,decorations.pathmorphing,decorations.footprints,fadings,calc,
trees,mindmap,shadows,decorations.text,patterns,positioning,shapes,matrix,fit}

\makeatletter
\newcommand{\rmnum}[1]{\romannumeral #1}
\newcommand{\Rmnum}[1]{\expandafter\@slowromancap\romannumeral #1@}

\newif\if@borderstar
\def\bordermatrix{\@ifnextchar*{%
  \@borderstartrue\@bordermatrix@i}{\@borderstarfalse\@bordermatrix@i*}%
}
\def\@bordermatrix@i*{\@ifnextchar[{\@bordermatrix@ii}{\@bordermatrix@ii[()]}}
\def\@bordermatrix@ii[#1]#2{%
\begingroup
  \m@th\@tempdima8.75\p@\setbox\z@\vbox{%
    \def\cr{\crcr\noalign{\kern 2\p@\global\let\cr\endline }}%
    \ialign {$##$\hfil\kern 2\p@\kern\@tempdima & \thinspace %
    \hfil $##$\hfil && \quad\hfil $##$\hfil\crcr\omit\strut %
    \hfil\crcr\noalign{\kern -\baselineskip}#2\crcr\omit %
    \strut\cr}}%
  \setbox\tw@\vbox{\unvcopy\z@\global\setbox\@ne\lastbox}%
  \setbox\tw@\hbox{\unhbox\@ne\unskip\global\setbox\@ne\lastbox}%
  \setbox\tw@\hbox{%
    $\kern\wd\@ne\kern -\@tempdima\left\@firstoftwo#1%
    \if@borderstar\kern2pt\else\kern -\wd\@ne\fi%
    \global\setbox\@ne\vbox{\box\@ne\if@borderstar\else\kern 2\p@\fi}%
    \vcenter{\if@borderstar\else\kern -\ht\@ne\fi%
    \unvbox\z@\kern -\if@borderstar2\fi\baselineskip}%
    \if@borderstar\kern -2\@tempdima\kern2\p@\else\,\fi\right\@secondoftwo#1 $%
  }\null \;\vbox{\kern\ht\@ne\box\tw@}%
\endgroup
}
\makeatother

\newtheorem{theorem}{Theorem}
\newtheorem{lemma}{Lemma}
\newtheorem{cor}{Corollary}

\newtheorem{prop}{Proposition}
\newtheorem{example}{Example}
\newtheorem{definition}{Definition}
\newtheorem{remark}{Remark}
\allowdisplaybreaks

\newcommand{\ein}{\mathcal{E}_{\mathrm{i}}}
\newcommand{\eout}{\mathcal{E}_{\mathrm{o}}}
\newcommand{\tail}{\mathrm{tail}}
\newcommand{\head}{\mathrm{head}}
\newcommand{\mA}{\mathcal{A}}

\newcommand{\mN}{\mathcal{N}}
\newcommand{\mE}{\mathcal{E}}
\newcommand{\mO}{\mathcal{O}}
\newcommand{\mC}{\mathcal{C}}
\newcommand{\mG}{\mathcal{G}}
\newcommand{\mV}{\mathcal{V}}
\newcommand{\mbC}{\mathbf{C}}

\def\Im{\mathrm{Im}}
\tikzstyle{vertex}=[draw,circle,fill=gray!30,minimum size=6pt, inner sep=0pt]

\newcommand{\Enone}{\bf{En1}}
\newcommand{\En}{\bf{En}}
\newcommand{\Entwo}{\bf{En2}}
\newcommand{\De}{\bf{De}}

\tikzstyle{block} = [draw, fill=blue!30, rectangle,
    minimum height=2em, minimum width=3em]
\tikzstyle{sum} = [draw, fill=blue!20, circle, node distance=1cm]
\tikzstyle{input} = [coordinate]
\tikzstyle{output} = [coordinate]
\tikzstyle{pinstyle} = [pin edge={to-,thin,black}]

\begin{document}
	
\title{Zero-Error Distributed Function Compression for Binary Arithmetic Sum}

\author{Xuan~Guang and Ruze~Zhang
	%\thanks{X. Guang is with School of Mathematical Sciences and LPMC, Nankai University, Tianjin, China (e-mail: xguang@nankai.edu.cn); R. Zhang is with School of Mathematical Sciences, Nankai University, Tianjin, China.   This work was funded in part by the NSFC (Grant No. 61771259) and the Fundamental Research Funds for the Central Universities, Nankai University.}}% (Grant Nos. 63191502, 65121007).}
}

\maketitle

	\begin{abstract}
In this paper, we put forward the model of \emph{zero-error distributed function compression system} of two binary memoryless sources $X$ and $Y$ as depicted in Fig.\,\ref{fig:Gen-System_use}. In this model, there are two encoders $\Enone$ and $\Entwo$ and one decoder $\De$, connected by two channels $(\Enone,~\De)$ and $(\Entwo,~\De)$ with the capacity constraints $C_1$ and $C_2$, respectively. The encoder $\Enone$ can observe $X$ or $(X,Y)$ and the encoder $\Entwo$ can observe $Y$ or $(X,Y)$ according to the two switches $\boldsymbol{\mathrm{s}}_1$ and $\boldsymbol{\mathrm{s}}_2$ open or closed (corresponding to taking values $0$ or $1$). The decoder $\De$ is required to compress the binary arithmetic sum $f(X,Y)=X+Y$ with zero error by using the system multiple times. We use $(\boldsymbol{\mathrm{s}}_1\boldsymbol{\mathrm{s}}_2;C_{1}, C_{2}; f)$ to denote the model in which it is assumed that $C_1\geq C_2$ by symmetry. The compression capacity for the model is defined as the maximum average number of times that the function $f$ can be compressed with zero error for one use of the system, which measures the efficiency of using the system. We fully characterize the compression capacities for all the four cases of the model $(\boldsymbol{\mathrm{s}}_1\boldsymbol{\mathrm{s}}_2;C_{1}, C_{2}; f)$ for $\boldsymbol{\mathrm{s}}_1\boldsymbol{\mathrm{s}}_2=00,01,10,11$. Here, the characterization of the compression capacity for the case  $(01;C_1,C_2;f)$ with $C_1>C_2$ is highly nontrivial, where  a novel graph coloring approach is developed.
Furthermore, we apply the compression capacity for  $(01;C_1,C_2;f)$ to an open problem in network function computation that whether the best known upper bound of Guang~\emph{et al.} on computing capacity is in general tight. This upper bound is always tight for all previously considered network function computation problems whose computing capacities are known. By considering a network function computation model transformed from $(01;C_1,C_2;f)$ with $C_1>C_2$, we give the answer that in general the upper bound of Guang~\emph{et al.} is not tight.
\end{abstract}

\section{ Introduction}

In this paper, we consider the model of \emph{zero-error distributed function compression system} of two binary memoryless sources $X$ and $Y$ as depicted in Fig.\,\ref{fig:Gen-System_use}. In this model, there are two encoders $\Enone$ and $\Entwo$ and one decoder $\De$, connected by two channels $(\Enone,~\De)$ and $(\Entwo,~\De)$ with capacity constraints $C_1$ and $C_2$, respectively. Without loss of generality, we assume that $C_1\geq C_2$. The encoder $\Enone$ can observe the single source $X$ or the two sources $(X,Y)$, and the encoder $\Entwo$ can observe the single source $Y$ or the two sources $(X,Y)$. We use two switches $\boldsymbol{\mathrm{s}}_1$ and $\boldsymbol{\mathrm{s}}_2$ open or closed (corresponding to taking values $0$ or $1$) to represent whether $\Enone$ can observe $Y$ and $\Entwo$ can observe $X$, respectively. Then, the pair $\boldsymbol{\mathrm{s}}_1\boldsymbol{\mathrm{s}}_2$ can take the four values $00,01,10$ and $11$. See Figs.\,\ref{fig:System_use-00}--\ref{fig:System_use-11}. The decoder $\De$ is required to compress the binary arithmetic sum $f(X,Y)=X+Y$ with zero error. %  by using the system multiple times.
We use $(\boldsymbol{\mathrm{s}}_1\boldsymbol{\mathrm{s}}_2;C_{1}, C_{2};  f)$ to denote this model of zero-error distributed function compression system.

\begin{figure}[t]
 \centering
  {
  \tikzstyle{format}=[draw,circle,fill=gray!30,minimum size=2pt, inner sep=0pt]
\begin{tikzpicture}[auto, node distance=2cm,>=latex',thick]
    % We start by placing the blocks
    \node [input, name=inputX, label=left:$\boldsymbol{X}$] at (2.5,0) {};
    \node [input, name=inputY, label=left:$\boldsymbol{Y}$] at (2.5,-2.4) {};
    \node [block] (en1) at (5.5,0) {{\bf En1}};
    \node [block] (en2) at (5.5,-2.4){{\bf En2}};
    \node [block] (de)  at (9,-1.2) {{\bf De}};
    \node [output, name=output, label=above:$\boldsymbol{X}+\boldsymbol{Y}$] at (11.5,-1.2) {};
    % We draw an edge between the controller and system block to
    % calculate the coordinate u. We need it to place the measurement block.
    \node[format](a1)at(3.2,-0.56){};\node[format](a11)at(3.5,-0.80){};
    \node[format](a2)at(3.2,-1.84){}; \node[format](a21)at(3.5,-1.6){};
    \node at(3.1,-0.84) {\small{$\boldsymbol{\mathrm{s}}_1$}};
     \node at(3.1,-1.56) {\small{$\boldsymbol{\mathrm{s}}_2$}};
    \draw [->] (inputX) -- (en1);\draw [->] (inputY) -- (en2);
     \draw (inputX) -- (a1); \draw [->](a11) -- (en2);
   \draw (inputY) -- (a2); \draw [->](a21) -- (en1);
    \draw (a1) -- (3.55,-0.56);
       \draw (a2) -- (3.55,-1.84);
         \draw [->] (en1) -- node       {$C_1$} (de);
    \draw [->] (en2) -- node[swap] {$C_2$} (de);
    \draw [->] (de) -- (output);
\end{tikzpicture}
}
\caption{The model of distributed function compression $(\boldsymbol{\mathrm{s}}_1\boldsymbol{\mathrm{s}}_2;C_1,C_2;f)$.}\label{fig:Gen-System_use}
\end{figure}
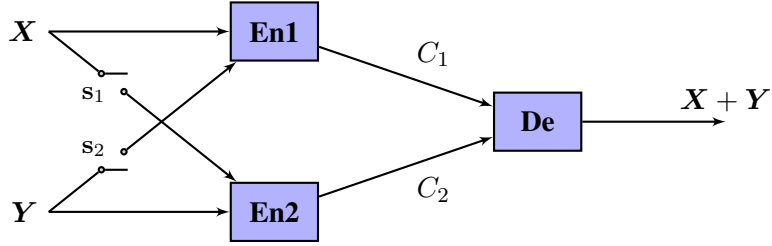
%In this paper, we first investigate the \emph{zero-error distribution (non-identity) function compression} of two binary sources. A general setup of the model   is presented as follows, including  two encoders $\Enone$ and $\Entwo$, one decoder $\De$ and  two channels $(\Enone,~\De)$ and $(\Entwo,~\De)$  with capacity constraints $C_1$ and $C_2$, respectively. In a distributed function compression system,  \rmnum{1}) the encoder $\Enone$ can observe $X$ or $(X,Y)$, and the encoder $\Entwo$ can observe $Y$ or $(X,Y)$, where we use two switches $\boldsymbol{\mathrm{s}}_1$ and $\boldsymbol{\mathrm{s}}_2$ open or closed (taking values $0$ or $1$) to represent whether $\Enone$ can observe $Y$ and $\Entwo$ can observe $X$, respectively; \rmnum{2}) the encoders $\Enone$ and $\Entwo$ can first compress the observed sources and then transmit them to the decoder $\De$ through the channels $(\Enone,~\De)$ and $(\Entwo,~\De)$, respectively;  \rmnum{3}) and the decoder $\De$ is required to compress repeatedly (with zero error) the arithmetic sum $f(X,Y)=X+Y$, called the {\em target function}, whose arguments are source messages generated at the two binary sources. We use  $(\boldsymbol{\mathrm{s}}_1\boldsymbol{\mathrm{s}}_2;C_{1}, C_{2};  f)$ to denote the model of distributed function compression  system.

\subsection{Related Works}

The line of research on data compression can be traced back to Shannon's source coding theorem in his celebrated work~\cite{Shannon48}, in which the compression of a source  over a point-to-point channel was discussed. Multi-terminal source coding was launched by Slepian and Wolf \cite{Slepian-Wolf-IT73}, where two correlated sources are compressed separately by two encoders; and the decoder, which receives the outputs of both  encoders, is required to reconstruct the two sources almost perfectly. Building on the Slepian-Wolf model, K\"{o}rner and Marton~\cite{Korner-Marton-IT73} investigated the computation of the modulo $2$ sum of two correlated binary sources.
%The coding rate region was established only for the case when the correlated binary sources are doubly symmetric.
To our knowledge, the K\"{o}rner-Marton problem was the first non-identity function compression problem.
Doshi \textit{et al}.~\cite{Doshi_fun_comp_graph_color_sch}
generalized the K{\"o}rner-Marton model by requiring the decoder to compress with asymptotically zero error an arbitrary function of two correlated sources. Feizi and M{\'e}dard~\cite{Feizi-Medard} further investigated the function compression over a tree network. In such a network, the leaf nodes are the source nodes where correlated sources are generated, and the root node is the single sink node where a
function of the sources is required to compress with asymptotically zero error. Orlitsy and Roche~\cite{Orlitsky-Roche_general_side_inf_model_rat_reg} considered an asymptotically-zero-error function compression problem with side information at the decoder, which is a generalization of Witsenhausen's model~\cite{Witsenhausen-IT-76} of zero-error source coding with side information at the decoder. Specifically, in the Orlitsy-Roche model, the encoder compresses a source $X$. The decoder, in addition to receiving the output of the encoder, also observes a source  $Y$ which is correlated with $X$. The decoder is required to compress an arbitrary function of $X$ and $Y$ with asymptotically zero error.

In the above mentioned compression models~\cite{Shannon48,Slepian-Wolf-IT73,Korner-Marton-IT73,Doshi_fun_comp_graph_color_sch,
Feizi-Medard,Orlitsky-Roche_general_side_inf_model_rat_reg}, the sources and functions are required to compress with asymptotically zero error. Another type of compression models is to consider compressing with zero error the sources and functions, e.g., \cite{Shannon48,Witsenhausen-IT-76,Alon-Orlitsky_source_cod_graph_entropies,Koulgi_zero-error-cod_cor_inf_sour}. In Shannon's seminal work \cite{Shannon48}, a variable-length source code, called {\em Shannon Code}, was given to compress a source with zero error. Witsenhausen \cite{Witsenhausen-IT-76} considered the  source coding problem with side information at the decoder, where the encoder compresses a source~$X$; and the decoder, which not only receives the output of the encoder but also observes another source  $Y$ correlated with $X$, is required to reconstruct $X$ with zero error. Subsequently, the zero-error compression problem building on Witsenhausen model was developed by Alon and Orlitsky \cite{Alon-Orlitsky_source_cod_graph_entropies}, where a variable-length code of the source $X$ was considered.
%The minimal asymptotical (per-instance) rate was  shown
%to be the limit of the (normalized) chromatic entropy of
%$(G_{X|Y}^{\land k} , X^{k})$, where $G_{X|Y}$ is the associated confusability
%graph, also known as Witsenhausen¡¯s characteristic graph \cite{Witsenhausen-IT-76}, and  $G_{X|Y}^{\land k}$ is   the $k$-fold AND product of the graph $G_{X|Y}$.
%A closed form expression for this quality is unknown.
%The paper \cite{Alon-Orlitsky_source_cod_graph_entropies} also considered
%an \emph{unrestricted-inputs} communication scenario, where the encoder also compresses a source  $X$ with a variable-length source code; and the decoder, which not only receives the output of the encoder but also observes any  $y$ in the support set of $Y$, is required to reconstruct $X$ with zero error only when $(X,y)$ is in the support set of $(X,Y)$.
%Accordingly, the minimal asymptotical (per-instance) rate was  shown
%to be the limit of the (normalized) chromatic entropy of
%$(G_{X|Y}^{\vee k} , X^{k})$, where $G_{X|Y}$ is also the associated confusability
%graph, and  $G_{X|Y}^{\land k}$ is   the $k$-fold OR product of the graph $G_{X|Y}$.
% The limit is exactly equal to
%the graph entropy $H(G_{X|Y}, X)$ introduced by K\"{o}rner~\cite{Korner-graph-entropy}.
%Koulgi \textit{et al}.~\cite[Section~\Rmnum{3}]{Koulgi_zero-error-cod_cor_inf_sour} and Shayevitz~\cite{Shayevitz-un-restr-distri-fun-comp} further studied the unrestricted-inputs communication scenario of distributed function computation.
Koulgi \textit{et al}.~\cite{Koulgi_zero-error-cod_cor_inf_sour} considered a zero-error version of the Slepian-Wolf model. In this model, two correlated sources are compressed separately by two encoders; and the decoder is required to reconstruct the two sources with zero error by applying the outputs of the two encoders.
% To our knowledge, there is no further research on  the zero-error compression of an arbitrary  function of two correlated sources.
The zero-error capacity problem of a discrete memoryless (stationary) channel was considered by Shannon in \cite{Shannon_zero_cap_noisy_channel}, where the Shannon zero-error capacity of a confusability graph was introduced. Subsequently, some upper bounds on the Shannon zero-error capacity were obtained in \cite{Shannon_zero_cap_noisy_channel, Rosenfeld_frac_color_num_bound_shannon_cap, Lovasz_shannon_cap_graph,Haemers_shannon_lovasz_problems}.
Another related line of research is to consider (zero-error) network function computation \cite{Appuswamy11,Appuswamy13,Ramamoorthy-Langberg-JSAC13-sum-networks,
Rai-Dey-TIT-2012,Tripathy-Ramamoorthy-IT18-sum-networks,
Huang_Comment_cut_set_bound,Guang_Improved_upper_bound,
Wang-Tightness_upper_bound}, where the following model was considered. In a directed acyclic
network, the single sink node is required to compute with zero error a function of the source messages that are separately observed by multiple source nodes.

\subsection{Contributions and Organization of the Paper}

%For the model of distributed function compression  system $(\boldsymbol{\mathrm{s}}_1\boldsymbol{\mathrm{s}}_2;C_{1}, C_{2};  f)$,  we consider in this paper the special case that the joint probability of the binary sources $X$ and $Y$, denoted by $P_{XY}$, satisfies $P_{XY}(x,y)>0,~ \forall~(x,y)\in\{0,1\}\times\{0,1\},$ and  the target function $f$ is the arithmetic sum. We note that the general case is that the binary sources $X$ and $Y$ are arbitrary, and the target function $f$ is an arbitrary function. This will be treated in  the next paper.

From the information theoretic point of view, we are interested in characterizing the \emph{compression capacity} for the model $(\boldsymbol{\mathrm{s}}_1\boldsymbol{\mathrm{s}}_2;C_{1}, C_{2};  f)$. Here, the compression capacity is defined as the maximum average number of times that the function $f$ can be compressed with zero error for one use of the system, which measures the efficiency of using the system, analogous to Shannon zero-error capacity \cite{Shannon_zero_cap_noisy_channel,Rosenfeld_frac_color_num_bound_shannon_cap,Lovasz_shannon_cap_graph,
Haemers_shannon_lovasz_problems}, network coding \cite{Ahlswede_net_flow,Li-Yeung-Cai-2003,Koetter-Medard-algebraic,co-construction,Yeung-book,Zhang-book,Fragouli-book,Fragouli-book-app,Ho-book}
and network function computation \cite{Appuswamy11,Appuswamy13,Ramamoorthy-Langberg-JSAC13-sum-networks,Rai-Dey-TIT-2012,
Tripathy-Ramamoorthy-IT18-sum-networks,Huang_Comment_cut_set_bound,Guang_Improved_upper_bound}. In the current paper, we focus on this compression capacity, rather than  the notion of compression capacity considered in many previously studied source coding models in which how to efficiently establish a system is investigated, e.g., lossless source coding models \cite{Shannon48,Slepian-Wolf-IT73,Korner-Marton-IT73,Orlitsky-Roche_general_side_inf_model_rat_reg,
Doshi_fun_comp_graph_color_sch,Feizi-Medard}, zero-error source coding models \cite{Witsenhausen-IT-76,Alon-Orlitsky_source_cod_graph_entropies,Koulgi_zero-error-cod_cor_inf_sour}, and lossy source coding models \cite{Wyner-Ziv_rat_distortion_fun_sid_inf,Yamamoto_rat_distortion_gener_fun_sid_inf,Feng_rat_distortion_net_fun_sid_inf,
Berger-Yeung_multi_source_cod_dist_cri,Barros_cor_sour_rat_dist_reg,Wagner_Rat_reg_Guassian_sou_cod}.

In this paper, we fully characterize the compression capacities for all the four cases of the model $(\boldsymbol{\mathrm{s}}_1\boldsymbol{\mathrm{s}}_2;C_{1}, C_{2};  f)$.  First, we characterize the compression capacities for the three cases that $\boldsymbol{\mathrm{s}}_1\boldsymbol{\mathrm{s}}_2=00,10$, and $11$ (Figs.\,\ref{fig:System_use-00}, \ref{fig:System_use-10} and \ref{fig:System_use-11}). For the case of the model $(00;C_1,C_2;f)$, compressing the binary arithmetic sum $f$ is equivalent to compressing the two sources, or equivalently, the binary identity function. For the case of the model $(11;C_1,C_2;f)$, compressing the binary arithmetic sum $f$ is equivalent to compressing with zero error a ternary source on one channel with the capacity constraint $C_1+C_2$. For the remaining case $(01;C_1,C_2;f)$ in Fig.\,\ref{fig:System_use}, the characterization of the compression capacity is very difficult. We first focus on a typical special case $(01;2,1;f)$ (i.e., $C_1=2$ and $C_2=1$) to explicitly show our technique by developing a novel graph coloring approach, and the proof is highly nontrivial.
In particular, we prove a crucial function, called \emph{aitch-function}, that plays a key role in estimating the minimum chromatic number of a conflict graph introduced in this paper. This aitch-function is analogous to the \emph{theta function} introduced by Lov\'{a}sz~\cite{Lovasz_shannon_cap_graph}, which is used to estimate the  Shannon zero-error capacity for a confusability graph. With the developed technique, we further characterize the compression capacity for the general case $(01;C_1,C_2;f)$.

An important application of the compression capacity for the model $(01;C_1,C_2;f)$ is in the tightness of the best known upper bound on the computing capacity in  network function computation~\cite{Guang_Improved_upper_bound}. In network function computation, several ``general'' upper bounds on the computing capacity have been obtained \cite{Appuswamy11,Huang_Comment_cut_set_bound,Guang_Improved_upper_bound}, where ``general'' means that the upper bounds are applicable to arbitrary network and arbitrary  function. Here, the best known upper bound is the one proved by Guang~\emph{et al.}~\cite{Guang_Improved_upper_bound} in using the approach of the cut-set strong partition. This bound is always tight for all previously considered network function computation problems whose computing capacities are known. However, whether the upper bound is in general tight remains open~\cite{Guang_Improved_upper_bound}. By transforming the model $(01;2,1;f)$ into an equivalent model of network function computation, we give the answer that in general the upper bound of Guang~\emph{et al.} is not tight. We further prove that the upper bound of Guang~\emph{et al.}~\cite{Guang_Improved_upper_bound} is not tight for any model of network function computation that is transformed from the model $(01;C_1,C_2;f)$ for any pair of the channel capacity constraints $(C_1,C_2)$ with $C_1>C_2$.

%It should be noted that, following  \cite{Guang_Improved_upper_bound}, Wang \emph{et al.} \cite{Wang-Tightness_upper_bound} had showed that the upper bound of Guang~\emph{et al.} \cite{Guang_Improved_upper_bound} is not achievable by
%using an explicit  network function computation
%model.
%However, the proof therein is invalid.
%Specifically, in \cite[Lemma~1]{Wang-Tightness_upper_bound}, it was claimed that $\gamma_{3^{k}-1}\geq 4^{k}-2^{k-1}$ for any positive integer $k>1$, which is incorrect. This thus causes their upper bound to be invalid. In fact, it can be proved that $\gamma_{3^{k}-1}=3\cdot 2^{k-1}<4^{k}-2^{k-1}$, $\forall~k>1$. Following the same argument with this correct equation, we can obtain the valid upper bound $1$, the same as the best known upper bound in~\cite{Guang_Improved_upper_bound}.
%We analyze the arguments
%	 then fix the issues to obtain the same valid upper bound   as the upper bound of Guang~\emph{et al.}~\cite{Guang_Improved_upper_bound}. This also implies that the technique proposed by \cite{Wang-Tightness_upper_bound} does not improve the upper bound of Guang~\emph{et al.}~\cite{Guang_Improved_upper_bound}.

The paper is organized as follows. In Section~\ref{sec:dis-fun-model}, we formally present the model of distributed function compression system. We characterize the compression capacities for the three models $(00;C_1,C_2;f)$, $(10;C_1,C_2;f)$ and $(11;C_1,C_2;f)$  in Section~\ref{sec:compre-cap00,10,11}. In Section~\ref{sec:compre-cap01}, we first characterize the compression capacity for the typical case $(01;2,1;f)$ of the model $(01;C_1,C_2;f)$. With this, we explicitly characterize the compression capacity for $(01;C_1,C_2;f)$. Section~\ref{sec:equ-model} is devoted to the application of this distributed function compression problem on the model $(01;C_1,C_2;f)$ in network function computation. In Section~\ref{sec:concl}, we conclude with a summary of our results and a remark on future research.

\section{Model of Distributed Function Compression}\label{sec:dis-fun-model}

Let $X$ and $Y$ be two binary random  variables (r.v.s) according to the joint probability distribution $P_{XY}$
such that
\begin{equation}\label{PXY}
P_{XY}(x,y)>0,\quad \forall~(x,y)\in\mathcal{A}\times\mathcal{A},
\end{equation}
where we let $\mathcal{A}\triangleq\{0,1\}$ in the paper.
Further, we let $P_{X}$ and $P_{Y}$ be the marginal distributions of $X$ and $Y$ referring to $P_{XY}$,
respectively. By~\eqref{PXY}, we immediately obtain that
\begin{equation*}
P_{X}(0)\cdot P_{X}(1)\cdot P_{Y}(0)\cdot P_{Y}(1)>0,
 \end{equation*}
 i.e., neither $X$ nor $Y$ is a constant.
Consider a positive integer $k$. Let
\begin{equation*}
\boldsymbol{X} \triangleq (X_{1}, X_{2}, \cdots, X_{k})\quad \text{and} \quad\boldsymbol{Y}\triangleq (Y_{1},~Y_{2},~\cdots,\, Y_{k})
\end{equation*}
be two length-$k$ sequential vectors of i.i.d. random variables with generic r.v.s $X$ and $Y$, respectively.
Let
\begin{equation*}
\vec{x}=(x_{1},x_{2},\dots, x_{k})\in\mathcal{A}^k\quad \text{and}\quad\vec{y}=(y_{1},y_{2},\dots,y_{k})\in\mathcal{A}^k
\end{equation*}
be two arbitrary outputs of $\boldsymbol{X}$ and $\boldsymbol{Y}$, respectively,  called the \emph{source messages}.

 Let the  \emph{target function} $f:~\mathcal{A}\times\mathcal{A}\to\Im\,f=\{0,1,2\}$  be the  \emph{binary arithmetic sum}, i.e.,
 \begin{equation*}
f(x,y)=x+y,\quad \forall~x\in\mathcal{A}\quad\text{and}\quad y\in\mathcal{A},
\end{equation*}
where we use $\Im\,f$ to stand for the image set of the function $f$.
Further, we let
\begin{equation*}
f(\boldsymbol{X},\boldsymbol{Y})\triangleq\big(f(X_1,Y_1),f(X_2,Y_2),\cdots,f(X_k,Y_k)\big).
\end{equation*}
Then, the $k$ function values of  $(\vec{x},\vec{y})$ are written as
\begin{equation*}
f(\vec{x},\vec{y})\triangleq\big(f(x_{1},y_{1}),f(x_{2},y_{2}),\dots ,f(x_{k},y_{k})\big).
\end{equation*}
%{\color{red}We note that the general case is that the support set $S_{XY}\subseteq\mathcal{A}\times\mathcal{A}$ is an arbitrary subset and the target function $f$ is an arbitrary function. This case will be treated in the next paper.}

 Now, we consider a  \emph{distributed function compression system} with two encoders $\Enone$ and $\Entwo$, one decoder $\De$,  and two channels $(\Enone,~\De)$ and $(\Entwo,~\De)$  with capacity constraints $C_1$ and $C_2$, respectively, i.e., for each use of the system, at most $C_1$ bits can be reliably transmitted on the channel $(\text{$\Enone$},~\text{$\De$})$ and at most $C_2$ bits can be reliably transmitted on the channel $(\text{$\Entwo$},~\text{$\De$})$.
  By symmetry, we  assume  that $C_1\geq C_2$.
  The encoder  $\Enone$ can observe the single source $\boldsymbol{X}$ or all the two sources $(\boldsymbol{X},\boldsymbol{Y})$, and also, the encoder  $\Entwo$ can observe the single source $\boldsymbol{Y}$ or all the two sources $(\boldsymbol{X},\boldsymbol{Y})$.
 To show whether $\Enone$ can observe $\boldsymbol{Y}$ and $\Entwo$ can observe $\boldsymbol{X}$, two switches  $\boldsymbol{\mathrm{s}}_1$ and $\boldsymbol{\mathrm{s}}_2$ are added in the system as depicted in Fig.\,\ref{fig:Gen-System_use}. For $i=1,2$,  the switch $\boldsymbol{\mathrm{s}}_i$ takes value $0$ or $1$ to represent  $\boldsymbol{\mathrm{s}}_i$ open or closed. Then, the pair $\boldsymbol{\mathrm{s}}_1\boldsymbol{\mathrm{s}}_2$ can take four values $00,01,10$ and $11$. To be specific,

\begin{figure}[t]
 \centering
  {
\begin{tikzpicture}[auto, node distance=2cm,>=latex',thick]
    % We start by placing the blocks
    \node [input, name=inputX, label=left:$\boldsymbol{X}$] at (2.5,0) {};
    \node [input, name=inputY, label=left:$\boldsymbol{Y}$] at (2.5,-2.4) {};
    \node [block] (en1) at (5.5,0) {{\bf En1}};
    \node [block] (en2) at (5.5,-2.4){{\bf En2}};
    \node [block] (de)  at (9,-1.2) {{\bf De}};
    \node [output, name=output, label=above:$\boldsymbol{X}+\boldsymbol{Y}$] at (11.5,-1.2) {};
    % We draw an edge between the controller and system block to
    % calculate the coordinate u. We need it to place the measurement block.
    \draw [->] (inputX) -- (en1);
    %\draw [->] (inputY) -- (en2);
    \draw [->] (inputY) -- (en2);
    \draw [->] (en1) -- node       {$C_1$} (de);
    \draw [->] (en2) -- node[swap] {$C_2$} (de);
    \draw [->] (de) -- (output);
\end{tikzpicture}
}
\caption{The distributed function compression model $(00;C_1,C_2;f)$.}\label{fig:System_use-00}
\end{figure}
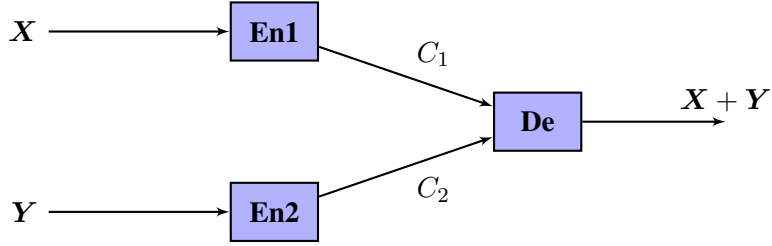

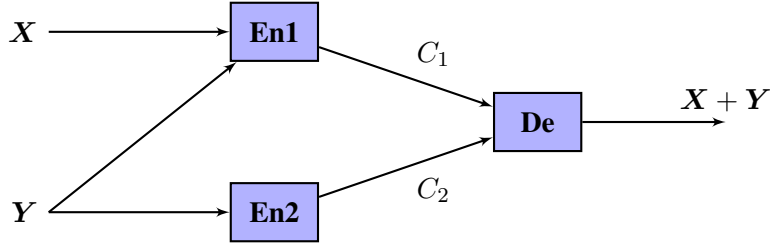
\begin{figure}[t]
 \centering
  {
\begin{tikzpicture}[auto, node distance=2cm,>=latex',thick]
    % We start by placing the blocks
    \node [input, name=inputX, label=left:$\boldsymbol{X}$] at (2.5,0) {};
    \node [input, name=inputY, label=left:$\boldsymbol{Y}$] at (2.5,-2.4) {};
    \node [block] (en1) at (5.5,0) {{\bf En1}};
    \node [block] (en2) at (5.5,-2.4){{\bf En2}};
    \node [block] (de)  at (9,-1.2) {{\bf De}};
    \node [output, name=output, label=above:$\boldsymbol{X}+\boldsymbol{Y}$] at (11.5,-1.2) {};
    % We draw an edge between the controller and system block to
    % calculate the coordinate u. We need it to place the measurement block.
    \draw [->] (inputX) -- (en1);
    \draw [->] (inputY) -- (en1);
    %\draw [->] (inputY) -- (en2);
    \draw [->] (inputY) -- (en2);
    \draw [->] (en1) -- node       {$C_1$} (de);
    \draw [->] (en2) -- node[swap] {$C_2$} (de);
    \draw [->] (de) -- (output);
\end{tikzpicture}
}
\caption{The distributed function compression model $(01;C_1,C_2;f)$.}\label{fig:System_use}
\end{figure}

\begin{figure}[t]
 \centering
  {
\begin{tikzpicture}[auto, node distance=2cm,>=latex',thick]
    % We start by placing the blocks
    \node [input, name=inputX, label=left:$\boldsymbol{X}$] at (2.5,0) {};
    \node [input, name=inputY, label=left:$\boldsymbol{Y}$] at (2.5,-2.4) {};
    \node [block] (en1) at (5.5,0) {{\bf En1}};
    \node [block] (en2) at (5.5,-2.4){{\bf En2}};
    \node [block] (de)  at (9,-1.2) {{\bf De}};
    \node [output, name=output, label=above:$\boldsymbol{X}+\boldsymbol{Y}$] at (11.5,-1.2) {};
    % We draw an edge between the controller and system block to
    % calculate the coordinate u. We need it to place the measurement block.
    \draw [->] (inputX) -- (en1);
    \draw [->] (inputX) -- (en2);
    %\draw [->] (inputY) -- (en2);
    \draw [->] (inputY) -- (en2);
    \draw [->] (en1) -- node       {$C_1$} (de);
    \draw [->] (en2) -- node[swap] {$C_2$} (de);
    \draw [->] (de) -- (output);
\end{tikzpicture}
}
\caption{The distributed function compression model $(10;C_1,C_2;f)$.}\label{fig:System_use-10}
\end{figure}
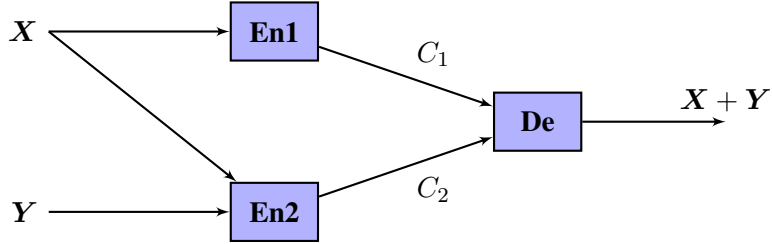
\begin{figure}[t]
 \centering
  {
\begin{tikzpicture}[auto, node distance=2cm,>=latex',thick]
    % We start by placing the blocks
    \node [input, name=inputX, label=left:$\boldsymbol{X}$] at (2.5,0) {};
    \node [input, name=inputY, label=left:$\boldsymbol{Y}$] at (2.5,-2.4) {};
    \node [block] (en1) at (5.5,0) {{\bf En1}};
    \node [block] (en2) at (5.5,-2.4){{\bf En2}};
    \node [block] (de)  at (9,-1.2) {{\bf De}};
    \node [output, name=output, label=above:$\boldsymbol{X}+\boldsymbol{Y}$] at (11.5,-1.2) {};
    % We draw an edge between the controller and system block to
    % calculate the coordinate u. We need it to place the measurement block.
    \draw [->] (inputX) -- (en1);
    \draw [->] (inputY) -- (en1);
    \draw [->] (inputX) -- (en2);
    \draw [->] (inputY) -- (en2);
    \draw [->] (en1) -- node       {$C_1$} (de);
    \draw [->] (en2) -- node[swap] {$C_2$} (de);
    \draw [->] (de) -- (output);
\end{tikzpicture}
}
\caption{The distributed function compression model $(11;C_1,C_2;f)$.}\label{fig:System_use-11}
\end{figure}
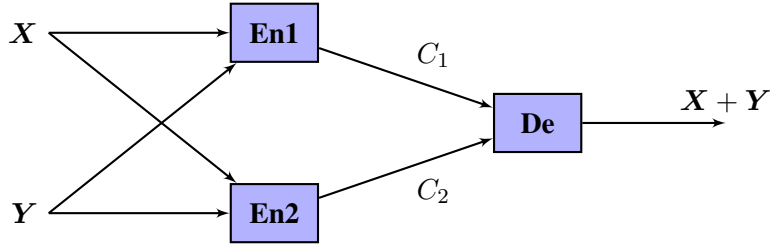
 \begin{itemize}
\item Case $1$: $\boldsymbol{\mathrm{s}}_1\boldsymbol{\mathrm{s}}_2=00$, i.e., $\boldsymbol{\mathrm{s}}_1$ and $\boldsymbol{\mathrm{s}}_2$ are open. The encoder $\Enone$ only  observes the source $\boldsymbol{X}$  and the encoder $\Entwo$ only observes  the source $\boldsymbol{Y}$ (see Fig.\,\ref{fig:System_use-00});

\item Case $2$: $\boldsymbol{\mathrm{s}}_1\boldsymbol{\mathrm{s}}_2=01$, i.e.,
 $\boldsymbol{\mathrm{s}}_1$ is open and $\boldsymbol{\mathrm{s}}_2$ is closed.
 The encoder $\Enone$  observes the two sources $\boldsymbol{X}$ and $\boldsymbol{Y}$, and the encoder $\Entwo$ only observes the source  $\boldsymbol{Y}$ (see Fig.\,\ref{fig:System_use});

\item Case $3$: $\boldsymbol{\mathrm{s}}_1\boldsymbol{\mathrm{s}}_2=10$, i.e., $\boldsymbol{\mathrm{s}}_1$ is closed and $\boldsymbol{\mathrm{s}}_2$ is open.  The encoder $\Enone$ only  observes the source $\boldsymbol{X}$, and the encoder $\Entwo$ observes the two sources $\boldsymbol{X}$ and $\boldsymbol{Y}$ (see Fig.\,\ref{fig:System_use-10});

\item Case $4$: $\boldsymbol{\mathrm{s}}_1\boldsymbol{\mathrm{s}}_2=11$, i.e., $\boldsymbol{\mathrm{s}}_1$ and $\boldsymbol{\mathrm{s}}_2$ are closed. Both the encoders $\Enone$ and  $\Entwo$ observe the two sources $\boldsymbol{X}$ and $\boldsymbol{Y}$ (see Fig.\,\ref{fig:System_use-11}).
\end{itemize}
To unify the above four cases, we say that the encoder $\Enone$ observes $\boldsymbol{X}$ and $\boldsymbol{Y}^{\boldsymbol{\mathrm{s}}_2}$, and the encoder  $\Entwo$ observes $\boldsymbol{X}^{\boldsymbol{\mathrm{s}}_1}$ and $\boldsymbol{Y}$ for the pair of switches $\boldsymbol{\mathrm{s}}_1\boldsymbol{\mathrm{s}}_2$, where $\boldsymbol{X}^{0}$ and $\boldsymbol{Y}^{0}$ are an empty vector if $\boldsymbol{\mathrm{s}}_1=\boldsymbol{\mathrm{s}}_2=0$, and $\boldsymbol{X}^{1}=\boldsymbol{X}$ if $\boldsymbol{\mathrm{s}}_1=1$ and $\boldsymbol{Y}^{1}=\boldsymbol{Y}$ if $\boldsymbol{\mathrm{s}}_2=1$.
The decoder $\De$ is required to compress the function $f(\boldsymbol{X},\boldsymbol{Y})$ with zero error by using the system multiple times.
 We have completed the specification of our model of distributed function compression  system,  denoted by $(\boldsymbol{\mathrm{s}}_1\boldsymbol{\mathrm{s}}_2;C_{1}, C_{2};  f)$.

A \emph{$k$-shot (function-compression) code} $\mbC$ for the model $(\boldsymbol{\mathrm{s}}_1\boldsymbol{\mathrm{s}}_2;C_{1}, C_{2};  f)$  consists of
\begin{itemize}
\item two \emph{encoding functions} $\phi_{1}$ and $\phi_{2}$ for the sources  at the encoders $\Enone$ and $\Entwo$, respectively, where
\begin{equation}
\begin{cases}
 \phi_{1}:~ \mathcal{A}^{k}\times\mathcal{A}^{\boldsymbol{\mathrm{s}}_2k}\rightarrow \Im\,\phi_1; \\
 \phi_{2}:~ \mathcal{A}^{\boldsymbol{\mathrm{s}}_1k}\times\mathcal{A}^{k}\rightarrow \Im\,\phi_2;
 \end{cases}\label{dis-fun-model-encod-fun-def}
\end{equation}
with $\Im\,\phi_i$ being the image set of $\phi_i$ and $\mathcal{A}^{\boldsymbol{\mathrm{s}}_{i}k}=\mathcal{A}^{0}$ being an empty set if $\boldsymbol{\mathrm{s}}_i=0$, $i=1,2$;
\item
a \emph{decoding function}
\begin{equation}\nonumber
\psi:~\Im\,\phi_1\times\Im\,\phi_2\rightarrow \Im f^{\,k}=\{0,1,2\}^{k},
\end{equation}
  which is used to compress the target function~$f$ with zero error at the decoder $\De$.
    \end{itemize}

    For the model $(\boldsymbol{\mathrm{s}}_1\boldsymbol{\mathrm{s}}_2;C_{1}, C_{2};  f)$, a $k$-shot code $\mbC=\big\{\phi_1,\phi_2;~\psi\big\}$ is said to be \emph{admissible} if for each pair of source messages $(\vec{x},\vec{y})\in\mathcal{A}^k\times\mathcal{A}^k$, the $k$ function values $f(\vec{x},\vec{y})$ can be compressed with zero error at the decoder, i.e.,
    \begin{equation*}
    \psi\big(\phi_1(\vec{x},\vec{y}^{\,\boldsymbol{\mathrm{s}}_2}),\phi_2(\vec{x}^{\,\boldsymbol{\mathrm{s}}_1},\vec{y})\big)=f(\vec{x},\vec{y}),
    \quad\forall~(\vec{x},\vec{y})\in\mathcal{A}^k\times\mathcal{A}^k,
    \end{equation*}
    where, similarly,
    \begin{align*}
    \vec{x}^{\,\boldsymbol{\mathrm{s}}_1}=
    \begin{cases}
    \vec{x},\quad &\text{if}~\boldsymbol{\mathrm{s}}_1=1;\\
    \text{an empty vector},&\text{if}~\boldsymbol{\mathrm{s}}_1=0;\\
    \end{cases}
    \quad\text{and}\quad
    \vec{y}^{\,\boldsymbol{\mathrm{s}}_2}=
    \begin{cases}
    \vec{y},\quad &\text{if}~\boldsymbol{\mathrm{s}}_2=1;\\
    \text{an empty vector},&\text{if}~\boldsymbol{\mathrm{s}}_2=0.\\
    \end{cases}
    \end{align*}
For such an admissible $k$-shot code $\mbC=\big\{\phi_1,\phi_2;~\psi\big\}$, we let
    \begin{equation}
    n_i(\mbC)=\bigg\lceil\frac{\log |\Im\,\phi_{i}|}{C_i}\bigg\rceil,\quad i=1,2,\footnote{In this paper, we always use ``$\log$'' to denote ``$\log_2$'', the logarithm with base $2$, for notational simplicity. Further, we use $\left\lceil \cdot \right\rceil$ to stand for the smallest positive integer greater than or equal to ``$\cdot$''.}\label{dis-cod-def-n_i(C)}
    \end{equation}
    which is   the number of times that the channel $(\text{{\bf{En$\boldsymbol{i}$}}},~\text{$\De$})$ is used by using the code $\mbC$.
    We further let
\begin{equation}
n(\mbC)= \max\big\{n_1(\mbC),n_2(\mbC)\big\},\label{dis-cod-def-n(C)}
 \end{equation}
 which can be regarded as the number of times that the system is used to compress the function $f$ $k$ times with zero error in using the code $\mbC$. When the code $\mbC$ is clear from the context, we write $n$ to replace $n(\mbC)$ for notational simplicity. To measure the performance of codes, we define the \emph{(compression) rate} of the $k$-shot code $\mbC$ by
     \begin{equation*}
     R(\mbC)=\frac{H\big(f(\boldsymbol{X},\boldsymbol{Y})\big)}{n}=\frac{k }{n}\cdot H\big(f(X,Y)\big),
     \end{equation*}
   which is the average information amount of the target function $f$ can be compressed with zero error for one use
of the system.
Furthermore, with the target function $f$ and the joint probability distribution $P_{XY}$,  $H\big(f(X,Y)\big)$ is constant, and thus
 it suffices to consider
 \begin{equation*}
 R(\mbC)=\frac{k}{n}
  \end{equation*}
  as the rate of the code $\mbC$, which is the average number of times the function $f$ can be compressed with zero error for one use of the
system.  Further, we say that a nonnegative real number $R$ is {\em (asymptotically) achievable} if $\forall~\epsilon > 0$, there exists an admissible $k$-shot  code $\mbC$ such that
$$ R(\mbC) > R-\epsilon.$$
 Accordingly, the {\em compression rate region} for the  model  $(\boldsymbol{\mathrm{s}}_1\boldsymbol{\mathrm{s}}_2;C_{1}, C_{2};  f)$ is defined as
 \begin{equation*}
   \mathfrak{R}(\boldsymbol{\mathrm{s}}_1\boldsymbol{\mathrm{s}}_2;C_{1}, C_{2};  f)\triangleq\Big\{R:~\text{$R$ is  achievable for the model $(\boldsymbol{\mathrm{s}}_1\boldsymbol{\mathrm{s}}_2;C_{1}, C_{2};  f)$}\Big\},
   \end{equation*}
which is evidently closed and bounded. Consequently, the {\em compression capacity} for $(\boldsymbol{\mathrm{s}}_1\boldsymbol{\mathrm{s}}_2;C_{1}, C_{2};  f)$ is defined as
\begin{equation*}\label{defi_compre-capacity}
\mC(\boldsymbol{\mathrm{s}}_1\boldsymbol{\mathrm{s}}_2;C_{1}, C_{2};  f)\triangleq  \max~\mathfrak{R}(\boldsymbol{\mathrm{s}}_1\boldsymbol{\mathrm{s}}_2;C_{1}, C_{2};  f).
\end{equation*}

\section{Characterization of the Compression Capacities: $\boldsymbol{\mathrm{s}}_1\boldsymbol{\mathrm{s}}_2=00,10$, and $11$}\label{sec:compre-cap00,10,11}

In this section, we will fully characterize the compression capacities for the three cases of the model
$(00;C_{1}, C_{2};  f)$, $(10;C_{1}, C_{2};  f)$ and $(11;C_{1}, C_{2};  f)$.

\begin{theorem}\label{thm:00-comp-cap}
Consider the model of distributed function compression system $(00;C_{1}, C_{2};  f)$ as depicted in Fig.\,\ref{fig:System_use-00}, where $C_1\geq C_2$. Then,
\begin{equation*}
\mathcal{C}(00;C_1,C_2;f)=C_2.
\end{equation*}
\end{theorem}

  \begin{proof}
 Let $k$ be an arbitrary positive integer. For  any admissible $k$-shot code  $\mbC=\big\{\phi_{1},\phi_{2};~\psi\big\}$ for the model $(00;C_1,C_2;f)$, we have
 \begin{equation*}
\phi_{1}(\vec{x})\neq\phi_{1}(\vec{x}\,')\label{eq:case00,converse-pf-0}
\end{equation*}
for any two different source messages $\vec{x}$ and $\vec{x}\,'$ in $\mathcal{A}^k$, because, otherwise,
\begin{equation*}
\big(\phi_1(\vec{x}),\phi_2(\vec{y})\big)=\big(\phi_1(\vec{x}\,'),\phi_2(\vec{y})\big),
\quad\forall~\vec{y}\in\mathcal{A}^k,
\end{equation*}
which immediately implies that
\begin{equation*}
f(\vec{x},\vec{y})=\psi\big(\phi_1(\vec{x}),\phi_2(\vec{y})\big)=\psi\big(\phi_1(\vec{x}\,'),\phi_2(\vec{y})\big)=
f(\vec{x}\,',\vec{y}),
\end{equation*}
a contradiction.
Thus, we obtain that
$|\Im\,\phi_{1}|=|\mathcal{A}^k|=2^k.$
Similarly, we also obtain that  $|\Im\,\phi_{2}|=|\mathcal{A}^k|=2^k.$
It follows from \eqref{dis-cod-def-n_i(C)} and \eqref{dis-cod-def-n(C)} that
 \begin{align}
 n&=\max\bigg\{\bigg\lceil\frac{\log|\Im\,\phi_{1}|}{C_1}\bigg\rceil,\bigg\lceil\frac{\log|\Im\,\phi_{2}|}{C_2}
\bigg\rceil\bigg\}\nonumber\\&=\max\bigg\{\bigg\lceil\frac{k}{C_1}\bigg\rceil,\bigg\lceil\frac{k}{C_2}
\bigg\rceil\bigg\}
=\bigg\lceil\frac{k}{C_2}
\bigg\rceil\geq\frac{k}{C_2},\label{eq:case00,converse-pf-11}
\end{align}
where the second equality in \eqref{eq:case00,converse-pf-11} follows from $C_1\geq C_2$. Hence, we have proved that for each admissible $k$-shot code $\mbC$, the compression rate is upper bounded by
\begin{equation}
R(\mbC)=\frac{k}{n}\leq C_2.\label{eq:case00,converse-pf-3}
\end{equation}
We note that the upper bound \eqref{eq:case00,converse-pf-3} on the compression rate is true for each positive integer $k$. This immediately implies that
  \begin{equation}
  \mathcal{C}(00;C_1,C_2;f)\leq C_2.\label{00-mod-upp-bound}
  \end{equation}

  On the other hand, we will prove
  \begin{equation}
  \mathcal{C}(00;C_1,C_2;f)\geq C_2.\label{dir-pf-thm:00-comp-cap-1}
  \end{equation}
For any positive integer $k$, we consider the following $k$-shot code $\mbC=\big\{\phi_1,\phi_2;~\psi\big\}$:
 \begin{itemize}
 \item the two encoding functions $\phi_1$ and $\phi_2$ are the identity mapping from $\mathcal{A}^k$ to $\mathcal{A}^k$, i.e.,
\begin{equation*}
\phi_{1}(\vec{x})=\vec{x},\quad\forall~ \vec{x}\in\mathcal{A}^k\quad\text{and}\quad
\phi_{2}(\vec{y})=\vec{y},\quad\forall ~\vec{y}\in\mathcal{A}^k;
\end{equation*}
\item
the decoding function $\psi$ is given as
\begin{equation*}
\psi\big(\phi_{1}(\vec{x}),\phi_{2}(\vec{y})\big)=\psi(\vec{x},\vec{y})=\vec{x}+\vec{y},
\quad\forall~\vec{x}\in\mathcal{A}^k~\text{and}~\vec{y}\in\mathcal{A}^k.
\end{equation*}
\end{itemize}
We readily see that the code $\mbC$  is admissible and
\begin{equation*}
n=\max\bigg\{\bigg\lceil\frac{k}{C_1}\bigg\rceil,\bigg\lceil\frac{k}{C_2}\bigg\rceil\bigg\}
=\bigg\lceil\frac{k}{C_2}\bigg\rceil.
\end{equation*}
Then, the compression rate of the code $\mbC$ is
\begin{equation*}
R(\mbC)=\frac{k}{n}=\dfrac{k}{~\big\lceil k/C_2\big\rceil~}.
\end{equation*}
We note that $\frac{k}{\lceil k/C_2\rceil}$ tends to $C_2$ as $k$ goes to infinity. This immediately implies \eqref{dir-pf-thm:00-comp-cap-1}. Combining \eqref{00-mod-upp-bound} and \eqref{dir-pf-thm:00-comp-cap-1}, the theorem is thus proved.
\end{proof}

In the model $(00;C_1,C_2;f)$, compressing the binary arithmetic sum $f$ is equivalent to compressing the identity function $f_{\rm{\text{id}}}(x,y)=(x,y)$, $\forall~(x,y)\in \mathcal{A} \times \mathcal{A}$. We thus obtain the compression capacity $\mathcal{C}(00;C_1,C_2;f_{\rm{\text{id}}})=C_2$ as stated in the following corollary.

\begin{cor}\label{cor1}
 Consider the model of distributed function compression system $(00;C_{1}, C_{2}; f_{\emph{\rm{\text{id}}}})$, where the target function $f_{\emph{\rm{\text{id}}}}$ is the identity function and $C_1\geq C_2$. Then,
 \begin{equation*}
 \mathcal{C}(00;C_1,C_2;f_{\emph{\rm{\text{id}}}})=C_2.
 \end{equation*}
\end{cor}

Next, we consider a compression problem over Witsenhausen's model as depicted in Fig.\,\ref{fig:Side-infor-model-00}, where the encoder $\En$ compresses the binary source $\boldsymbol{Y}$; and the decoder, which  receives the output of the encoder and also observes another binary source  $\boldsymbol{X}$ as side information, is required to reconstruct $\boldsymbol{Y}$ with zero error; and the channel $(\En,~\De)$ has the capacity constraint
$C$, i.e., for each use of the channel $(\En,~\De)$, at most $C$ bits can be reliably transmitted.
This compression problem can be regarded as a dual problem of the original compression problem introduced by   Witsenhausen~\cite{Witsenhausen-IT-76}. Similarly, we consider the compression capacity for this model, which is defined as the maximum average number of outputs of the source $Y$ that can be reconstructed at the decoder with zero error for one use of the system.
We can see that the above model is a special case of the model $(00;C_1,C_2;f_{\rm{\text{id}}})$ by letting $C_1=+\infty$ and $C_2=C$. Thus, we immediately obtain the compression capacity $C$ for the model depicted in Fig.\,\ref{fig:Side-infor-model-00} as stated in the following corollary.

\begin{cor}
 For the compression problem over Witsenhausen's model as depicted in Fig.\,\ref{fig:Side-infor-model-00}, the compression capacity is $C$.
\end{cor}

\begin{figure}[t]
 \centering
  {
\begin{tikzpicture}[auto, node distance=2cm,>=latex',thick]
    % We start by placing the blocks
    \node [input, name=inputX, label=right:$\boldsymbol{X}$] at (9,-0.8) {};
    \node [input, name=inputY, label=left:$\boldsymbol{Y}$] at (2.5,-2.4) {};
    \node [block] (en2) at (5.5,-2.4){{\bf En}};
    \node [block] (de)  at (9,-2.4) {{\bf De}};
    \node [output, name=output, label=above:$\text{$\boldsymbol{Y}$}$] at (11.5,-2.4) {};
    % We draw an edge between the controller and system block to
    % calculate the coordinate u. We need it to place the measurement block.
    \draw [->] (inputX) -- (de);
    %\draw [->] (inputY) -- (en2);
    \draw [->] (inputY) -- (en2);
    \draw [->] (en2) -- node[above] {$C$} (de);
    \draw [->] (de) -- (output);
\end{tikzpicture}
}
\caption{A  compression problem  over Witsenhausen's model.}\label{fig:Side-infor-model-00}
\end{figure}
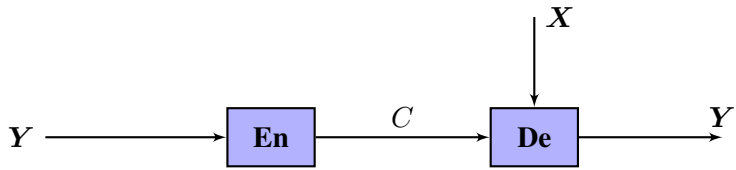

The following two theorems present the compression capacities for the models $(10;C_{1}, C_{2};  f)$ and $(11;C_{1}, C_{2};  f)$, respectively.

\begin{theorem}\label{thm:10-comp-cap}
Consider the model of distributed function compression system $(10;C_{1}, C_{2};  f)$ as depicted in Fig.\,\ref{fig:System_use-10}, where $C_1\geq C_2$. Then,
\begin{equation*}
\mathcal{C}(10;C_1,C_2;f)=C_2.
\end{equation*}
\end{theorem}

\begin{proof}
First, we readily see that  each admissible code for $(00;C_1,C_2;f)$ is also an admissible code for
$(10,C_1,C_2;f)$. Together with Theorem~\ref{thm:00-comp-cap},  this implies that
\begin{equation}
\mathcal{C}(10;C_1,C_2;f)\geq \mathcal{C}(00;C_1,C_2;f)=C_2.\label{pfthm:10-comp-cap-dir-1}
\end{equation}
In the following, we will prove that
\begin{equation*}
\mathcal{C}(10;C_1,C_2;f)\leq C_2.\label{pfthm:10-comp-cap-conv-1}
\end{equation*}

For any positive integer $k$, we
consider an arbitrary admissible $k$-shot code  $\mbC=\big\{\phi_{1},\phi_{2};~\psi\big\}$ for the model $(10;C_1,C_2;f)$. For this code $\mbC$,
we readily see that
for any two pairs $(\vec{x},\vec{y})$ and $(\vec{x},\vec{y}\,')$ in $\mathcal{A}^k\times\mathcal{A}^k$ with $\vec{y}\neq\vec{y}\,'$ (where $f(\vec{x},\vec{y})\neq f(\vec{x},\vec{y}\,')$),
\begin{equation}
\big(\phi_{1}(\vec{x}),\phi_{2}(\vec{x},\vec{y})\big)\neq
\big(\phi_{1}(\vec{x}),\phi_{2}(\vec{x},\vec{y}\,')\big),
\label{eq:case01,converse-pf-1}
\end{equation}
because, otherwise, we have
\begin{equation*}
f(\vec{x},\vec{y})=\psi\big(\phi_{1}(\vec{x}),\phi_{2}(\vec{x},\vec{y})\big)=\psi
\big(\phi_{1}(\vec{x}),\phi_{2}(\vec{x},\vec{y}\,')\big)=f(\vec{x},\vec{y}\,'),
\end{equation*}
which is a contradiction to $f(\vec{x},\vec{y})\neq f(\vec{x},\vec{y}\,')$. Thus, it follows from \eqref{eq:case01,converse-pf-1} that for each $\vec{x}\in\mathcal{A}^k$,
\begin{equation*}
\phi_{2}(\vec{x},\vec{y})\neq \phi_{2}(\vec{x},\vec{y}\,'),\quad\forall~\vec{y},\vec{y}\,'\in\mathcal{A}^k
~\text{with}~\vec{y}\neq\vec{y}\,'.
\end{equation*}
This further implies that
\begin{equation*}
|\Im\,\phi_{2}|\geq|\mathcal{A}^k|=2^k.
\end{equation*}
It follows from  \eqref{dis-cod-def-n_i(C)} and \eqref{dis-cod-def-n(C)} that
  \begin{align}
n&=\max\bigg\{\bigg\lceil\frac{\log|\Im\,\phi_{1}|}{C_1}\bigg\rceil,\bigg\lceil\frac{\log|\Im\,\phi_{2}|}{C_2}
\bigg\rceil\bigg\}\geq \frac{\log|\Im\,\phi_{2}|}{C_2}\geq \frac{k}{C_2},\nonumber
\end{align}
i.e.,
\begin{equation}
R(\mbC)=\frac{k}{n}\leq C_2.\label{eq:case01,converse-pf-2}
\end{equation}
Finally, we note that the upper bound \eqref{eq:case01,converse-pf-2} is true for each positive integer $k$ and each admissible $k$-shot code for the model $(10;C_1,C_2;f)$. We thus have proved that
  $\mathcal{C}(10;C_1,C_2;f)\leq C_2$.
  Together with \eqref{pfthm:10-comp-cap-dir-1}, the theorem has been proved.
\end{proof}

  \begin{theorem}\label{thm:11-comp-cap}
  Consider the model of distributed function compression system $(11;C_{1}, C_{2};  f)$ as depicted in Fig.\,\ref{fig:System_use-11}, where $C_1\geq C_2$. Then,
\begin{equation*}
\mathcal{C}(11;C_1,C_2;f)=\frac{C_1+C_2}{\log3}.
\end{equation*}
\end{theorem}

\begin{IEEEproof}
Let $k$ be an arbitrary positive integer.
Consider an admissible $k$-shot code  $\big\{\phi_{1},\phi_{2};~\psi\big\}$ for the model $(11;C_1,C_2;f)$.
Clearly,  for any two  source message pairs $(\vec{x},\vec{y})$ and $(\vec{x}\,',\vec{y}\,')$ in $\mathcal{A}^k\times\mathcal{A}^k$ such that $f(\vec{x},\vec{y})\neq f(\vec{x}\,',\vec{y}\,')$, we have
\begin{equation*}
\big(\phi_{1}(\vec{x},\vec{y}),\phi_{2}(\vec{x},\vec{y})\big)\neq
\big(\phi_{1}(\vec{x}\,',\vec{y}\,'),\phi_{2}(\vec{x}\,',\vec{y}\,')\big)
.\label{eq:case11,converse-pf-0}
\end{equation*}
This implies that
\begin{equation}
|\Im\,\phi_1|\cdot|\Im\,\phi_2|\geq\big|\Im\,(\phi_{1},\phi_{2})\big|\geq |\Im\,f|^k=3^k,\label{11case-conver-pf-eq1}
\end{equation}
where we let
\begin{equation*}
\Im\,(\phi_{1},\phi_{2})=\Big\{\big(\phi_{1}(\vec{x},\vec{y}),\phi_{2}(\vec{x},\vec{y})\big):~\forall~
(\vec{x},\vec{y})\in\mathcal{A}^k\times\mathcal{A}^k\Big\}.
\end{equation*}
By \eqref{dis-cod-def-n(C)}, we thus obtain that
\begin{equation*}
n=\max\bigg\{\bigg\lceil\frac{\log|\Im\,\phi_{1}|}{C_1}\bigg\rceil,\bigg\lceil\frac{\log|\Im\,\phi_{2}|}{C_2}
\bigg\rceil\bigg\}\geq\max\bigg\{\frac{\log|\Im\,\phi_{1}|}{C_1},\frac{\log|\Im\,\phi_{2}|}{C_2}
\bigg\},
\end{equation*}
or equivalently,
\begin{equation*}
n C_1\geq \log|\Im\,\phi_{1}|\quad\text{and}\quad n C_2\geq \log|\Im\,\phi_{2}|.
\end{equation*}
By \eqref{11case-conver-pf-eq1}, we consider
\begin{equation*}
n(C_1+C_2)\geq\log|\Im\,\phi_{1}|+\log|\Im\,\phi_{2}|\geq \log|\Im\,(\phi_{1},\phi_{2})|\geq\log3^k=k\log3,
\end{equation*}
which implies that
\begin{equation}
R(\mbC)=\frac{k}{n}\leq\frac{C_1+C_2}{\log3}.\label{11-case-upp-bound-eq6}
\end{equation}
Finally, we note that the above upper bound \eqref{11-case-upp-bound-eq6}  is true for each positive integer $k$ and each admissible $k$-shot code   for the model $(11;C_1,C_2;f)$. We thus have proved that
\begin{equation*}
\mathcal{C}(11;C_1,C_2;f)\leq\frac{C_1+C_2}{\log3}.\label{11-case-upp-bound}
\end{equation*}

\begin{figure}[t]
 \centering
  {
\begin{tikzpicture}[auto, node distance=2cm,>=latex',thick]
    % We start by placing the blocks
    %\node [input, name=inputX, label=left:$\boldsymbol{X}$] at (2.5,0) {};
%    \node [input, name=inputY, label=left:$\boldsymbol{Y}$] at (2.5,-2.4) {};
    \node (z) at (2,-1.2)[left=-4mm] {$\boldsymbol{Z}=f(\boldsymbol{X},\boldsymbol{Y})$};
    \node [block] (en1) at (5.5,0) {{\bf En1}};
    \node [block] (en2) at (5.5,-2.4){{\bf En2}};
    \node [block] (de)  at (9,-1.2) {{\bf De}};
    \node [output, name=output, label=above:$\boldsymbol{Z}$] at (11.5,-1.2) {};
    % We draw an edge between the controller and system block to
    % calculate the coordinate u. We need it to place the measurement block.
    \draw [->] (2.5,-1.0) -- (en1);
     \draw [->] (2.5,-1.4) -- (en2);
   % \draw [->] (inputX) -- (en1);
%    \draw [->] (inputY) -- (en1);
%    \draw [->] (inputX) -- (en2);
%    \draw [->] (inputY) -- (en2);
    \draw [->] (en1) -- node       {$C_1$} (de);
    \draw [->] (en2) -- node[swap] {$C_2$} (de);
    \draw [->] (de) -- (output);
\end{tikzpicture}
}
\caption{The equivalent model of $(11;C_1,C_2;f)$.}\label{fig:System_use-11-eqv-mod}
\end{figure}
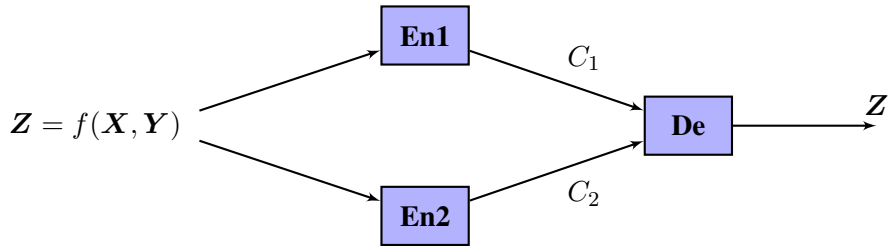
To prove that
\begin{align*}
\mathcal{C}(11;C_1,C_2;f)\geq\frac{C_1+C_2}{\log3},
\end{align*}
we consider the model as depicted in
Fig.\,\ref{fig:System_use-11-eqv-mod}, where $$\boldsymbol{Z}=f(\boldsymbol{X},\boldsymbol{Y})=\boldsymbol{X}+\boldsymbol{Y}.$$
Clearly, each admissible code for this model can  be regarded as an admissible code for $(11;C_1,C_2;f)$. For each positive integer $k$, we consider directly transmitting
\begin{equation*}
\vec{z}=(z_1,z_2,\cdots,z_k)\in \Im f^{\,k}=\{0,1,2\}^k
\end{equation*}
by using the channels  $(\Enone,~\De)$ and $(\Entwo,~\De)$ $n$ times, where
\begin{equation}\label{equ1-pf-thm-11}
n=\bigg\lceil\frac{k\log 3}{C_1+C_2}\bigg\rceil\geq \frac{k\log 3}{C_1+C_2}.
\end{equation}
We readily see that the decoder $\De$ can recover $\vec{z}$ since by~\eqref{equ1-pf-thm-11},
\begin{equation*}
2^{nC_1}\cdot 2^{nC_2}\geq 2^{k\log3}=3^k=|\Im\,f|^k.
\end{equation*}
Together with the fact that the compression rate
\begin{equation*}
\frac{k}{n}=\dfrac{k}{\,\Big\lceil\frac{k\log 3}{C_1+C_2}\Big\rceil\,} \to \frac{C_1+C_2}{\log3}\quad\text{as}\quad  k\to+\infty.
\end{equation*}
We thus prove that $\frac{C_1+C_2}{\log3}$ is achievable, i.e.,
 \begin{equation*}
\mathcal{C}(11;C_1,C_2;f)\geq\frac{C_1+C_2}{\log3}.\label{11-case-low-bound}
\end{equation*}
We thus have proved the theorem.
\end{IEEEproof}

\section{Characterization of the Compression Capacity: $\boldsymbol{\mathrm{s}}_1\boldsymbol{\mathrm{s}}_2=01$}\label{sec:compre-cap01}

%\subsection{Characterization of the Compression Capacity of $(01;2,1;f)$}\label{sec:compre-cap01;2,1case}

We will characterize the compression capacity for the remaining case of the model $(01;C_1,C_2;f)$ as depicted in Fig.\,\ref{fig:System_use}. The characterization of this compression capacity is very difficult. We first focus on a typical special case $(01;2,1;f)$ (i.e., $C_1=2$ and $C_2=1$) to explicitly show our technique in which we develop a novel graph coloring approach, and the proof is highly nontrivial. With the same technique, we similarly characterize the compression capacity $\mathcal{C}(01;C_1,C_2;f)$. We first present the compression capacity of the model $(01;2,1;f)$ in the following theorem.

\begin{theorem}\label{thm:01;2,1-comp-cap}
Consider the model of distributed function compression system $(01;2,1;  f)$. Then,
\begin{equation*}
\mathcal{C}(01;2,1;f)=\log_{3} 6.
\end{equation*}
\end{theorem}

\subsection{Proof of Theorem~\ref{thm:01;2,1-comp-cap}}\label{subsec_pf_Thm4}

\smallskip
\noindent\textbf{Direct Part:}

 In this part, we will prove the achievability of $\log_{3}6$, or equivalently,
\begin{equation*}
\mathcal{C}(01;2,1;f)\geq\log_{3}6.\label{N_cap_low_bound_log_3_6}
\end{equation*}

For any positive integer $k$,
we will construct an admissible $k$-shot code $\mbC=\big\{\phi_1,\phi_2;~\psi\big\}$ for the model $(01;2,1;f)$ as follows. First, let
 \begin{equation*}
 \vec{x}=(x_1,x_2,\cdots,x_k)\in\mathcal{A}^k\quad \text{and}\quad \vec{y}=(y_1,y_2,\cdots,y_k)\in\mathcal{A}^k
 \end{equation*}
  be two arbitrary source messages generated by
 $\boldsymbol{X}$ and $\boldsymbol{Y}$, respectively. Let
\begin{equation}\label{eq:code-constr-k_1-k_2-k_3}
k_{1}=\bigg\lceil\frac{k}{\log6}\bigg\rceil.
 \end{equation}
 \begin{itemize}
\item The  encoding functions $\phi_{1}(\vec{x},\vec{y})$ and $\phi_{2}(\vec{y})$ are defined to be two $k$-vectors below:
\begin{align}
 \phi_{1}(\vec{x},\vec{y})& =
 \begin{bmatrix}
  x_{i}+y_{i}: & 1\leq i\leq k_{1}-1 \\
  x_{j}: & k_{1}\leq j\leq k\\
  \end{bmatrix},\label{eq:encod-fun-1}\\
 \phi_{2}(\vec{y})& =
 \begin{bmatrix}
 0: & 1\leq i\leq k_{1}-1 \\
 y_{j}: & k_{1}\leq j\leq k\\
 \end{bmatrix};\label{eq:encod-fun-3}
\end{align}
\item The decoding function $\psi$ is defined by
\begin{equation*}
		\psi\big( \phi_{1}(\vec{x},\vec{y}),\phi_{2}(\vec{y})\big)
=\phi_{1}(\vec{x},\vec{y})+\phi_{2}(\vec{y}).
\end{equation*}
\end{itemize}
We readily see that
\begin{equation*}
\phi_{1}(\vec{x},\vec{y})+\phi_{2}(\vec{y})=\vec{x}+\vec{y},
\end{equation*}
which implies the admissibility of
 the code  $\mbC$.
It follows from \eqref{eq:encod-fun-1} and \eqref{eq:encod-fun-3} that
\begin{align}
|\Im\,\phi_{1}|=3^{k_{1}-1}\cdot 2^{k-k_{1}+1}\quad \text{and} \quad |\Im\,\phi_{2}|= 2^{k-k_1+1},\nonumber
\end{align}
which are equivalent to
\begin{equation*}
\log|\Im\,\phi_{1}|=(k_{1}-1)\log3+(k-k_{1}+1)\quad\text{and}\quad \log|\Im\,\phi_{2}|=k-k_1+1.
\end{equation*}
Together with \eqref{eq:code-constr-k_1-k_2-k_3}, we further obtain that
\begin{align}
\log|\Im\,\phi_{1}|&=\bigg(\bigg\lceil\frac{k}{\log6}\bigg\rceil-1\bigg)\log3+\bigg( k-\bigg\lceil\frac{k}{\log6}\bigg\rceil+1\bigg)\nonumber\\
&\leq \frac{k}{\log6}\cdot\log3+\bigg( k-\frac{k}{\log6}+1 \bigg)= 2k\log_{6}3+1\label{N_theta_1_upper_bound},
\end{align}
and
\begin{align}
\log|\Im\,\phi_{2}|= k-\bigg\lceil\frac{k}{\log 6}\bigg\rceil+1
\leq k-\frac{k}{\log6}+1
= k\log_{6}3+1\label{N_theta_3_upper_bound}.
\end{align}
It follows from \eqref{N_theta_1_upper_bound} and \eqref{N_theta_3_upper_bound} that
\begin{align}
n(\mbC)& =\max\bigg\{\bigg\lceil\frac{\log|\Im\,\phi_{1}|}{2}\bigg\rceil,~\Big\lceil\log|\Im\,\phi_{2}|
\Big\rceil\bigg\}\nonumber\\
&\leq\max\bigg\{\frac{\log|\Im\,\phi_{1}|}{2}+1,~\log|\Im\,\phi_{2}|+1
\bigg\}\nonumber\\
&\leq\Big(k\log_{6}3+1\Big)+1=k\log_{6}3+2.\nonumber
\end{align}
 Thus,  the  compression rate $R(\mbC)$ of the  code $\mbC$ satisfies
\begin{equation}\label{eq:constr-code-com-rate-lower-bound}
R(\mbC)=\frac{k}{n(\mbC)}\geq \frac{k}{k\log_{6}3+2}.
\end{equation}
The above inequality \eqref{eq:constr-code-com-rate-lower-bound} is valid for all $k\geq 1$, and thus
\begin{equation*}
\mathcal{C}(01;2,1;f)\geq \frac{k}{k\log_{6}3+2},\quad\forall~k\geq 1.
\end{equation*}
Consequently,
\begin{equation*}
\mathcal{C}(01;2,1;f)\geq\lim_{k\to\infty} \frac{k}{k\log_{6}3+2}=\log_{3}6.
\end{equation*}
The achievability of $\log_{3}6$ has been proved.

\smallskip
\noindent\textbf{Converse Part:}

In this part, we will prove that
\begin{equation*}
\mathcal{C}(01;2,1;f)\leq\log_{3}6.\label{N_cap_upp_bound_log_3_6}
\end{equation*}

We first present some notation and definitions. For any  positive integer $k$ and any admissible $k$-shot code $\mbC=\big\{\phi_1,\phi_2;~\psi\big\}$ for the model $(01;2,1;f)$, by \eqref{dis-fun-model-encod-fun-def}
 we note that the domain of $\phi_{2}$ is $\mathcal{A}^k$ ($=\{0,1\}^k$) and thus $1\leq |\Im\,\phi_{2}|\leq|\mathcal{A}^k|=2^{k}$.
For  any positive integer $m$ with $1\leq m\leq 2^k$, we define
\begin{equation}
\begin{aligned}
\chi_{m}\triangleq&\min\big\{|\Im\,\phi_{1}|:~\text{$\mbC=\big\{\phi_{1},\phi_2;~\psi\big\}
$  is an admissible $k$-shot code with $|\Im\,\phi_{2}|=m$} \big\}.\label{eq:def:chi_m}
\end{aligned}
\end{equation}
We remark that for any positive integer $1\leq m\leq 2^k$ and any  function $\phi_2$ from $\mathcal{A}^k$ to $\Im\,\phi_2$ with $|\Im\,\phi_2|=m$, there always exists an encoding function $\phi_{1}$ and a decoding function $\psi$ such that $\mbC=\big\{\phi_1,\phi_2;~\psi\big\}$  is an admissible $k$-shot code, and
thus $\chi_m$ is well-defined for each positive integer $m$ with $1\leq m\leq 2^k$. To be specific,  we let
\begin{equation*}
\phi_{1}(\vec{x}, \vec{y})=\vec{x}+\vec{y},\quad\forall~(\vec{x},\vec{y})\in\mathcal{A}^k\times\mathcal{A}^k
\end{equation*}
regardless of $\phi_2$. We readily see that this $k$-shot code $\mbC=\big\{\phi_1,\phi_2;~\psi\big\}$ can
compress $\vec{x}+\vec{y}$ with zero error by setting  the decoding function $\psi(\phi_1,\phi_2)=\phi_1$.
This shows that the code $\mbC$ is admissible.

\begin{definition}\label{def:conf-graph-G^k}
For two subsets $\mathcal{M} $ and $ \mathcal{L}$  of $\mathcal{A}^{k}$, let
\begin{equation*}
\mathcal{M}+\mathcal{L}\triangleq\big\{\vec{x}+\vec{y}:
~\text{$(\vec{x},\vec{y})\in\mathcal{M}\times\mathcal{L}$}\big\}.\footnote{If either $\mathcal{M}$
or $\mathcal{L}$ is an empty set, then $\mathcal{M}+\mathcal{L}$ is viewed as an empty set, too.}
\label{eq:def-set-M+L}
\end{equation*}
The conflict graph of $\mathcal{M}+\mathcal{L}$, denoted by  $G(\mathcal{M},\mathcal{L})$, is defined as follows: there are $|\mathcal{M}|\cdot|\mathcal{L}|$ vertices in the graph labeled by all the pairs $(\vec{x}, \vec{y})$ in $\mathcal{M}\times\mathcal{L}$, and there exists an undirected edge between two vertices $(\vec{x}, \vec{y})$ and $(\vec{x}\,',\vec{y}\,')$ if and only if $\vec{x}+\vec{y}\neq \vec{x}\,'+\vec{y}\,'$.
\end{definition}

A \emph{coloring} $c$ of  the conflict graph $G(\mathcal{M},\mathcal{L})$ is a mapping from $\mathcal{M}\times\mathcal{L}$ to $\Im\,c$ that satisfies
\begin{equation*}
c(\vec{x},\vec{y})\neq c(\vec{x}\,',\vec{y}\,'),\quad \forall~(\vec{x},\vec{y})~\text{and}~ (\vec{x}\,',\vec{y}\,')~\text{in}~\mathcal{M}\times\mathcal{L}~\text{s.t.}~\vec{x}+\vec{y}\neq \vec{x}\,'+\vec{y}\,'.
\end{equation*}
The size of the image set $|\Im\,c|$ is called the \emph{chromatic number} of the coloring $c$.
The minimum chromatic number of all the colorings of $G(\mathcal{M},\mathcal{L})$ is called the
 \emph{minimum chromatic number} of the conflict graph $G(\mathcal{M},\mathcal{L})$, denoted by
  $\chi\big(G(\mathcal{M},\mathcal{L})\big)$.
Then, we give the following lemma.
\begin{lemma}\label{lemma:chi(G(P))-equiv-form}
The minimum chromatic number $\chi\big(G(\mathcal{M},\mathcal{L})\big)$ is equal to the number of values of the binary arithmetic sum for all possible pairs $(\vec{x},\vec{y})\in\mathcal{M}\times \mathcal{L}$, i.e.,
\begin{align}\label{eq:chi(G(P))-low-bound}
\chi\big(G(\mathcal{M},\mathcal{L})\big)=|\mathcal{M}+\mathcal{L}|.\nonumber
\end{align}
\end{lemma}

\begin{IEEEproof}
Let $d=|\mathcal{M}+\mathcal{L}|$ and clearly, $d\leq |\mathcal{M}|\cdot|\mathcal{L}|$. We first prove that $\chi\big(G(\mathcal{M},\mathcal{L})\big)\geq d$.
Let $(\vec{x}_i,\vec{y}_i),~1\leq i \leq d$ be $d$ pairs in $\mathcal{M}\times\mathcal{L}$ of which the values of the binary arithmetic sum take all the $d$ different values in $\mathcal{M}+\mathcal{L}$. By Definition~\ref{def:conf-graph-G^k}, we consider the conflict graph of $\big\{(\vec{x}_i,\vec{y}_i):~1\leq i\leq d\big\}$, where there are total $d$ vertices labeled by the $d$ pairs $(\vec{x}_i,\vec{y}_i),~1\leq i\leq d$, and there exists an undirected edge between two vertices $(\vec{x}_i, \vec{y}_i)$ and $(\vec{x}_j,\vec{y}_j)$, $1\leq i,j\leq d$ if and only if $\vec{x}_i+\vec{y}_i\neq \vec{x}_j+\vec{y}_j$. Then we readily see that this conflict graph is a complete graph~$K_d$ of $d$ vertices. It thus follows from  the fact that $K_d$ is a subgraph of $G(\mathcal{M},\mathcal{L})$ and $\chi(K_d)=d$ that
\begin{equation*}
\chi\big(G(\mathcal{M},\mathcal{L})\big)\geq\chi(K_d)=d.
\end{equation*}

Next, we prove that $\chi\big(G(\mathcal{M},\mathcal{L})\big)\leq d$. Consider the mapping $c$ from $\mathcal{M}\times\mathcal{L}$ to $\Im\,c$ ($=\mathcal{M}+\mathcal{L}$)
 such that
\begin{equation*}
c(\vec{x}, \vec{y})=\vec{x}+\vec{y},\quad\forall~(\vec{x}, \vec{y})\in\mathcal{M}\times\mathcal{L}.
\end{equation*}
We readily see that $c$ is a coloring of the conflict graph $G(\mathcal{M},\mathcal{L})$ with
the chromatic number $d$. This immediately implies that
$\chi\big(G(\mathcal{M},\mathcal{L})\big)\leq d$. We thus have proved that $\chi\big(G(\mathcal{M},\mathcal{L}) \big)= d$.
\end{IEEEproof}

In the rest of proof, we use $G$ to denote the conflict graph of $\mathcal{A}^k+\mathcal{A}^k$.
We can see that for an admissible $k$-shot  code $\mbC=\big\{\phi_1,\phi_2;~\psi\big\}$, $(\phi_{1},\phi_{2})$  is a coloring of the conflict graph $G$, because for any two adjacent vertices $(\vec{x}, \vec{y})$ and $(\vec{x}\,',\vec{y}\,')$  (i.e., $\vec{x}+\vec{y}\neq \vec{x}\,'+\vec{y}\,'$), we have
\begin{equation*}
\big(\phi_{1}(\vec{x},\vec{y}), \phi_{2}(\vec{y})\big)\neq \big(\phi_{1}(\vec{x}\,',\vec{y}\,'),
\phi_{2}(\vec{y}\,')\big).\label{eq:adm-code-is-color-G^k}
\end{equation*}
When $\mathcal{M}=\mathcal{A}^k$, we use $G(\mathcal{L})$ to denote the conflict graph of $\mathcal{A}^k+\mathcal{L}$ for notational simplicity, which is a subgraph of $G$. By Lemma~\ref{lemma:chi(G(P))-equiv-form}, we immediately have $\chi\big(G(\mathcal{L})\big)=\big|\mathcal{A}^k+\mathcal{L}\big|$.

\begin{lemma}\label{lemma:N_relat_code_color_set}
For a nonnegative integer $\ell$ with  $0\leq \ell  \leq 2^{k}$, let
\begin{align}\label{def:Q_k(l)}
Q_{k}(\ell)\triangleq\min_{\substack{\;\;\mathcal{L}\subseteq\mathcal{A}^{k}\;\text{\rm s.t.}\;|\mathcal{L}|
=\ell}}\big|\mathcal{A}^{k}+ \mathcal{L}\big|=\min_{\substack{\;\;\mathcal{L}\subseteq\mathcal{A}^{k}\;\text{\rm s.t.}\;|\mathcal{L}|
=\ell}}\chi\big(G(\mathcal{L})\big).
\end{align}
Then, for an  admissible $k$-shot  code $\mbC=\big\{\phi_1,\phi_2;~\psi\big\}$ with $|\Im\, \phi_{2}|=m$ \rm{(}$1\leq m\leq 2^{k}$\rm{)},
\begin{equation*}
|\Im\,\phi_{1}|\geq\chi_{m}\geq Q_{k}\Big(\Big\lceil\;\frac{2^{k}}{m}\;\Big\rceil\Big).\label{eq:N_lemma:relat_code_color_set}
	\end{equation*}
\end{lemma}

To prove the above lemma, we need the following lemma, of which the proof is deferred to Appendix~\ref{append-lemma:chi_m_low_bound}.

\begin{lemma}\label{lemma:chi_m_low_bound}
For any positive integer $m$ such that $1\leq m\leq 2^k$,
\begin{equation*}\label{eq:chi_m=min_max_color}
	\chi_{m}{\color{blue}}=\min_{\substack{\text{ \rm{all partitions} }\\\text{$\{P_1,P_2,\cdots,P_m\}$ \rm{of} $\mathcal{A}^{k}$}}}\max_{1\leq i\leq m}\chi\big(G(P_i)\big),
\end{equation*}
where every subset {\rm{(}}called a block{\rm{)}} of a  partition is nonempty.
\end{lemma}

\begin{proof}[Proof of Lemma~\ref{lemma:N_relat_code_color_set}]
For an admissible $k$-shot code $\mbC=\big\{\phi_1,\phi_2;~\psi\big\}$ with $|\Im\,\phi_2|=m$, by the definition of $\chi_{m}$ (cf.~\eqref{eq:def:chi_m}), we have
 \begin{equation*}
|\Im\,\phi_{1}|\geq\chi_{m}.
 \end{equation*}
Thus, it  suffices to prove
\begin{equation*}
\chi_{m}\geq Q_{k}\Big(\Big\lceil\;\frac{2^{k}}{m}\;\Big\rceil\Big).
\end{equation*}
 Consider an arbitrary partition $\mathcal{P}=\{P_1,P_2,\cdots,P_m\}$ of  $\mathcal{A}^{k}$. We see that there exists a block in $\mathcal{P}$, say $P_1$, such that
\begin{equation}
|P_1|\geq \frac{2^k}{m},\label{pflemma:N_relat_code_color_set-1}
\end{equation}
 because otherwise
 \begin{equation*}
 \big|\mathcal{A}^k\big|=\bigg|\bigcup_{i=1}^{m}P_i\bigg|=\sum_{i=1}^{m}|P_i|<m\cdot\frac{2^k}{m}=2^k,
 \end{equation*}
a contradiction.
The inequality \eqref{pflemma:N_relat_code_color_set-1} immediately implies that
$|P_1|\geq \big\lceil2^{k}/m\big\rceil$, and  thus we obtain that
\begin{equation}
\max_{1\leq i\leq m}\chi\big(G(P_i)\big)\geq \chi\big(G(P_1)\big)\geq \chi\big(G(\widehat{P}_1)\big) \geq\min_{\text{$P\subseteq\mathcal{A}^{k}$ s.t. $|P|=\big\lceil\frac{2^{k}}{m}\big\rceil$ }}\chi\big(G(P)\big),\label{pflemma:N_relat_code_color_set-2}
\end{equation}
where $\widehat{P}_1$ is a subset of $P_1$ with $|\widehat{P}_1|= \big\lceil2^{k}/m\big\rceil$, and
the inequality $\chi\big(G(P_1)\big)\geq \chi\big(G(\widehat{P}_1)\big)$  follows because a coloring of $G(P_1)$ restricted to its subgraph $G(\widehat{P}_1)$ is a coloring of $G(\widehat{P}_1)$.
Together with Lemma~\ref{lemma:chi_m_low_bound}, we obtain that
\begin{align}
\chi_{m}&=\min_{\substack{\text{ all partitions }\\\text{$\{P_1,P_2,\cdots,P_m\}$ of $\mathcal{A}^{k}$}}}\max_{1\leq i\leq m}\chi\big(G(P_i)\big)\nonumber\\
&\geq\min_{\text{$P\subseteq\mathcal{A}^{k}$ s.t. $|P|=\big\lceil\frac{2^{k}}{m}\big\rceil$ }}\chi\big(G(P)\big)\label{pflemma:N_relat_code_color_set-3.1}\\
&=\min_{\text{$P\subseteq\mathcal{A}^{k}$ s.t. $|P|=\big\lceil\frac{2^{k}}{m}\big\rceil$ }}|\mathcal{A}^k+P|\label{pflemma:N_relat_code_color_set-3.2}\\
&=Q_{k}\Big(\Big\lceil\;\frac{2^{k}}{m}\;\Big\rceil\Big),\label{pflemma:N_relat_code_color_set-3.3}
\end{align}
where the inequality \eqref{pflemma:N_relat_code_color_set-3.1} follows from the inequality \eqref{pflemma:N_relat_code_color_set-2} which  holds for any partition $\mathcal{P}$ of $\mathcal{A}^k$ that includes $m$ blocks;   the  equality \eqref{pflemma:N_relat_code_color_set-3.2} follows from $\chi\big(G(P)\big)=|\mathcal{A}^k+P|$ by Lemma~\ref{lemma:chi(G(P))-equiv-form}; and  the  equality \eqref{pflemma:N_relat_code_color_set-3.3} follows from the definition of $Q_k(\ell)$ in \eqref{def:Q_k(l)}.
The lemma has been proved.
\end{proof}

Now, we define the following function of nonnegative integers to lower bound $\chi\big(G(\mathcal{L})\big),~\forall~\mathcal{L}\subseteq\mathcal{A}^k$:
\begin{equation}
h(\ell)=\ell^{\log 3-1},\quad \ell=0,1,2,\cdots,\label{def:h(l)}
\end{equation}
where we take $h(0)=0$ because $h(x)$ tends to $0$ as $x\to 0^{+}$.

\begin{lemma}\label{lemma:N_X+L_low_bound}
Let $\mathcal{L}$ be a subset of $\mathcal{A}^k$. Then,
\begin{equation*}
 \chi\big(G(\mathcal{L})\big)=|\mathcal{A}^{k}+\mathcal{L}|\geq 2^{k}\cdot h(\ell),\label{lemma:eq-N_X+L_low_bound}
  \end{equation*}
  where $\ell=|\mathcal{L}|$ with $0\leq \ell\leq 2^k$.
\end{lemma}

 Before proving the above lemma, we need the following lemma, whose proof is deferred to Appendix~\ref{appendix-pf-lemma:h(l)-property}.

\begin{lemma}\label{lemma:h(l)-property}
Consider a nonnegative integer $\ell$.
For any two nonnegative integers $\ell_a$ and $\ell_b$ satisfying $\ell=\ell_a+\ell_b$ and $\ell_a\geq\ell_b$,
\begin{equation}
2\cdot h(\ell_a)+h(\ell_b)\geq 2\cdot h(\ell).\label{lemma:eq-h(l)-property}
\end{equation}
\end{lemma}

\begin{remark}
For a positive real number $\tau$, we let $h_{\tau}(\ell)=\ell^{\tau}$ for $\ell=0,1,2,\cdots$ {\rm{(}}where we similarly take $h_{\tau}(0)=0${\rm{)}}. By~\eqref{def:h(l)}, we see that $h(\ell)=h_{\log3-1}(\ell)$ for $\tau=\log3-1$.
In fact, $\log3-1$ is the maximum~$\tau$ such that the inequality
\begin{align*}
2\cdot h_{\tau}(\ell_{a})+h_{\tau}(\ell_{b})\geq 2\cdot h_{\tau}(\ell)
\end{align*}
(the same as the inequality \eqref{lemma:eq-h(l)-property}) is satisfied for any nonnegative integer $\ell$, where $\ell_a$ and $\ell_b$ are two arbitrary nonnegative integers  such that $\ell=\ell_a+\ell_b$ and $\ell_a\geq\ell_b$.
\end{remark}

\begin{proof}[Proof of~Lemma~\ref{lemma:N_X+L_low_bound}]
We will prove this lemma by induction in $k$. We first consider the case $k=1$.
For a subset $\mathcal{L}\subseteq\mathcal{A}$, we readily calculate that
\begin{equation}
|\mathcal{A}+\mathcal{L}|=
\begin{cases}
 0, & \text{if $|\mathcal{L}|=0$, i.e., $\mathcal{L}=\emptyset$;}\\
 2, & \text{if $|\mathcal{L}|=1$, i.e., $\mathcal{L}=\{0\}$ or $\{1\}$;}\\
 3, & \text{if $|\mathcal{L}|=2$, i.e., $\mathcal{L}=\mathcal{A}$.}\\
\end{cases}\label{pflemma:N_X+L_low_bound-1}
\end{equation}
After a simple calculation for the function $h(\ell)$, we have
\begin{equation}
h(0)=0,\quad h(1)=1,\quad \text{and}\quad h(2)=\frac{3}{2}.\label{pflemma:N_X+L_low_bound-1.1}
\end{equation}
It thus follows from \eqref{pflemma:N_X+L_low_bound-1} and \eqref{pflemma:N_X+L_low_bound-1.1} that
\begin{equation*}
|\mathcal{A}+\mathcal{L}|=2\cdot h\big(|\mathcal{L}|\big),\quad\forall~\mathcal{L}\subseteq\mathcal{A}.
\label{pflemma:N_X+L_low_bound-2}
\end{equation*}
This shows that the lemma is true for the case $k=1$.
We now assume the lemma is true for some $k\geq 1$. Next, we consider the case $k+1$. For an arbitrary subset $\mathcal{L}$ of $\mathcal{A}^{k+1}$, we let
\begin{equation*}
		\mathcal{L}_{0}\triangleq\big\{ \vec{y}:~\text{$\vec{y}\in \mathcal{L}$ with $\vec{y}[1]=0$}\big\}\quad\text{and}\quad\mathcal{L}_{1}\triangleq\big\{ \vec{y}:~\text{$\vec{y}\in \mathcal{L}$ with $\vec{y}[1]=1$}\big\},\label{eq:N_def:L_y}
	\end{equation*}
where $\vec{y}[1]$ stands for the first component of the vector $\vec{y}$.
Clearly, $\mathcal{L}_{0}\cup\mathcal{L}_{1}=\mathcal{L}$ and $\mathcal{L}_{0}\cap\mathcal{L}_{1}=\emptyset$. We remark that  $\mathcal{L}_{0}$ (resp. $\mathcal{L}_{1}$)  may be an empty set. We further let $\ell_{0}=|\mathcal{L}_{0}|$, $\ell_{1}=|\mathcal{L}_{1}|$, and let $(a,b)$ be a permutation of $(0,1)$ such that $\ell_{a}\geq \ell_{b}$. Then, we have
\begin{equation}
\begin{aligned}
\big|\mathcal{A}^{k+1}+\mathcal{L}\big|&= \big|\mathcal{A}^{k+1}+\mathcal{L}_{a}\big|+\big|\big(\mathcal{A}^{k+1}+\mathcal{L}_{b}\big)
\setminus\big(\mathcal{A}^{k+1}+\mathcal{L}_{a}\big)\big|.
\end{aligned}\label{pflemma:N_X+L_low_bound-3}
\end{equation}

  We first consider
\begin{align}
\big|\mathcal{A}^{k+1}+\mathcal{L}_{a}\big|&=\sum_{z\in\mathcal{A}+\{a\}}\Big|\Big\{\vec{x}+\vec{y}:~
(\vec{x},\vec{y})\in\mathcal{A}^{k+1}\times\mathcal{L}_{a}\,~
\text{s.t.}~\vec{x}[1]+\vec{y}[1]=\vec{x}[1]+a=z\Big\}\Big|\nonumber\\
&=\sum_{z\in\mathcal{A}+\{a\}}\Big|\Big\{\vec{x}+\vec{y}:~
(\vec{x},\vec{y})\in\mathcal{A}^{k+1}\times\mathcal{L}_{a}\,~
\text{with}~\vec{x}[1]=z-a\Big\}\Big|.
\label{pflemma:N_X+L_low_bound-4}
\end{align}
Let $\mathcal{L}_{a}^{(k)}$ be the set of the $k$-vectors obtained from the $(k+1)$-vectors in $\mathcal{L}_{a}$ by deleting their first components, all of which are equal to $a$. We readily see that
\begin{equation}
\big|\mathcal{L}_{a}^{(k)}\big|=|\mathcal{L}_{a}|=\ell_{a}.\label{pf:N_lemma:recur_ineq_4.0}
\end{equation}
For each $z\in\mathcal{A}+\{a\}$, we continue to consider
\begin{align}
&\Big|\Big\{\vec{x}+\vec{y}:~
(\vec{x},\vec{y})\in\mathcal{A}^{k+1}\times\mathcal{L}_{a}~
\text{with}~\vec{x}[1]=z-a\Big\}\Big|\nonumber\\
=&\Big|\Big\{\Big[\begin{smallmatrix}z-a  \vspace{1mm} \\\vec{x}^{(k)}\end{smallmatrix}\Big]+
\Big[\begin{smallmatrix}a \vspace{1mm} \\\vec{y}^{(k)}\end{smallmatrix}\Big]:~
\big(\vec{x}^{(k)},\vec{y}^{(k)}\big)\in\mathcal{A}^{k}\times\mathcal{L}^{(k)}_{a}\Big\}\Big|\nonumber\\
=&\Big|\Big\{\vec{x}^{(k)}+\vec{y}^{(k)}:
~\big(\vec{x}^{(k)},\vec{y}^{(k)}\big)\in\mathcal{A}^{k}\times\mathcal{L}^{(k)}_{a}\Big\}\Big|\nonumber\\
=&\big|\mathcal{A}^{k}+\mathcal{L}^{(k)}_{a}\big|.\nonumber%\label{pflemma:N_X+L_low_bound-5}
\end{align}
Together with \eqref{pflemma:N_X+L_low_bound-4}, we immediately obtain that
\begin{equation}
\big|\mathcal{A}^{k+1}+\mathcal{L}_{a}\big|=\big|\mathcal{A}+\{a\}\big|\cdot \big|\mathcal{A}^k+\mathcal{L}_a^{(k)}\big|.\label{pflemma:N_X+L_low_bound-5.1}
\end{equation}

Next, we consider $\big|(\mathcal{A}^{k+1}+\mathcal{L}_{b})
\setminus(\mathcal{A}^{k+1}+\mathcal{L}_{a})\big|$.
First, we note that
\begin{align}
&\big(\mathcal{A}^{k+1}+\mathcal{L}_b\big)\setminus\big(\mathcal{A}^{k+1}+\mathcal{L}_a\big)\nonumber\\
&\supseteq\Big\{\vec{x}+\vec{y}:~(\vec{x},\vec{y})\in\mathcal{A}^{k+1}\times\mathcal{L}_b~\text{s.t.}
~\vec{x}[1]+\vec{y}[1]\in\big(\mathcal{A}+\{b\}\big)\setminus\big(\mathcal{A}+\{a\}\big)\Big\}.\nonumber
%\label{pflemma:N_X+L_low_bound-6}
\end{align}
This immediately implies that
\begin{align}
&\Big|\big(\mathcal{A}^{k+1}+\mathcal{L}_b\big)\setminus\big(\mathcal{A}^{k+1}+\mathcal{L}_a\big)\Big|\nonumber\\
&\geq\sum_{z\in(\mathcal{A}+\{b\})\setminus(\mathcal{A}+\{a\})}
\Big|\Big\{\vec{x}+\vec{y}:~(\vec{x},\vec{y})\in\mathcal{A}^{k+1}\times\mathcal{L}_b~\text{with}
~\vec{x}[1]+\vec{y}[1]=z\Big\}\Big|\nonumber\\
&=\sum_{z\in(\mathcal{A}+\{b\})\setminus(\mathcal{A}+\{a\})}
\Big|\Big\{\vec{x}+\vec{y}:~(\vec{x},\vec{y})\in\mathcal{A}^{k+1}\times\mathcal{L}_b~\text{with}
~\vec{x}[1]=z-b\Big\}\Big|.\label{pflemma:N_X+L_low_bound-7}
\end{align}
Similarly,
let $\mathcal{L}_{b}^{(k)}$ be the set of the $k$-vectors obtained from the vectors in $\mathcal{L}_{b}$ by deleting their first components, all of which are equal to $b$.
 Clearly,
\begin{equation}
\big|\mathcal{L}_{b}^{(k)}\big|=|\mathcal{L}_{b}|=\ell_{b}.\label{pflemma:N_X+L_low_bound-7.1}
\end{equation}
 Continuing from \eqref{pflemma:N_X+L_low_bound-7},
for each $z\in\big(\mathcal{A}+\{b\}\big)\setminus\big(\mathcal{A}+\{a\}\big)$ we  consider
\begin{align}
&\Big|\Big\{\vec{x}+\vec{y}:~
(\vec{x},\vec{y})\in\mathcal{A}^{k+1}\times\mathcal{L}_{b}~
\text{with}~\vec{x}[1]=z-b\Big\}\Big|\nonumber\\
&=\Big|\Big\{\Big[\begin{smallmatrix}z-b  \vspace{1mm} \\\vec{x}^{(k)}\end{smallmatrix}\Big]+
\Big[\begin{smallmatrix}b \vspace{1mm} \\\vec{y}^{(k)}\end{smallmatrix}\Big]:~\big(\vec{x}^{(k)},\vec{y}^{(k)}\big)\in
\mathcal{A}^{k}\times\mathcal{L}^{(k)}_{b}\Big\}\Big|\nonumber\\
&=\Big|\Big\{\vec{x}^{(k)}+\vec{y}^{(k)}:~\big(\vec{x}^{(k)},\vec{y}^{(k)}\big)\in
\mathcal{A}^{k}\times\mathcal{L}^{(k)}_{b}\Big\}\Big|\nonumber\\
&=\big|\mathcal{A}^{k}+\mathcal{L}^{(k)}_{b}\big|.\nonumber%\label{pflemma:N_X+L_low_bound-8}
\end{align}
By \eqref{pflemma:N_X+L_low_bound-7}, we thus  obtain that
\begin{equation}
\big|\big(\mathcal{A}^{k+1}+\mathcal{L}_b\big)\setminus\big(\mathcal{A}^{k+1}+\mathcal{L}_a\big)\big|\geq
\big|\big(\mathcal{A}+\{b\}\big)\setminus\big(\mathcal{A}+\{a\}\big)\big|\cdot \big |\mathcal{A}^k+\mathcal{L}_b^{(k)}\big|.\label{pflemma:N_X+L_low_bound-8.1}
\end{equation}

Combining \eqref{pflemma:N_X+L_low_bound-5.1} and \eqref{pflemma:N_X+L_low_bound-8.1} with \eqref{pflemma:N_X+L_low_bound-3},
we obtain that
\begin{align}
\big|\mathcal{A}^{k+1}+\mathcal{L}\big|
&\geq \big|\mathcal{A}+\{a\}\big|\cdot \big|\mathcal{A}^k+\mathcal{L}_a^{(k)}\big|
 +\big|\big(\mathcal{A}+\{b\}\big)\setminus\big(\mathcal{A}+\{a\}\big)\big|\cdot \big |\mathcal{A}^k+\mathcal{L}_b^{(k)}\big|\nonumber\\
&= 2\cdot \big|\mathcal{A}^k+\mathcal{L}_a^{(k)}\big|
 +\big|\mathcal{A}^k+\mathcal{L}_b^{(k)}\big|\label{pflemma:N_X+L_low_bound-9.2}\\
 & \geq 2 \cdot 2^k \cdot h(\ell_a)+2^k\cdot h(\ell_b)\label{pflemma:N_X+L_low_bound-9.3}\\
 &=2^k \cdot \Big[2\cdot h(\ell_a)+h(\ell_b)\Big]\nonumber\\
& \geq 2^k\cdot 2\cdot h(\ell)=2^{k+1}\cdot h(\ell),\label{pflemma:N_X+L_low_bound-9.5}
\end{align}
where the equality
\eqref{pflemma:N_X+L_low_bound-9.2} follows because
\begin{align*}
\big|\mathcal{A}+\{a\}\big|=2\quad \text{and}\quad \big|\big(\mathcal{A}+\{b\}\big)\setminus\big(\mathcal{A}+\{a\}\big)\big|=1;
\end{align*}
the inequality~\eqref{pflemma:N_X+L_low_bound-9.3} follows from
 the induction hypothesis for the case $k$ and the equalities~\eqref{pf:N_lemma:recur_ineq_4.0} and \eqref{pflemma:N_X+L_low_bound-7.1}, more specifically,
 \begin{align*}
 \big|\mathcal{A}^k+\mathcal{L}_a^{(k)}\big|\geq 2^k\cdot h\big(\big|\mathcal{L}_{a}^{(k)}\big|\big)=2^k\cdot h(\ell_{a}) \quad\text{and} \quad \big|\mathcal{A}^k+\mathcal{L}_b^{(k)}\big|\geq 2^k\cdot h\big(\big|\mathcal{L}_{b}^{(k)}\big|\big)=2^k\cdot h(\ell_{b});
 \end{align*}
 and the inequality~\eqref{pflemma:N_X+L_low_bound-9.5} follows from Lemma~\ref{lemma:h(l)-property}. Thus, we have proved the lemma.
\end{proof}

Immediately, we obtain the following consequence  of Lemma~\ref{lemma:N_X+L_low_bound}, which provides a lower bound on~$Q_{k}(\ell)$ (cf.~\eqref{def:Q_k(l)}).

\begin{cor}\label{cor:Q_k^l_low_bound}
For any nonnegative integer $\ell$ with $0\leq \ell\leq 2^{k}$,
\begin{equation*}
Q_{k}(\ell)\geq 2^{k}\cdot h(\ell).
\end{equation*}
\end{cor}
%{\color{blue}
%\begin{remark}
%\begin{equation}
%\inf_{k\geq 1}\frac{1}{k}\log Q_{k}(\ell)=\lim_{k\to+\infty}\frac{1}{k}\log Q_{k}(\ell)=\lim_{k\to+\infty}\frac{1}{k}\log\big(2^{k}\cdot h(\ell)\big)=
%\end{equation}
%\end{remark}
%}
Now, we start to prove $\mathcal{C}(01;2,1;f)\leq \log_{3}6$.
For any positive integer $k$, we consider
 an arbitrary  admissible $k$-shot code  $\mbC=\big\{\phi_1,\phi_2;~\psi\big\}$  for the model $(01;2,1;f)$. Let $m=|\Im\, \phi_{2}|$ (where $1\leq m\leq 2^{k}$).
 For the code $\mbC$,
 we have
\begin{align}
n=\max\bigg\{\bigg\lceil\frac{\log|\Im\,\phi_{1}|}{2}\bigg\rceil,~
\Big\lceil\log|\Im\,\phi_{2}|\Big\rceil\bigg\}.\label{eq:pf_upp_bound_log3(6)-1}
\end{align}
Further, by Lemma~\ref{lemma:N_relat_code_color_set} and Corollary~\ref{cor:Q_k^l_low_bound},
we obtain that
	\begin{align}
		|\Im\, \phi_{1}| \geq
\chi_{m} & \geq Q_{k}\Big(\Big\lceil\frac{2^{k}}{m}\Big\rceil\Big)\geq 2^{k}\cdot h\Big( \Big\lceil\frac{2^{k}}{m}\Big\rceil\Big)
 =2^{k}\cdot \Big\lceil\frac{2^{k}}{m}\Big\rceil^{\log 3-1}
\nonumber\\
&\geq2^{k}\cdot \Big(\frac{2^{k}}{m}\Big)^{\log 3-1}=3^k\cdot m^{1-\log3}.\nonumber%\label{eq:pf_upp_bound_log3(6)-2}
	\end{align}
This implies that
\begin{equation*}
\bigg\lceil\frac{\log|\Im\,\phi_{1}|}{2}\bigg\rceil \geq \dfrac{1}{2}\Big(k\log3-\log\frac{3}{2}\cdot\log m\Big).\label{eq:pf_upp_bound_log3(6)-4}
\end{equation*}
By $\log|\Im\,\phi_{2}|=\log m$, we continue from \eqref{eq:pf_upp_bound_log3(6)-1} and obtain that
\begin{align}
n\geq&\max\Big\{\frac{1}{2}\Big(k\log3-\log\dfrac{3}{2}\cdot\log m\Big),~\log m\Big\}.\nonumber%\label{eq:pf_upp_bound_log3(6)-5}
\end{align}
Considering all $1\leq m\leq 2^k$, we further obtain that
\begin{align}
  n&\geq\min_{\substack{m\in\{1,2,\cdots,2^k\}}}\max\Big\{\frac{1}{2}\Big(k\log3-\log\frac{3}{2}\cdot\log m\Big),~\log m\Big\}\nonumber\\
  &\geq\min_{\substack{m\in[1,2^k]}}\max\Big\{\frac{1}{2}\Big(k\log3-\log\frac{3}{2}\cdot\log m\Big),~\log m\Big\}\nonumber\\
  &=\min_{\substack{t\in[0,k]}}\max\Big\{\frac{1}{2}\Big(k\log3-\log\frac{3}{2}\cdot t\Big),~t\Big\},\label{eq:pf_upp_bound_log3(6)-6.3}
\end{align}
where in \eqref{eq:pf_upp_bound_log3(6)-6.3}, we let $t=\log m$. With \eqref{eq:pf_upp_bound_log3(6)-6.3}, we define the following function:
\begin{align}
F_{k}(t)=\max\Big\{\frac{1}{2}\Big(k\log3-\log\frac{3}{2}\cdot t\Big),~t\Big\},\quad \forall~t\in[0,k],\label{eq:pf_upp_bound_log3(6)-7}
\end{align}
which can be rewrite as
\begin{equation}
F_{k}(t)=\begin{cases}
\quad\dfrac{1}{2}\Big(k\log3-\log\dfrac{3}{2}\cdot t\Big),\qquad &\text{if}~ t\in\Big[0,~\frac{k\log3}{\log3+1}\Big);\\
\quad t,\qquad &\text{if}~ t\in\Big[\frac{k\log3}{\log3+1},~ k\Big].
\end{cases}\label{eq:pf_upp_bound_log3(6)-8}
\end{equation}
The graph of $F_{k}(t)$ is plotted in Fig.\,\ref{fig:F_k(t)}. Immediately,
\begin{equation}
\min_{t\in[0,k]}F_{k}(t)=F_{k}\Big(\frac{k\log3}{\log3+1}\Big)=\frac{k\log3}{\log3+1}=k\cdot\log_{6}3.
\label{eq:pf_upp_bound_log3(6)-9}
\end{equation}
Hence, by~\eqref{eq:pf_upp_bound_log3(6)-6.3}, we have
\begin{align}
n\geq\min_{t\in[0,k]}F_{k}(t)=k\cdot\log_{6}3,\nonumber
%\label{eq:pf_upp_bound_log3(6)-9.1}
\end{align}
i.e.,
\begin{align}\label{equ1}
R(\mbC)=\frac{k}{n} \leq \log_{3}6.
\end{align}
Finally, we note that the upper bound \eqref{equ1} on the compression rate is true for any positive integer $k$ and any admissible $k$-shot code. This immediately implies {$\mathcal{C}(01;2,1;f)\leq\log_{3}6$. Theorem \ref{thm:N_cap} is thus proved.

  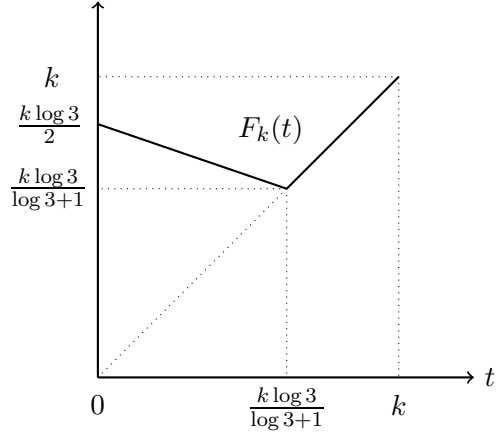
\begin{figure}[!t]
	\centering
	\begin{tikzpicture}
	\draw[->][thick](0,0)--(0,5)node[above]{};
	\draw[->][thick](0,0)--(5,0)node[right]{$t$};
	\node at (2.3,3.3){$F_{k}(t)$};
\node at (0,-0.1)[below]{$0$};
\node at (4,-0.1)[below]{$k$};
\node at (2.51022128,0)[below]{$\frac{k\log3}{\log3+1}$};
\node at (-0.37,4)[left]{$k$};
\node at (-0.07,3.3699)[left]{$\frac{k\log3}{2}$};
\node at (0,2.51022128)[left]{$\frac{k\log3}{\log3+1}$};
\draw[thick](0,3.3699)--(2.51022128,2.51022128);
\draw[thick](4,4)--(2.51022128,2.51022128);
\draw[dotted](4,4)--(4,0);
\draw[dotted](4,4)--(0,4);
\draw[dotted](2.51022128,2.51022128)--(2.51022128,0);
\draw[dotted](2.51022128,2.51022128)--(0,0);
\draw[dotted](2.51022128,2.51022128)--(0,2.51022128);
	\end{tikzpicture}
	\caption{The function $F_{k}(t)$.}
	\label{fig:F_k(t)}
\end{figure}

\subsection{Characterization of the Compression Capacity of $(01;C_1,C_2;f)$}\label{subsec_01_C1C2}

We present the compression capacity for the general case of the model $(01;C_1,C_2;f)$ with $C_1\geq C_2$ as depicted in Fig.\,\ref{fig:System_use}.

\begin{theorem}\label{C1-C2-net-model-cap}
Consider the model of distributed function compression system $(01;C_1,C_2;f)$ as depicted in Fig.\,\ref{fig:System_use}, where $C_1\geq C_2$.
Then,
\begin{equation*}
\mathcal{C}(01;C_1,C_2;f)=(C_1-C_2)\log_{3}2+C_2.
\end{equation*}
\end{theorem}

\begin{proof}
We first prove
\begin{equation}
\mathcal{C}(01;C_1,C_2;f)\leq(C_1-C_2)\log_{3}2+C_2\label{C1-C2_cap_upp_bound}
\end{equation}
by using the same technique in the converse part of the proof of Theorem~\ref{thm:01;2,1-comp-cap}.
Consider an arbitrary positive integer $k$ and an arbitrary admissible $k$-shot code $\mbC=\big\{\phi_{1}, \phi_2;~\psi\big\}$ for the model $(01;C_1,C_2;f)$. By \eqref{dis-fun-model-encod-fun-def}, the
two encoding functions $\phi_1$ and $\phi_2$ are from $\mathcal{A}^k\times\mathcal{A}^k$ to $\Im\,\phi_1$ and  from $\mathcal{A}^k$ to $\Im\,\phi_2$, respectively. We readily see that $1\leq|\Im\,\phi_2|\leq 2^k$. Let $|\Im\,\phi_2|=m$, where $1\leq m\leq 2^k$. Then, we define
\begin{equation}
\begin{aligned}
\widehat{\chi}_{m}\triangleq&\min\big\{|\Im\,\phi_1|:~\text{$\mbC=\big\{\phi_{1},\phi_2;~\psi\big\}
$ is an admissible $k$-shot  code} \big.\\\big.&\qquad\quad\qquad\quad\quad\quad\quad\qquad\quad\qquad\quad\qquad\quad
\;\;\;\,\text{ for $(01;C_1,C_2;f)$  with $|\Im\,\phi_{2}|=m$}\big\}.\label{eq:def:hat-chi_m}
\end{aligned}
\end{equation}
We see that $\widehat{\chi}_m$ does not depend on the channel capacity constraints  $C_1$
and $C_2$, which shows that
\begin{equation}
\widehat{\chi}_m=\chi_m,\label{rel-hat-chi_m-chi_m}
\end{equation}
where we recall \eqref{eq:def:chi_m} for the definition of $\chi_m$.

Now, for the admissible $k$-shot code $\mbC=\big\{\phi_{1}, \phi_2;~\psi\big\}$ with $|\Im\,\phi_2|=m$, by the definition of $\widehat{\chi}_m$ in \eqref{eq:def:hat-chi_m} and the equality \eqref{rel-hat-chi_m-chi_m}, we obtain that
\begin{equation*}
|\Im\,\phi_1|\geq \widehat{\chi}_m=\chi_m.
\end{equation*}
 This implies that
 \begin{align}
		|\Im\,\phi_1|\geq\chi_m\geq Q_
		{k}\Big(\Big\lceil\frac{2^{k}}{m}\Big\rceil\Big)\geq 2^{k}\cdot h\Big( \Big\lceil\frac{2^{k}}{m}\Big\rceil\Big)\geq3^k\cdot m^{1-\log3},\label{eq:gener-dis-fun-pf_upp_bound-1}
	\end{align}
where the second inequality follows from  Lemma~\ref{lemma:N_relat_code_color_set}, the third inequality follows from Corollary~\ref{cor:Q_k^l_low_bound}, and the last inequality follows from
$h(\ell)=\ell^{\log3-1},~\ell=0,1,2,\cdots$  (cf.~\eqref{def:h(l)}).

With the above discussion, by  \eqref{dis-cod-def-n_i(C)} and \eqref{dis-cod-def-n(C)} we obtain that
 \begin{align}
n(\mbC)&=\max\bigg\{\bigg\lceil\frac{\log|\Im\,\phi_{1}|}{C_1}\bigg\rceil,~\bigg\lceil\frac{\log|\Im\,\phi_{2}|}{C_2}
\bigg\rceil\bigg\}\nonumber\\
&\geq\max\bigg\{\frac{\log|\Im\,\phi_{1}|}{C_1},~\frac{\log|\Im\,\phi_{2}|}{C_2}
\bigg\}\nonumber\\
&\geq\max\bigg\{\frac{k\log3+(1-\log3)\cdot\log m}{C_1},~\frac{\log m}{C_2}\bigg\},\label{eq:gener-dis-fun-pf_upp_bound-3}
\end{align}
where the inequality \eqref{eq:gener-dis-fun-pf_upp_bound-3} follows from \eqref{eq:gener-dis-fun-pf_upp_bound-1} and $|\Im\,\phi_2|=m$. Considering all $1\leq m\leq 2^k$ and continuing from \eqref{eq:gener-dis-fun-pf_upp_bound-3}, we further obtain that
  \begin{align}
  n(\mbC)&\geq\min_{\substack{m\in\{1,2,\cdots,2^k\}}}\max\bigg\{\frac{k\log3+(1-\log3)\cdot\log m}{C_1},~\frac{\log m}{C_2}\bigg\}\nonumber\\
  &\geq\min_{\substack{m\in[1,2^k]}}\max\bigg\{\frac{k\log3+(1-\log3)\cdot\log m}{C_1},~\frac{\log m}{C_2}\bigg\}\nonumber\\
  &=\min_{\substack{t\in[0,k]}}\max\bigg\{\frac{k\log3+(1-\log3) t}{C_1},~\frac{t}{C_2}\bigg\},\label{eq:gener-pf_upp_bound-4}
  \end{align}
  where in \eqref{eq:gener-pf_upp_bound-4}, we let $t=\log m$.
  By a calculation similar to the one for the special case $(01;2,1;f)$ in Section~\ref{subsec_pf_Thm4} (cf.~\eqref{eq:pf_upp_bound_log3(6)-7}-\eqref{eq:pf_upp_bound_log3(6)-9}), we can obtain that
\begin{equation*}
\min_{\substack{t\in[0,k]}}\max\bigg\{\frac{k\log3+(1-\log3) t}{C_1},~\frac{t}{C_2}\bigg\}=\frac{k\log3}{(C_1-C_2)+C_2\cdot\log 3}.
\label{eq:gener-pf_upp_bound-5}
\end{equation*}
This implies that
\begin{align*}
n(\mbC)\geq\frac{k\log3}{(C_1-C_2)+C_2\cdot\log 3},
\end{align*}
or equivalently,
\begin{equation}
\frac{k}{n(\mbC)}\leq\frac{(C_1-C_2)+C_2\cdot\log 3}{\log3}=(C_1-C_2)\log_{3}2+C_2.\label{eq:gener-pf_upp_bound-6}
\end{equation}
We note that the above upper bound \eqref{eq:gener-pf_upp_bound-6} is satisfied for each positive integer $k$ and each admissible $k$-shot code for the model $(01;C_1,C_2;f)$. Hence, we have proved the inequality \eqref{C1-C2_cap_upp_bound}.

 Next, we will prove that
\begin{equation}
 \mathcal{C}(01;C_1,C_2;f)\geq(C_1-C_2)\log_{3}2+C_2.\label{C_1-C2-Net_cap_low_bound}
\end{equation}
Let $k$ be  a positive integer.
 We construct an admissible $k$-shot code $\mbC=\big\{\phi_1,\phi_2;~\psi\big\}$ for the model $(01;C_1,C_2;f)$ as follows.
 Let
 \begin{equation*}
 \vec{x}=(x_1,x_2,\cdots,x_k)\in\mathcal{A}^k\quad \text{and}\quad \vec{y}=(y_1,y_2,\cdots,y_k)\in\mathcal{A}^k
 \end{equation*}
  be two arbitrary source messages generated by
 $\boldsymbol{X}$ and $\boldsymbol{Y}$, respectively. We further let
 \begin{equation}
 k_{1}\triangleq\bigg\lceil\frac{(C_1-C_2)\cdot k}{(C_1-C_2)+C_2\cdot\log3}\bigg\rceil.
 \label{eq:C1-C2-encod-hat{k}1-value}
 \end{equation}
 \begin{itemize}
 \item The  encoding function $\phi_{1}(\vec{x},\vec{y})$  is defined to be a $k$-vector:
\begin{equation}\label{eq:C1-C2-encod-fun-1}
 \phi_{1}(\vec{x},\vec{y})=
 \begin{bmatrix}
 x_{i}+y_{i}:&1\leq i\leq k_{1}-1\\
 x_{j}:&k_{1}\leq j\leq k
 \end{bmatrix},
% \left[\begin{matrix}\,x_{i}+y_{i}: &1\leq i\leq k_{1}-1,  \vspace{1mm} \\\,x_{i}:& k_{1}\leq i\leq k; \end{matrix}\right],
\end{equation}
and the encoding function $\phi_2(\vec{y})$ is defined to be another $k$-vector:
\begin{equation}\label{eq:C1-C2-encod-fun-2}
  \phi_{2}(\vec{y})=\begin{bmatrix}
  0: &1\leq i\leq k_{1}-1\\
  y_{j}: &k_{1}\leq j\leq k
  \end{bmatrix};
% \left[\begin{matrix}\,0: &1\leq i\leq k_{1}-1,  \vspace{1mm} \\y_{i}:& k_{1}\leq i\leq k; \end{matrix}\right];
\end{equation}
\item  The decoding function $\psi$ is defined by
\begin{equation*}
		\psi\big( \phi_{1}(\vec{x},\vec{y}),\phi_{2}(\vec{y})\big)
= \phi_{1}(\vec{x},\vec{y})+
\phi_{2}(\vec{y}).
\end{equation*}
\end{itemize}

Evidently, the code $\mbC$ is admissible.
It follows from \eqref{eq:C1-C2-encod-fun-1} and \eqref{eq:C1-C2-encod-fun-2} that
\begin{equation*}
|\Im\,\phi_{1}|=3^{k_{1}-1}\cdot 2^{k-k_{1}+1}\quad\text{and}\quad |\Im\,\phi_{2}|=2^{k-k_1+1},
\end{equation*}
or equivalently,
\begin{align*}
\log|\Im\,\phi_{1}|=(k_{1}-1)\log3+k-k_{1}+1\quad\text{and}\quad
\log|\Im\,\phi_{2}|=k-k_1+1.
\end{align*}
 Together with \eqref{eq:C1-C2-encod-hat{k}1-value}, we consider
\begin{align}
\log|\Im\,\phi_{1}|&=\bigg(\bigg\lceil\frac{(C_1-C_2)\cdot k}{(C_1-C_2)+C_2\cdot\log3}\bigg\rceil-1\bigg)\cdot\log3+\bigg( k-\bigg\lceil\frac{(C_1-C_2)\cdot k}{(C_1-C_2)+C_2\cdot\log3}\bigg\rceil+1\bigg)\nonumber\\
&\leq \frac{(C_1-C_2)\cdot k}{(C_1-C_2)+C_2\cdot\log3}\cdot\log3+\bigg( k-\frac{(C_1-C_2)\cdot k}{(C_1-C_2)+C_2\cdot\log3}+1 \bigg)\nonumber\\
&= \frac{k\cdot C_1 \cdot \log3}{(C_1-C_2)+C_2\cdot\log3}+1\label{C1-C2-N_theta_1_upper_bound},
\end{align}
and
\begin{align}
\log|\Im\,\phi_{2}|&=k-\bigg\lceil\frac{(C_1-C_2)\cdot k}{(C_1-C_2)+C_2\cdot\log3}\bigg\rceil+1 \leq k-\frac{(C_1-C_2)\cdot k}{(C_1-C_2)+C_2\cdot\log3}+1\nonumber\\
&= \frac{k\cdot C_2 \cdot \log3}{(C_1-C_2)+C_2\cdot\log3}+1\label{C1-C2-N_theta_2_upper_bound}.
\end{align}
Hence, we can obtain that
\begin{align}
n(\mbC)&=\max\bigg\{\bigg\lceil\frac{\log|\Im\,\phi_{1}|}{C_1}\bigg\rceil,~\bigg\lceil\frac{\log|\Im\,\phi_{2}|}{C_2}
\bigg\rceil\bigg\}\nonumber\\
&\leq\max\bigg\{\frac{\log|\Im\,\phi_{1}|}{C_1}+1,~\frac{\log|\Im\,\phi_{2}|}{C_2}+1
\bigg\}\nonumber\\
&\leq \max\bigg\{\frac{k\log3}{(C_1-C_2)+C_2\cdot\log3}+\frac{1}{C_1}+1,
~\frac{k\log3}{(C_1-C_2)+C_2\cdot\log3}+\frac{1}{C_2}+1\bigg\}\label{C1-C2-upp-b-pf-n(C)-upp-bound1}\\
&=\frac{k\log3}{(C_1-C_2)+C_2\cdot\log3}+\frac{1}{C_2}+1,\nonumber%\label{C1-C2-upp-b-pf-n(C)-upp-bound2}
\end{align}
where the inequality \eqref{C1-C2-upp-b-pf-n(C)-upp-bound1} follows from \eqref{C1-C2-N_theta_1_upper_bound} and \eqref{C1-C2-N_theta_2_upper_bound}.
 Then, the  compression rate of the  code  $\mbC$ satisfies
\begin{align}
R(\mbC)=\frac{k}{n(\mbC)} & \geq \left[\frac{\log3}{(C_1-C_2)+C_2\cdot\log3}+\frac{1}{k}\Big(\frac{1}{C_2}+1\Big)\right]^{-1}\nonumber\\
&\to (C_1-C_2)\log_{3}2+C_2, \quad \text{as} \quad k\to +\infty.\nonumber%\label{eq:C1-C2-constr-code-com-rate-lower-bound}
\end{align}
Thus, we have proved the inequality \eqref{C_1-C2-Net_cap_low_bound}. We have completed the proof.
\end{proof}

\section{An Application in Network Function Computation}\label{sec:equ-model}

An important application of the compression capacity for the model $(01;C_1,C_2;f)$ is in the tightness of the best known upper bound on the computing capacity in (zero-error) network function computation~\cite{Guang_Improved_upper_bound}. In network function computation, several ``general'' upper bounds on the computing capacity have been obtained \cite{Appuswamy11,Huang_Comment_cut_set_bound,Guang_Improved_upper_bound}, where ``general'' means that the upper bounds are applicable to arbitrary network and arbitrary target function. Here, the best known upper bound is the one proved by Guang~\emph{et al.}~\cite{Guang_Improved_upper_bound} in using the approach of the cut-set strong partition. This bound is always tight for all previously considered network function computation problems whose computing capacities are known. However, whether the upper bound is in general tight remains open~\cite{Guang_Improved_upper_bound}. In this section, we will give the answer that in general the upper bound of Guang~\emph{et al.}~\cite{Guang_Improved_upper_bound} is not tight by considering equivalent network function computation models of $(01;C_1,C_2;f)$.

\subsection{An Equivalent Model of Network Function Computation}\label{subsec:equ-model}

\begin{figure}
	\centering
\tikzstyle{format}=[draw,circle,fill=gray!30,minimum size=6pt, inner sep=0pt]
	\begin{tikzpicture}
		\node[format](a1)at(-2,5){};
		\node[format](a2)at(2,5){};
		\node[format](r1)at(-1,2.5){};
		\node[format](r2)at(1,2.5){};
		\node[format](p)at(0,0){};
        \draw[->,>=latex] (a1) to[bend right=50] node[midway, auto,swap, left=0mm] {\small{$d_{1}$}} (r1);
        \draw[->,>=latex] (a1) to node[above=1mm] {} (r1);
        \draw[->,>=latex] (a2) to[bend right=40] node[midway, auto, left=0mm] {} (r1);
        \draw[->,>=latex] (a2) to  node[midway, auto, below=0.8mm] {} (r1);
        \draw[->,>=latex] (a2) to  node [midway, auto, right=0mm] {} (r2);
        \draw[->,>=latex] (a2) to[bend left=50] node [midway, auto, right=0mm] {\small{$d_{6}$}} (r2);
        \draw[->,>=latex] (r1) to[bend right=50] node [below=0.1,left] {\small{$e_{1}$}} (p);
        \draw[->,>=latex] (r1) to  node [midway, auto, right=0mm] {\small{$e_{2}$}} (p);
		\draw[->,>=latex](r2)to node [midway, auto, right=0mm] {\small{$e_{3}$}}(p);
		\node at (-2,5.4){ $\sigma_{1}$};\node at (-1.3,4){ \small{$d_{2}$}};
        \node at (1.25,3.8){ \small{$d_{5}$}};
        \node at (0.3,3.3){ \small{$d_{4}$}};
        \node at (-0.2,4.5){ \small{$d_{3}$}};
		\node at (2,5.4){ $\sigma_{2}$};
        \node at (-0.6,2.4){ $v_{1}$};
		\node at (1.2,2.2){ $v_{2}$};
		\node at (0,-0.4){ $\rho$};
	\end{tikzpicture}
	\caption{The equivalent model $(\mathcal{N},f)$ of  zero-error network function computation. }
	\label{fig:N}
\end{figure}
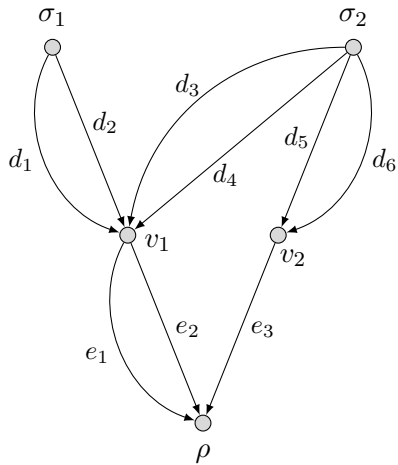

We also focus on the typical special case $(01;2,1;f)$ and transform it into an equivalent model of network function computation, which is specified below. Let $\mG=(\mV,\mE)$ be a directed acyclic graph in Fig.\,\ref{fig:N}, where the set of nodes and the set of edges are $\mathcal{V}=\{\sigma_1,\sigma_2,v_1,v_2,\rho\}$ and $\mathcal{E}=\{d_i:1\leq i\leq 6;~e_j,1\leq j\leq 3 \}$, respectively. We assume that a symbol taken from the set $\mA$ (i.e., a bit) can be reliably transmitted on each edge for each use, i.e., we take the capacity of each edge to be $1$ with respect to $\mA$.
In the set of nodes $\mathcal{V}$, $\sigma_1$ and $\sigma_2$ are two \emph{source nodes} that  model the sources $\boldsymbol{X}$ and $\boldsymbol{Y}$, respectively.  We let $S=\{\sigma_1, \sigma_2\}$, the set of sources nodes. The two intermediate nodes $v_1$ and $v_2$  model the two encoders $\Enone$ and $\Entwo$, respectively. The single \emph{sink node} $\rho$ models the decoder $\De$ to compress the target function $f$ with zero error.  Next, we consider the set of edges $\mathcal{E}$. In the model $(01;2,1;f)$, the channels
 $(\Enone,~\De)$ and $(\Entwo,~\De)$  have capacities $2$ and $1$, respectively. Accordingly, we set two parallel edges $e_1$ and $e_2$ from $v_1$ to $\rho$, and one edge $e_3$ from $v_2$ to $\rho$. Further, in the model  $(01;2,1;f)$,  $\Enone$ can observe the two sources $\boldsymbol{X}$ and $\boldsymbol{Y}$, and $\Entwo$ only observes the source  $\boldsymbol{Y}$. To show this in the equivalent model of network function computation, by $2=C_1>C_2=1$, it suffices to set two parallel edges $d_1$ and $d_2$ from $\sigma_1$ to $v_1$, and two  parallel edges $d_3$ and $d_4$
 from $\sigma_2$ to $v_1$; and two parallel edges $d_5$ and $d_6$
 from $\sigma_2$ to $v_2$. This will become clear later. In the graph $\mathcal{G}$,
for an edge $e \in \mE$, the \emph{tail} node and \emph{head} node of $e$ are denoted by $\tail(e)$ and $\head(e)$, respectively. For a node $u \in \mV$, we let $\ein(u)=\{ e \in \mE:\head(e)=u \}$ and $\eout(u)=\{ e \in \mE:\tail(e)=u \}$, both of which are the set of input edges and the set of output edges of $u$, respectively. The two source nodes have no input edges and the single sink node  has no output edges, i.e., $\ein(\sigma_i)=\eout(\rho)=\emptyset$, $i=1,2$.  The graph $\mG$, together with the set of the source nodes $S$ and the single sink node~$\rho$, forms a \emph{network} $\mN$, i.e., $\mN=(\mG, S, \rho)$.

Let
$k$ be a positive integer. The  source node $\sigma_1$ generates $k$ symbols $x_{1},x_{2},\cdots,x_{k}$  in $\mA$, written as a vector $\vec{x}$, i.e., $\vec{x}=(x_{1},x_{2},\cdots,x_{k})$; and
the  source node $\sigma_2$ generates $k$ symbols $y_{1},y_{2},\cdots,y_{k}$  in $\mA$, written as a vector $\vec{y}$, i.e., $\vec{y}=(y_{1},y_{2},\cdots,y_{k})$.
 The $k$ values of the target function
\begin{align}\label{function_f}
f(\vec{x},\vec{y})\triangleq\big(f(x_{i},y_{i}):~ i=1,2,\cdots,k \big)\nonumber
\end{align}
are required to  compute with zero error at $\rho$. We have completed the specification of the equivalent  model of network function computation, denoted by $(\mN,f)$, as depicted in Fig.\,\ref{fig:N}.

For the model $(\mathcal{N},f)$, we consider computing $f(\vec{x},\vec{y})$ for any source message pair $(\vec{x},\vec{y})\in\mA^{k}\times\mA^{k}$ by transmitting multiple symbols in $\mA$ on each edge in $\mE$. A \emph{$k$-shot (function-computing) network code} for $(\mathcal{N},f)$  consists of
\begin{itemize}
\item {\em a local encoding function} $\theta_{e}$ for each edge $e$ in $\mE$, where
\begin{equation}
  \theta_{e}:
  \begin{cases}
    \qquad \mA^k \rightarrow \Im\,\theta_e, & \text{if}~\mathrm{tail}(e)=\sigma_{i}~\text{for  $i=1,2$}; \\
    \prod\limits_{d\in \ein(\mathrm{tail}(e))} \Im\,\theta_d \rightarrow
    \Im\,\theta_e, & \text{otherwise;}
  \end{cases}\label{def:loc-encod-fun}
\end{equation}
with $\Im\,\theta_e$ being the image set of $\theta_e$;
and
 \item a decoding function  $\varphi:~\prod_{e\in \ein(\rho)} \Im\,\theta_e \rightarrow \Im f^{\,k}=\{0,1,2\}^k$ at the sink node $\rho$, which is used  to compute the target function $f$ with zero error.
\end{itemize}
With the causality of the encoding mechanism as specified in \eqref{def:loc-encod-fun}, the local encoding function $\theta_{e}$ for an output edge $e$ of a node depends  on the local encoding functions for the input edges of this node, so on and so forth.
Hence, we see that the message transmitted on each edge $e\in\mathcal{E}$ is a function of $(\vec{x},\vec{y})$, denoted by $g_{e}(\vec{x},\vec{y})$, which is obtained by recursively applying the local encoding
functions $\theta_{e},\,e\in\mathcal{E}$. To be specific, for each edge $e\in\mathcal{E}$,
\begin{equation}
  g_{e}(\vec{x},\vec{y})=
  \begin{cases}
    \quad \theta_{e}(\vec{x}), & \text{if}~\mathrm{tail}(e)=\sigma_{1}; \\
    \quad \theta_{e}(\vec{y}), & \text{if}~\mathrm{tail}(e)=\sigma_{2}; \\
    \quad\theta_{e}\big(g_{\mathcal{E}_{\text{i}}(u)}(\vec{x},\vec{y})\big), & \text{otherwise;}
  \end{cases}\label{def:glob-encod-fun}
\end{equation}
where $u=\text{tail}(e)$ and $g_{E}(\vec{x},\vec{y}) = \big(g_{e}(\vec{x},\vec{y}):~e\in E \big)$ for an edge subset $E\subseteq\mathcal{E}$ so that in \eqref{def:glob-encod-fun} we have $g_{\mathcal{E}_{\text{i}}(u)}(\vec{x},\vec{y})=\big(g_{e}(\vec{x},\vec{y}):~e\in\mathcal{E}_{{\rm i}}(u)\big)$. We call $g_{e}$ the \emph{global encoding function} for the edge $e$.
A $k$-shot network code $\mbC=\big\{\theta_e,e\in\mathcal{E};~\varphi\big\}$ is called {\em admissible} if the $k$ function values $f(\vec{x},\vec{y})$ for each pair of source messages $(\vec{x},\vec{y})\in\mA^k\times\mA^k$ can be computed with zero error at $\rho$ by using the code $\mbC$, i.e.,
\begin{equation}\label{def:decod-fun-zero-error}
\varphi\big(g_{\ein(\rho)}(\vec{x},\vec{y})\big)=f(\vec{x},\vec{y}),\quad\forall~(\vec{x},\vec{y}) \in \mA^k\times\mA^k.\nonumber
\end{equation}

For an admissible $k$-shot network code $\mbC=\big\{\theta_{e}:~e\in \mE;~\varphi \big\}$,
 we let $n_e\triangleq\big\lceil\log|\Im\,\theta_{e}|\big\rceil$ for each $e\in\mathcal{E}$, which is the number of bits transmitted on the edge $e$  in using the code $\mbC$. We further let
$n\triangleq \max_{e\in\mathcal{E}}n_e$, which can be regarded as the number of times the network $\mN$ is used to compute the function $f$ $k$ times in using the code $\mbC$. The {\em computing rate} of $\mbC$ is defined by $R(\mbC)\triangleq\frac{k}{n}$, i.e., the average number of times the function $f$ can be computed with zero error at $\rho$ for one use of the network $\mN$ by using the code $\mbC$.  Further, we say that a nonnegative real number $R$ is {\em (asymptotically) achievable} if $\forall~\epsilon > 0$, there exists an admissible $k$-shot network code $\mbC$ such that
$$R(\mbC) > R-\epsilon.$$
 Accordingly, the {\em computing rate region} for the  model $(\mN,f)$ is defined as
\begin{align}%\label{def:rate_region}
\mathfrak{R}(\mN,f) \triangleq \Big\{ R:~R \text{ is asymptotically achievable for $(\mN,f)$} \Big\},\nonumber
\end{align}
which is also closed and bounded. Then, the {\em computing capacity} for $(\mN, f)$ is defined as
\begin{align}%\label{defi_capacity}
\mC(\mN,f)\triangleq  \max~\mathfrak{R}(\mN, f).\nonumber
\end{align}

We remark that for the model $(\mathcal{N},f)$, the above definition of $k$-shot network codes is equivalent to the definition  in  \cite{Guang_Improved_upper_bound}. This is explained as follows. In \cite{Guang_Improved_upper_bound},  a coding scheme is said to be a $(k,n)$ network code for $(\mathcal{N},f)$, if this coding scheme can compute the target function $f$ $k$ times by transmitting at most $n$ symbols in $\mathcal{B}$ on each edge $e$  (cf.~\cite[Section~\Rmnum{2}-A]{Guang_Improved_upper_bound} for
the formal definition of a $(k,n)$ network code). The computing rate of a $(k,n)$ network code is defined by $k/n$.
Clearly,  a $(k,n)$ network code is a $k$-shot network code. On the other hand, by definition, a $k$-shot network code can be viewed as a  $(k,n)$ network code, where $n=\max_{e\in\mathcal{E}}\big\lceil\log|\Im\,\theta_{e}|\big\rceil$. This thus shows  that the two computing capacities induced by the two definitions of (function-computing) network codes are identical.

\subsection{The Tightness of the Best Known Upper Bound on Computing Capacity}

For the models $(\mathcal{N},f)$ and $(01;2,1;f)$, we will show that
\begin{equation}
\mathcal{C}(01;2,1;f)=\mathcal{C}(\mathcal{N},f).\nonumber
\end{equation}
First, we show that  for the model $(\mathcal{N},f)$, each admissible $k$-shot network code is equivalent to an admissible $k$-shot network code for which $\sigma_1$ transmits all the source symbols $x_1,x_2,\cdots,x_k$ to
$v_1$, and $\sigma_2$ transmits all the source symbols $y_1,y_2,\cdots,y_k$ to $v_1$ and $v_2$.
This immediately presents an equivalent  simplified  form of  admissible $k$-shot network codes for $(\mathcal{N},f)$.
  We consider an arbitrary admissible $k$-shot network code $\mbC=\big\{\theta_{e}:~e\in \mE;~\varphi \big\}$ for the model $(\mathcal{N},f)$, of which the computing rate is $k/n$ with $n=\max_{e\in\mathcal{E}}\big\lceil\log|\Im\,\theta_{e}|\big\rceil$.
  Further, we let $\big\{g_{e}:~e\in\mathcal{E}\big\}$ be the set of all the global encoding functions of the  code  $\mbC$. For any two different source messages $\vec{x}$ and $\vec{x}\,'$ in $\mathcal{A}^k$ and any source message $\vec{y}$ in $\mathcal{A}^k$ (which shows that $f(\vec{x},\vec{y})=\vec{x}+\vec{y}\neq\vec{x}\,'+\vec{y}= f(\vec{x}\,',\vec{y})$), we claim that
  \begin{equation}
  \big(g_{e_1}(\vec{x},\vec{y}),g_{e_2}(\vec{x},\vec{y})\big)\neq  \big(g_{e_1}(\vec{x}\,',\vec{y}),g_{e_2}(\vec{x}\,',\vec{y})\big),\label{eq:glob-fun-img-diff}
  \end{equation}
  because, otherwise, we have
  \begin{equation}
  \big(g_{e_1}(\vec{x},\vec{y}),g_{e_2}(\vec{x},\vec{y}),g_{e_3}(\vec{y})\big)= \big(g_{e_1}(\vec{x}\,',\vec{y}),g_{e_2}(\vec{x}\,',\vec{y}),g_{e_3}(\vec{y})\big),\nonumber
  \end{equation}
  which implies that
   \begin{equation}
  f(\vec{x},\vec{y})=\varphi\big(g_{e_1}(\vec{x},\vec{y}),g_{e_2}(\vec{x},\vec{y}),g_{e_3}(\vec{y})\big)= \varphi \big(g_{e_1}(\vec{x}\,',\vec{y}),g_{e_2}(\vec{x}\,',\vec{y}),g_{e_3}(\vec{y})\big)=f(\vec{x}\,',\vec{y}),
  \nonumber
  \end{equation}
  a contradiction. It thus follows from \eqref{eq:glob-fun-img-diff} that
  \begin{equation}
|\Im\,g_{e_1}|\cdot|\Im\,g_{e_2}|\geq\big|\Im\,(g_{e_1},g_{e_2})\big|\geq |\mathcal{A}|^k=2^k,\nonumber
\end{equation}
or equivalently,
\begin{equation}
\log|\Im\,g_{e_1}|+\log|\Im\,g_{e_2}|\geq k,\nonumber
\end{equation}
where
\begin{equation}
\Im\,(g_{e_1},g_{e_2})\triangleq\Big\{\big(g_{e_1}(\vec{x},\vec{y}),g_{e_2}(\vec{x},\vec{y})\big):~\forall~
(\vec{x},\vec{y})\in\mathcal{A}^k\times\mathcal{A}^k\Big\}.\nonumber
\end{equation}
This, together with the fact that $\Im\,\theta_{e}=\Im\,g_{e}$ for all $e\in\mathcal{E}$ by \eqref{def:glob-encod-fun}, implies that
\begin{equation}
2n\geq \log|\Im\,\theta_{e_1}|+\log|\Im\,\theta_{e_2}|=\log|\Im\,g_{e_1}|+\log|\Im\,g_{e_2}|\geq k,
\nonumber
\end{equation}
i.e., $n\geq k/2$, and thus $n\geq \lceil k/2\rceil$, where we recall that $n=\max_{e\in\mathcal{E}}\big\lceil|\log\Im\,\theta_e|\big\rceil$. Then, for each admissible $k$-shot network
code $\mbC$, at least $\lceil k/2 \rceil$ bits can be transmitted on each edge.
In particular, on each edge $d_{i},~1\leq i\leq 6$, at least $\lceil k/2 \rceil$ bits can be transmitted.
Thus, we assume without loss of generality that the node $v_1$ obtains all the $k$ source symbols $x_{1},x_{2},\cdots,x_{k}$ of $\sigma_1$ (i.e., $\vec{x}$) from the two input edges $d_1$ and $d_2$ and all the~$k$ source symbols~$y_{1},y_{2},\cdots,y_{k}$  of $\sigma_2$ (i.e., $\vec{y}$) from the two input edges $d_3$ and $d_4$; and the node $v_2$ obtains all the~$k$ source symbols $y_{1},y_{2},\cdots,y_{k}$   of $\sigma_2$ (i.e., $\vec{y}$)  from the two input edges $d_5$ and $d_6$.
   With the assumption, the inputs of the local encoding functions $\theta_{e_1}$ and $\theta_{e_2}$ are all possible pairs $(\vec{x},\vec{y})\in\mathcal{A}^k\times\mathcal{A}^k$, or equivalently,~$\theta_{e_1}$ and $\theta_{e_2}$ can be viewed as two functions from $\mathcal{A}^k\times\mathcal{A}^k$ to $\Im\,\theta_{e_1}$ and $\Im\,\theta_{e_2}$, respectively, denoted by $\theta_{e_1}(\vec{x},\vec{y})$ and $\theta_{e_2}(\vec{x},\vec{y})$. Similarly, the input of the local encoding function $\theta_{e_3}$ is all possible vectors $\vec{y}\in\mathcal{A}^k$ so that $\theta_{e_3}$ can be viewed as a function from $\mathcal{A}^k$ to $\Im\,\theta_{e_3}$, denoted by~$\theta_{e_3}(\vec{y})$.
  Hence, for an admissible $k$-shot network code $\mbC$,
 it suffices to consider the   local encoding functions~$\theta_{e_1}(\vec{x},\vec{y})$,~$\theta_{e_2}(\vec{x},\vec{y})$ and $\theta_{e_3}(\vec{y})$, which are written as $\theta_{1}$, $\theta_{2}$ and $\theta_{3}$, respectively, for notational simplicity. In the rest of the paper, we will use $\big\{\theta_1,\theta_2,\theta_3;~\varphi\big\}$ to stand for an admissible $k$-shot network code for the model~$(\mathcal{N},f)$.

 Based on the above discussion, we readily see that an admissible $k$-shot network code $\big\{\theta_1,\theta_2,\theta_3;~\varphi\big\}$ for the model $(\mathcal{N},f)$ can be viewed as an admissible $k$-shot code $\big\{\phi_{1},\phi_{2};~\psi\big\}$ for the model $(01;2,1;f)$ by setting $\phi_1=(\theta_1,\theta_2)$, $\phi_2=\theta_3$ and $\psi=\varphi$; and vice versa. We thus show
$\mathcal{C}(01;2,1;f)=\mathcal{C}(\mathcal{N},f)$.
%\subsection{The Characterization of the Computing Capacity $\mathcal{C}(\mathcal{N},f)$}
By Theorem~\ref{thm:01;2,1-comp-cap}, we immediately present the computing capacity for the model $(\mathcal{N},f)$  below.

\begin{theorem}\label{thm:N_cap}
The computing capacity of the model $(\mathcal{N},f)$ is $\log_{3}6$, i.e.,
\begin{equation*}
\mathcal{C}(\mathcal{N},f)=\log_{3} 6.
\end{equation*}
\end{theorem}

Further, the best known upper bound proved by Guang~\emph{et al.}~\cite{Guang_Improved_upper_bound} on the computing capacity $\mathcal{C}(\mathcal{N},f)$ is $\log_{3}8$, which is specified in Appendix~\ref{appendix:GYYB-bound-intro_and_com}. Thus, we have $\mathcal{C}(\mathcal{N},f)=\log_{3}6<\log_{3}8$, which immediately shows that in general the best known upper bound is not tight.

Similar to the equivalence of the models $(01;2,1;f)$ and $(\mathcal{N},f)$ as discussed above, we can also transform the model $(01;C_1,C_2;f)$ into an equivalent model of network function computation $\big(\mathcal{N}(C_1,C_2),f\big)$. The set of nodes is still $\mathcal{V}=\{\sigma_1,\sigma_2,v_1,v_2,\rho\}$, where, similarly, the two \emph{source nodes}~$\sigma_1$ and $\sigma_2$ are used to model the sources $\boldsymbol{X}$ and $\boldsymbol{Y}$, respectively; the two intermediate nodes $v_1$ and $v_2$ are used to model the two encoders $\Enone$ and $\Entwo$, respectively; and the single \emph{sink node} $\rho$ is used to model the decoder $\De$. There are $C_1$ edges from $v_1$ to $\rho$ and $C_2$ edges from $v_2$ to $\rho$; and $C_1$ edges from $\sigma_1$ to $v_1$, from $\sigma_2$ to $v_1$, and from $\sigma_2$ to $v_2$. We also can obtain that
 \begin{equation}
 \mathcal{C}\big(\mathcal{N}(C_1,C_2),f\big)=\mathcal{C}(01;C_1,C_2;f)=(C_1-C_2)\log_{3}2+C_2.\label{gen-net-mod-cap}
 \end{equation}
In addition, the upper bound of Guang~\emph{et al.} \cite{Guang_Improved_upper_bound} on the computing capacity $\mathcal{C}\big(\mathcal{N}(C_1,C_2),f\big)$ can be calculated as follows:
\begin{equation}
\mathcal{C}\big(\mathcal{N}(C_1,C_2),f\big)\leq
\begin{cases}
 C_1,\quad&\text{if}~C_2\leq C_1\leq \dfrac{C_2}{\log3-1};\\
 \dfrac{C_1+C_2}{\log 3},\quad&\text{if}~C_1> \dfrac{C_2}{\log3-1}.
\end{cases}\label{C1-C2-Guang-upper-bound-value}
\end{equation}
Comparing the above upper bound \eqref{C1-C2-Guang-upper-bound-value} with the computing capacity
$C\big(\mathcal{N}(C_1,C_2),f\big)$ in \eqref{gen-net-mod-cap},  we thus show that the upper bound of Guang~\emph{et al.} \cite{Guang_Improved_upper_bound} is not tight for the model  $\big(\mathcal{N}(C_1,C_2),f\big)$ with $C_1>C_2$.

We finally remark that it was also claimed in \cite{Wang-Tightness_upper_bound} that the upper bound of Guang~\emph{et al.} \cite{Guang_Improved_upper_bound} is not tight for a specific model. However, the proof in \cite{Wang-Tightness_upper_bound} is invalid. Specifically, in \cite[Lemma~1]{Wang-Tightness_upper_bound}, it was claimed that $\gamma_{3^{k}-1}\geq 4^{k}-2^{k-1}$ for any positive integer $k>1$, which is incorrect. This thus causes their upper bound to be invalid. In fact, it can be proved that $\gamma_{3^{k}-1}=3\cdot 2^{k-1}<4^{k}-2^{k-1}$, $\forall~k>1$. Following the same argument with this correct equation, we can obtain the valid upper bound $1$, the same as the best known upper bound in \cite{Guang_Improved_upper_bound}.

}

%\section{Proof of Theorem~\ref{thm:N_cap}}\label{subsec:cap-pf}

\begin{journalonly}

\begin{figure}
	\centering
\tikzstyle{format}=[draw,circle,fill=gray!30,minimum size=6pt, inner sep=0pt]
	\begin{tikzpicture}
		\node[format](a1)at(-3,3.5){};
		\node[format](a2)at(0,3.5){};
		\node[format](a3)at(3,3.5){};
		\node[format](r1)at(-1.5,1.75){};
		\node[format](r2)at(1.5,1.75){};
		\node[format](p)at(0,0){};
		\draw[->,>=latex](a1)--(r1) node[midway, auto,swap, left=0mm] {\small{$d_{1}$}};
		\draw[->,>=latex](a3)--(r1) ;
		\draw[->,>=latex](a2)--(r1) node[midway, auto,swap, left=0mm] {\small{$d_{2}$}};
		\draw[->,>=latex](a2)--(r2) ;
		\draw[->,>=latex](r1)--(p)  node[midway, auto,swap, left=0mm] {\small{$e_{1}$}};
		\draw[->,>=latex](r2)--(p)  node[midway, auto,swap, right=0mm] {\small{$e_{2}$}};
		\draw[->,>=latex](a3)--(r2) node[midway, auto,swap, right=0mm] {\small{$d_{5}$}};
		\node at (-3,3.9){ $\sigma_{1}$};
		\node at (0,3.9){ $\sigma_{2}$};
		\node at (3,3.9){ $\sigma_{3}$};
		\node at (-1.75,1.45){ $v_{1}$};
		\node at (1.75,1.45){ $v_{2}$};
		\node at (0,-0.4){ $\rho$};
         \node at(0.52,3.22)  {\small{$d_{3}$}};
         \node at(1.6,3.22)  {\small{$d_{4}$}};
	\end{tikzpicture}
	\caption{The  $(\widehat{\mathcal{N}},\widehat{f})$, where $\widehat{\mathcal{N}}=\big(\widehat{\mathcal{G}},\widehat{S}=\{\sigma_1,\sigma_2,\sigma_3\},\rho\big)$
and $\widehat{f}$ is the binary arithmetic sum $\widehat{f}(x_{1},x_{2},x_3)=x_{1}+x_{2}+x_{3}$.  }
	\label{fig:N_1}
\end{figure}

{\color{blue}
\section{The Issues in Literature~\cite{Wang-Tightness_upper_bound}}\label{sec:issues-Wang}

In \cite{Wang-Tightness_upper_bound}, Wang \emph{et al.} considered a
specific network function computation   model  $(\widehat{\mathcal{N}},\widehat{f})$ as depicted in Fig.\,\ref{fig:N_1}.
The network $\widehat{\mathcal{N}}=\big(\widehat{\mathcal{G}},\widehat{S},\rho\big)$ consists of the directed acyclic graph  $\widehat{\mathcal{G}}$ with the set of nodes $\widehat{\mathcal{V}}=\{\sigma_1,\sigma_2,\sigma_3,v_1,v_2,\rho\}$ and the set of edges $\widehat{\mathcal{E}}=\{d_i:1\leq i\leq 5;~e_1,e_2\}$, the set of three source nodes
$\widehat{S}=\{\sigma_{1},\sigma_{2},\sigma_{3}\}$ and the single sink node $\rho$. The target function $\widehat{f}$ is specified to be the  \emph{binary arithmetic sum}
\begin{equation*}
\widehat{f}(x_{1},x_{2},x_3)=x_{1}+x_{2}+x_{3},~~ x_{1},x_{2},x_{3}\in\{0,1\}.
\end{equation*}
($\mathcal{A}_{1}=\mathcal{A}_{2}=\mathcal{A}_{3}=\{0,1\}\triangleq\mathcal{A}$ and $\mathcal{O}=\{0,1,2,3\}$). Furthermore,  $\mathcal{B}=\{0,1\}$, i.e., a bit can be transmitted on each edge reliably  for each use.
 The technique based on graph coloring and linear programming  is used to characterize the corresponding
computing capacity
\begin{equation}\mathcal{C}(\widehat{\mathcal{N}},\widehat{f})=\frac{1}{2}\log_{3}6<1,
 \end{equation}
 where   the upper bound of Guang~\emph{et al.}~\cite{Guang_Improved_upper_bound} on $\mathcal{C}(\widehat{\mathcal{N}},\widehat{f})$ is equal to $1$. Wang \emph{et al.}~\cite{Wang-Tightness_upper_bound} thus proved  the non-tightness of the upper bound of Guang~\emph{et al.} \cite{Guang_Improved_upper_bound}.

However, in this section, we will show the invalidity of the arguments.
Specifically, Section~\ref{subsec:Wang_tech-intro} briefly introduces the technique   to characterize
the computing capacity $\frac{1}{2}\log_{3}6$ of the model  $(\widehat{\mathcal{N}},\widehat{f})$. Section~\ref{subsec:critical_error} gives a counterexample  to illustrate the critical error in this technique. Section~\ref{subsec:corr_tech} further prove that even if this error is corrected, similarly by using the technique proposed by \cite{Wang-Tightness_upper_bound}, we can only obtain  the same upper bound $1$ (instead of $\frac{1}{2}\log_{3}6$) as the upper bound of Guang~\emph{et al.} \cite{Guang_Improved_upper_bound}.

\subsection{The Technique of Characterizing the Computing Capacity}\label{subsec:Wang_tech-intro}

In this subsection, we  will briefly introduce  this technique proposed by \cite{Wang-Tightness_upper_bound}. Following \cite{Wang-Tightness_upper_bound},
we first present a simplified  form of  admissible $k$-shot network codes for the model $(\widehat{\mathcal{N}},\widehat{f})$ in Fig.\,\ref{fig:N_1}, and give some necessary definitions.

 We consider an arbitrary admissible $k$-shot network code $\mbC=\big\{\theta_{e}:~e\in \widehat{\mE};~\varphi \big\}$ for the model $(\widehat{\mathcal{N}},\widehat{f})$, whose the computing  rate $k/n$ with $n=\max_{e\in\widehat{\mathcal{E}}}\big\lceil\log|\Im\,\theta_{e}|\big\rceil$. By  $\mathcal{C}(\widehat{\mathcal{N}},\widehat{f})\leq 1$ (the upper bound of Guang~\emph{et al.}~\cite{Guang_Improved_upper_bound}),  we have $k/n\leq\mathcal{C}(\widehat{\mathcal{N}},\widehat{f})\leq 1$.  This immediately implies  $n\geq k$. Thus, for each admissible $k$-shot network
code $\mbC$, at least $k$ bits can be transmitted on each edge, or equivalently, the network can be used $k$ times at least.
In particular, on each edge $d_{i},~1\leq i\leq 5$, at least $k$ bits can be transmitted.
 Thus, we assume without loss of generality that the node $v_1$ obtains all the $k$ source symbols $x_{1,1},x_{1,2},\cdots,x_{1,k}$ of $\sigma_1$ (i.e., $\vec{x}_1$) from the input edge $d_1$, all the $k$ source symbols $x_{2,1},x_{2,2},\cdots,x_{2,k}$  of $\sigma_2$ (i.e., $\vec{x}_2$) from the  input edge $d_2$ and all the $k$ source symbols $x_{3,1},x_{3,2},\cdots,x_{3,k}$  of $\sigma_3$ (i.e., $\vec{x}_3$) from the  input edge $d_4$; and the node $v_2$ also obtains all the $k$ source symbols $x_{2,1},x_{2,2},\cdots,x_{2,k}$   of $\sigma_2$ (i.e., $\vec{x}_2$)  from the input edge $d_3$ and all the $k$ source symbols $x_{3,1},x_{3,2},\cdots,x_{3,k}$  of $\sigma_3$ (i.e., $\vec{x}_3$) from the  input edge $d_5$.
 With the assumption, the input of the local encoding function $\theta_{e_1}$ is all possible pairs $(\vec{x}_1,\vec{x}_2,\vec{x}_3)\in\mathcal{A}^k\times\mathcal{A}^k\times\mathcal{A}^k$, or equivalently, $\theta_{e_1}$ can be viewed as a function from $\mathcal{A}^k\times\mathcal{A}^k\times\mathcal{A}^k$ to $\Im\,\theta_{e_1}$, denoted by $\theta_{e_1}(\vec{x}_1,\vec{x}_2,\vec{x}_3)$. Similarly, the input of the local encoding function $\theta_{e_2}$ is all possible pairs $(\vec{x}_2,\vec{x}_3)\in\mathcal{A}^k\times\mathcal{A}^k$, i.e., $\theta_{e_3}$ can be viewed as a function from $\mathcal{A}^k\times\mathcal{A}^k$ to $\Im\,\theta_{e_2}$, denoted by $\theta_{e_2}(\vec{x}_2,\vec{x}_3)$.
In fact, the edge subset $C=\{e_1,e_2\}$ is the bottleneck on the computing capacity.
 Then, it is sufficient to consider the  local encoding functions $\theta_{e_1}(\vec{x}_1,\vec{x}_2,\vec{x}_3)$ and $\theta_{e_2}(\vec{x}_2,\vec{x}_3)$,
  denoted by $\theta_{\rm{\Rmnum{1}}}(\vec{x}_1,\vec{x}_2,\vec{x}_3)$ and $\theta_{\rm{\Rmnum{2}}}(\vec{x}_2,\vec{x}_3)$, respectively, for notational simplicity. On the other hand, because we need to compute the arithmetic sum  at the sink node $\rho$, it is not necessary to distinguish $(\vec{x}_{2},\vec{x}_{3})$ and
$(\vec{x}\,'_{2},\vec{x}\,'_{3})$ if $\vec{x}_{2}+\vec{x}_{3}=
\vec{x}_{2}'+\vec{x}_{3}'$.
Then, we further assume that
\begin{equation*}\theta_{\rm{\Rmnum{1}}}(\vec{x}_{1},\vec{x}_{2},\vec{x}_{3})=\theta_{\rm{\Rmnum{1}}}(\vec{x}_{1},\vec{x}\,'_{2},\vec{x}\,'_{3})
\quad\text{and}\quad\theta_{\rm{\Rmnum{2}}}(\vec{x}_{2},\vec{x}_{3})=\theta_{\rm{\Rmnum{2}}}(\vec{x}\,'_{2},\vec{x}\,'_{3})
\end{equation*}
 whenever $\vec{x}_{2}+\vec{x}_{3}=
\vec{x}_{2}'+\vec{x}_{3}'$. Similarly, this assumption has no penalty either on the computing capacity. For notational simplicity, in the rest of the paper, we  let
\begin{equation}
\vec{x}=\vec{x}_{1}\in\{0,1\}^{k}\quad\text{ and}\quad \vec{y}=\vec{x}_{2}+\vec{x}_{3}\in\{0,1,2\}^{k},
\end{equation} and write $\theta_{\rm{\Rmnum{1}}}(\vec{x},\vec{y})$ and $\theta_{\rm{\Rmnum{2}}}(\vec{y})$ replacing
$\theta_{\rm{\Rmnum{1}}}(\vec{x}_{1},\vec{x}_{2},
\vec{x}_{3})$  and $\theta_{\rm{\Rmnum{2}}}(\vec{x}_{2},\vec{x}_{3})$, respectively.
The sink node $\rho$ requires to solve $\vec{x}+\vec{y}$. Further, we will use $\{\theta_{\rm{\Rmnum{1}}},\theta_{\rm{\Rmnum{2}}};~\varphi\}$ to stand for an admissible $k$-shot network code for the model $(\widehat{\mathcal{N}},\widehat{f})$.

We let $k$ be a positive integer.
For convenience, we let $\mathcal{Y}=\{0,1,2\}$.
For a subset $\mathcal{M}\subseteq \mathcal{A}^{k}$ and a subset $\mathcal{L}\subseteq \mathcal{Y}^{k}$, we let
\begin{equation}
\mathcal{M}+\mathcal{L}\triangleq\{\vec{x}+\vec{y}:\text{$\vec{x}\in\mathcal{M}$ and $\vec{y}\in\mathcal{L}$ }\}.\label{eq:pf-W-def-set-M+L}
\end{equation}
We consider an arbitrary admissible $k$-shot network code $\{\theta_{\rm{\Rmnum{1}}},\theta_{\rm{\Rmnum{2}}};~\varphi\}$ for the model $(\widehat{\mathcal{N}},\widehat{f})$. By the discussion in the paragraph immediately below Section~\ref{subsec:Wang_tech-intro}, we note that the domain of $\theta_{\rm{\Rmnum{2}}}$ is $\mathcal{Y}^{k}$ and thus $1\leq |\Im\,\theta_{\rm{\Rmnum{2}}}|\leq |\mathcal{Y}^k|=3^{k}$. Then for any positive integer $m$ with $1\leq m\leq 3^k$,  we let
\begin{equation}
\begin{aligned}
\gamma_{m}\triangleq\min\big\{|\Im\,\theta_{\rm{\Rmnum{1}}}|:\text{$\{\theta_{\rm{\Rmnum{1}}},\theta_{\rm{\Rmnum{2}}};~\varphi\}$ is an admissible $k$-shot }\big.\\\big.\text{network code with $|\Im\,\theta_{\rm{\Rmnum{2}}}|=m$}\big\}.\label{eq:def:gamma_m}
\end{aligned}
\end{equation}
We remark that for any local encoding function $\theta_{\rm{\Rmnum{2}}}$ with $1\leq |\Im\,\theta_{\rm{\Rmnum{2}}}|=m\leq 3^k$, there always exists a local encoding function $\theta_{\rm{\Rmnum{1}}}$ such that $\{\theta_{\rm{\Rmnum{1}}},\theta_{\rm{\Rmnum{2}}};~\varphi\}$  is an admissible $k$-shot network code, and
thus $\gamma_m$ is well-defined for each $m$ with $1\leq m\leq 3^k$. To be specific, for any $(\vec{x},\vec{y})\in\mathcal{A}^k\times\mathcal{Y}^k$, we let
\begin{equation}
\theta_{\rm{\Rmnum{1}}}(\vec{x}, \vec{y})=\vec{x}+\vec{y},\quad\forall~(\vec{x},\vec{y})\in\mathcal{A}^k\times\mathcal{Y}^k
\end{equation}
regardless of $\theta_{\rm{\Rmnum{2}}}$. We readily see that this $k$-shot network code $\{\theta_{\rm{\Rmnum{1}}},\theta_{\rm{\Rmnum{2}}};~\varphi\}$ can
compute $\vec{x}+\vec{y}$ with zero error by setting  the decoding function $\varphi(\theta_{\rm{\Rmnum{1}}},\theta_{\rm{\Rmnum{2}}})=\theta_{\rm{\Rmnum{1}}}$.
This shows that the code $\{\theta_{\rm{\Rmnum{1}}},\theta_{\rm{\Rmnum{2}}};~\varphi\}$ is admissible.

Next we will present an equivalent form of $\gamma_m$ ($1\leq m\leq 3^k$), which will be specified in the following \eqref{eq:def:gamma_m}. Before that, we first give the definitions of a conflict graph $S$ and a class of subgraphs of $S$.

\begin{definition}\label{def:conf-graph-S^k}
The conflict graph of $\mathcal{A}^k+\mathcal{Y}^k$, denoted by  $S$, is defined as follows: there are $|\mathcal{A}|^k\cdot|\mathcal{Y}|^k$ vertices in the graph labeled by all the pairs $(\vec{x}, \vec{y})$ in $\mathcal{A}^k\times\mathcal{Y}^k$, and there exists an undirected edge between two vertices $(\vec{x}, \vec{y})$ and $(\vec{x}\,',\vec{y}\,')$ if and only if $\vec{x}+\vec{y}\neq \vec{x}\,'+\vec{y}\,'$.
\end{definition}

 We can readily see that for an admissible $k$-shot network  code $\{\theta_{\rm{\Rmnum{1}}},\theta_{\rm{\Rmnum{2}}};~\varphi\}$,  $(\theta_{\rm{\Rmnum{1}}},\theta_{\rm{\Rmnum{2}}})$ is a coloring of the conflict graph $S$, because for any two adjacent vertices $(\vec{x},\vec{y})$ and $(\vec{x}\,',\vec{y}\,')$ (i.e., $\vec{x}+\vec{y}\neq\vec{x}+\vec{y}\,'$), we have
\begin{equation*}
\big(\theta_{\rm{\Rmnum{1}}}(\vec{x},\vec{y}), \theta_{\rm{\Rmnum{2}}}(\vec{y})\big)\neq \big(\theta_{\rm{\Rmnum{1}}}(\vec{x}\,',\vec{y}\,'), \theta_{\rm{\Rmnum{2}}}(\vec{y}\,'
)\big).
\end{equation*}
For a subset $\mathcal{L}\subseteq\mathcal{Y}^k$, denote by  $S(\mathcal{L})$ the  subgraph of $S$, where all the vertices in $S(\mathcal{L})$  are labeled by the pairs $(\vec{x}, \vec{y})$ with $
\vec{x}\in\mathcal{A}^k$ and $\vec{y}\in \mathcal{L}$, and there exists an undirected edge between two different vertices $(\vec{x}, \vec{y})$ and $(\vec{x}\,',\vec{y}\,')$ in $S(\mathcal{L})$ if and
only if $\vec{x}+\vec{y}\neq \vec{x}\,'+\vec{y}\,'$. Further, we denote by $\gamma\big(S(\mathcal{L})\big)$  the minimum chromatic number of the subgraph $S(\mathcal{L})$.
Then, as mentioned in \cite{Wang-Tightness_upper_bound}, it is not difficult to see that for any positive integer $m$ with $1\leq m\leq 3^k$,
\begin{equation}
	\gamma_{m}=\min_{\substack{\text{all partitions }\\ \text{$\{P_1,P_2,\cdots,P_m\}$ of $\mathcal{Y}^k$}}}\max_{1\leq i\leq m}\gamma\big(S(P_i)\big).\label{eq:gamma_m=min_max_color}
\end{equation}

Let $\mathcal{C}(k)$ be the maximum achievable rate  of  all admissible $k$-shot network codes and their linear combinations (time sharing schemes). By the definition of $\gamma_{m}$ (cf.~\eqref{eq:def:gamma_m}), we immediately obtain that $\mathcal{C}(k)=\frac{k}{\,\widehat{n}\,}$, where $\widehat{n}$ is derived by the following optimization problem:
 \begin{align}
\widehat{n}=&\min_{\lambda_{1},\cdots,\lambda_{3^{k}}}\max\Bigg\{\sum_{m=1}^{3^{k}}\lambda_{m}\log\gamma_{m},
\sum_{m=1}^{3^{k}}\lambda_{m}\log m\Bigg\}\nonumber\\
&\text{s.t.}\;\sum_{m=1}^{3^{k}}\lambda_{m}=1;\;\;\lambda_{m}\geq 0, ~\text{for}~1\leq m\leq 3^{k}.\label{time-shar-opti-pro-C(k)}
 \end{align}
 Equivalently, we can derive $\mathcal{C}(k)$ by solving the \emph{optimization problem}:
\begin{equation}
	\begin{split}
		\max_{\alpha_{1},\ldots,\alpha_{3^{k}}}k\sum_{m=1}^{3^{k}}\alpha_{m}\;\;
		\text{s.t.}&\sum_{m=1}^{3^{k}}\alpha_{m}(\log\gamma_{m},\log m)\leq (1,1);\\
		&\alpha_{m}\geq 0\;\;\text{for}\;\; m= 1,2,\ldots,3^{k}.
	\end{split}\label{eq:opti_problem_C_k}
\end{equation}

{\color{red}
Now, we prove the equivalence of the above two optimization problems, i.e.,
$\frac{k}{\,\,\widehat{n}\,\,}=\widehat{r}$, where we denote by $\widehat{r}$ the optimal value of the problem \eqref{eq:opti_problem_C_k}.  We first prove that
\begin{equation}
\widehat{r}\geq \frac{k}{\,\,\widehat{n}\,\,}.
\end{equation}
Let $\widehat{\lambda}_1,\widehat{\lambda}_2,\cdots,\widehat{\lambda}_{3^k}$ be an optimal solution achieving the optimal value $\widehat{n}$ of the problem \eqref{time-shar-opti-pro-C(k)}. Clearly, we have $\sum_{m=1}^{3^{k}}\widehat{\lambda}_{m}=1$ and
\begin{equation}
\max\Bigg\{\sum_{m=1}^{3^{k}}\widehat{\lambda}_{m}\log\gamma_{m},
\sum_{m=1}^{3^{k}}\widehat{\lambda}_{m}\log m\Bigg\}=\widehat{n}.\label{opti-pro-eqv-pf-eq1}
\end{equation}
It follows from \eqref{opti-pro-eqv-pf-eq1} that
\begin{equation}
\sum_{m=1}^{3^{k}}\widehat{\lambda}_{m}\big(\log\gamma_{m},
\log m\big)\leq\big(\widehat{n},\widehat{n}\big),
\end{equation}
or equivalently,
\begin{equation}
\sum_{m=1}^{3^{k}}\frac{\widehat{\lambda}_{m}}{\widehat{n}}\big(\log\gamma_{m},
\log m\big)\leq(1,1).\label{opti-pro-eqv-pf-eq2}
\end{equation}
By \eqref{opti-pro-eqv-pf-eq2}, we can readily see that $$\frac{\widehat{\lambda}_{1}}{\,\,\widehat{n}\,\,},\frac{\widehat{\lambda}_{2}}{\,\,\widehat{n}\,\,},\cdots,
\frac{\widehat{\lambda}_{3^k}}{\,\,\widehat{n}\,\,}$$ is a solution of the problem \eqref{eq:opti_problem_C_k}, in which we regard $\frac{\widehat{\lambda}_{m}}{\widehat{n}}$ as $\alpha_m$ in \eqref{eq:opti_problem_C_k} for $1\leq m\leq 3^k$. This immediately follows that
\begin{equation}
\widehat{r}\geq k\sum_{m=1}^{3^{k}}\frac{\widehat{\lambda}_{m}}{\widehat{n}}=\frac{k}{\,\,\widehat{n}\,\,},
\label{opti-pro-eqv-pf-eq3}
\end{equation}
where the equality in \eqref{opti-pro-eqv-pf-eq3} follows from that $\sum_{m=1}^{3^{k}}\widehat{\lambda}_{m}=1$.

Next, we prove that
\begin{equation}
\widehat{r}\leq \frac{k}{\,\,\widehat{n}\,\,}.
\end{equation}
Let  $\widehat{\alpha}_1,\widehat{\alpha}_2,\cdots,\widehat{\alpha}_{3^k}$ be an optimal solution achieving the optimal value $\widehat{r}$ of the problem \eqref{eq:opti_problem_C_k}.
Clearly, we have
\begin{equation}
\sum_{m=1}^{3^{k}}\widehat{\alpha}_{m}(\log\gamma_{m},\log m)\leq (1,1),\label{opti-pro-eqv-pf-eq4}
\end{equation}
and
\begin{equation}
\widehat{r}=k\sum_{m=1}^{3^{k}}\widehat{\alpha}_{m}.\label{opti-pro-eqv-pf-eq5}
\end{equation}
We further
let \begin{equation}
\lambda_{m}=\frac{\widehat{\alpha}_{m}}{\sum_{i=1}^{3^{k}}\widehat{\alpha}_{i}},\quad 1\leq m\leq 3^k.
\end{equation}
First, we can readily see that $\lambda_1,\lambda_2,\cdots,\lambda_{3^k}$
is a solution of the problem \eqref{time-shar-opti-pro-C(k)}, because
\begin{equation}
\sum_{m=1}^{3^k}\lambda_{m}=\sum_{m=1}^{3^k}\frac{\widehat{\alpha}_{m}}{\sum_{i=1}^{3^{k}}\widehat{\alpha}_{i}}=1.
\end{equation}
It further follows from \eqref{opti-pro-eqv-pf-eq4} that
\begin{align}
\sum_{m=1}^{3^{k}}\lambda_{m}(\log\gamma_{m},\log m)&=
\sum_{m=1}^{3^{k}}\frac{\widehat{\alpha}_{m}}{\sum_{i=1}^{3^{k}}\widehat{\alpha}_{i}}(\log\gamma_{m},\log m)\\&=
\frac{1}{\sum_{i=1}^{3^{k}}\widehat{\alpha}_{i}}\cdot\sum_{m=1}^{3^{k}}\widehat{\alpha}_{m}
(\log\gamma_{m},\log m)\nonumber\\
&\leq\frac{1}{\sum_{i=1}^{3^{k}}\widehat{\alpha}_{i}}\cdot (1,1)
= \Big(\frac{1}{\sum_{i=1}^{3^{k}}\widehat{\alpha}_{i}},\frac{1}{\sum_{i=1}^{3^{k}}\widehat{\alpha}_{i}}\Big),\label{opti-pro-eqv-pf-eq6}
\end{align}
or equivalently,
\begin{equation}
\max\Bigg\{\sum_{m=1}^{3^{k}}\lambda_{m}\log\gamma_{m},
\sum_{m=1}^{3^{k}}\lambda_{m}\log m\Bigg\}\leq \frac{1}{\sum_{i=1}^{3^{k}}\widehat{\alpha}_{i}}.\label{opti-pro-eqv-pf-eq6}
\end{equation}
By \eqref{opti-pro-eqv-pf-eq6}, we thus obtain that
\begin{equation}
\widehat{n}\leq\max\Bigg\{\sum_{m=1}^{3^{k}}\lambda_{m}\log\gamma_{m},
\sum_{m=1}^{3^{k}}\lambda_{m}\log m\Bigg\}\leq \frac{1}{\sum_{i=1}^{3^{k}}\widehat{\alpha}_{i}},
\end{equation}
i.e.,
\begin{equation}
\widehat{r}=k\sum_{i=1}^{3^{k}}\widehat{\alpha}_{i}\leq \frac{k}{\,\,\widehat{n}\,\,}.\label{opti-pro-eqv-pf-eq7}
\end{equation}
Combining \eqref{opti-pro-eqv-pf-eq3} and \eqref{opti-pro-eqv-pf-eq7}, we have proved the equivalence, i.e., $\widehat{r}=\frac{k}{\,\,\widehat{n}\,\,}=\mathcal{C}(k)$.
}

 As proved in \cite[Section~\Rmnum{3}-A]{Wang-Tightness_upper_bound}, for the special case $k=1$,  $\gamma_{1}=4$, $\gamma_{2}=3$ and $\gamma_{3}=4$.
 %, which will also be proved in detail later (see Lemma~\ref{lem:gamma_3^k-1_value}) for the completeness of the paper.
  Hence, we  derive $\mathcal{C}(1)=\frac{1}{2}\log_{3}6$ by solving the optimization problem \eqref{eq:opti_problem_C_k} explicitly. This also gives a lower bound $\frac{1}{2}\log_{3}6$ on
 the computing capacity $\mathcal{C}(\widehat{\mathcal{N}},\widehat{f})$, i.e.,
 \begin{equation}
 \mathcal{C}(\widehat{\mathcal{N}},\widehat{f})\geq\mathcal{C}(1)=\frac{1}{2}\log_{3}6.\label{eq:C(1)-value}
 \end{equation}

Next, the authors  obtained an upper bound on the optimal value of \eqref{eq:opti_problem_C_k} for the case  $k> 1$. The details are as follows.
 For any positive integer $k>1$, the authors claimed in \cite[Lemma~1]{Wang-Tightness_upper_bound} that $\gamma_{1}=4^{k}$,  $\gamma_{3^{k}}=2^{k}$ and the following invalid inequality:
\begin{equation}
\gamma_{3^{k}-1}\geq 4^{k}-2^{k-1},\label{ineq:error_in_Wang}
\end{equation}
where we remark that
the invalidity of \eqref{ineq:error_in_Wang} will be proved in Example~\ref{counterexample} of Section~\ref{subsec:critical_error}.
%, and  we will further prove that
%$\gamma_{1}=4^{k}$,  $\gamma_{3^{k}}=2^{k}$ and  $\gamma_{3^{k}-1}=3\times 2^{k-1}$ ($< 4^{k}-2^{k-1}$)  in Section~\ref{subsec:critical_error} for the completeness of the paper.
In addition, it is also shown in \cite[Lemma~2]{Wang-Tightness_upper_bound} that
$\gamma_{m}\geq\gamma_{m+1}$ for $m = 1, 2,\cdots, 3^{k}-1$.
 %In fact,  by   \eqref{eq:def:gamma_m}, we can also readily see that  $\gamma_{m}\geq\gamma_{m+1}$ for $m = 1, 2,\cdots, 3^{k}-1$.
With the above preparation, the authors studied an upper bound on the optimal value by
relaxing the constraint in the problem \eqref{eq:opti_problem_C_k}. Since  $\gamma_{m}\geq\gamma_{3^{k}-1}\geq 4^{k}-2^{k-1}$ for $m =  2,3,\cdots, 3^{k}-1$, then the authors  formulated \emph{a relaxed optimization problem} by reducing $(\gamma_{m},m)$ to $(4^{k}-2^{k-1},2)$ for $m=2,3,\cdots,3^{k}-1$:
\begin{equation}
	\begin{split}
		\max_{\alpha_{1},\beta,\alpha_{3^{k}}}&k(\alpha_{1}+\beta+\alpha_{3^{k}})\\
		\text{s.t.}\quad&2k\alpha_{1}+\beta\log(4^{k}-2^{k-1})+k\alpha_{3^{k}}\leq 1;\\
		&\beta+\alpha_{3^{k}}\log3^{k}\leq 1; \quad\;\alpha_{1},\beta,\alpha_{3^{k}}\geq 0.
	\end{split}\label{eq:opti_problem_error}
\end{equation}
where $\beta=\sum_{m=2}^{3^{k}-1}\alpha_{m}$. The problem \eqref{eq:opti_problem_error} can be solved explicitly (by solving its dual problem) and the optimal value is $\frac{1}{2}\log_{3}6$. Since \eqref{eq:opti_problem_error} is a relaxed optimization problem of \eqref{eq:opti_problem_C_k},  $\frac{1}{2}\log_{3}6$ is also an \emph{upper bound} on $\mathcal{C}(k)$, i.e., $\mathcal{C}(k)\leq \frac{1}{2}\log_{3}6$ for $ k>1$.
Together with $\mathcal{C}(1)=\frac{1}{2}\log_{3}6$, the authors thus obtained that
 \begin{equation*}
\frac{1}{2}\log_{3}6=\mathcal{C}(1)\leq\mathcal{C}(\widehat{\mathcal{N}},\widehat{f})=\sup\{\mathcal{C}(k):k\geq1\}\leq\frac{1}{2}\log_{3}6,
\end{equation*}
i.e., the computing capacity $\mathcal{C}(\widehat{\mathcal{N}},\widehat{f})=\frac{1}{2}\log_{3}6$.

\subsection{The Issues in The  Technique}\label{subsec:critical_error}

In this subsection, we will use a counterexample to show that the claim \eqref{ineq:error_in_Wang} is not correct.  This incorrect claim, which is used to induce the  invaild relaxed optimization problem \eqref{eq:opti_problem_error}, will lead eventually an invalid upper bound $\frac{1}{2}\log_{3}6$ on $\mathcal{C}(\widehat{\mathcal{N}},\widehat{f})$.

\begin{example}\label{counterexample}
For any positive integer $k$,
we will construct an admissible $k$-shot network code $\{\theta_{\rm{\Rmnum{1}}},\theta_{\rm{\Rmnum{2}}};~\varphi\}$ as follows.
Let $\vec{0}$ be  a $k$-vector whose all components are equal to $0$, and  $\vec{e}_{1}\in\mathcal{Y}^{k}$ be a $k$-vector whose the first component is equal to $1$ while all other components are equal to $0$. The local encoding functions $\theta_{\rm{\Rmnum{1}}}$ and $\theta_{\rm{\Rmnum{2}}}$ are defined as follows:
 	\begin{align*}
		\begin{split}
			\theta_{\rm{\Rmnum{1}}}(\vec{x},\vec{y})= \left \{
			\begin{array}{ll}
				\vec{x}+\vec{y},     \quad\quad            &\text{\rm{if}}\;\;\vec{x}\in\mathcal{A}^{k}, \vec{y}\in \{\vec{0},\vec{e}_{1}\};  \\
				\vec{x},                   &\text{\rm{if}}\;\;\vec{x}\in\mathcal{A}^{k}, \vec{y}\in\mathcal{Y}^{k}\backslash\{\vec{0},\vec{e}_{1}\},\\
			\end{array}
			\right.
		\end{split}\label{construct_h_1}
	\end{align*}
	and
	\begin{align*}
		\begin{split}
			\theta_{\rm{\Rmnum{2}}}(\vec{y})= \left \{
			\begin{array}{ll}
				\vec{0},          \quad\quad      &\text{\rm{if}}\;\;\vec{y}\in\{\vec{0},\vec{e}_{1}\}; \\
				\vec{y}.           &\text{\rm{if}}\;\; \vec{y}\in\mathcal{A}^{k}\backslash\{\vec{0},\vec{e}_{1}\}.\\
			\end{array}
			\right.
		\end{split}
	\end{align*}
	On the other hand,  the decoding function $\varphi$ is defined by
\begin{equation*}
		\varphi\big( \theta_{\rm{\Rmnum{1}}}(\vec{x},\vec{y}),\theta_{\rm{\Rmnum{2}}}(\vec{y})\big)=\theta_{\rm{\Rmnum{1}}}(\vec{x},\vec{y}
		)+\theta_{\rm{\Rmnum{2}}}(\vec{y}).
\end{equation*}
It is easily seen  that  $\{\theta_{\rm{\Rmnum{1}}},\theta_{\rm{\Rmnum{2}}};~\varphi\}$ is indeed an admissible $k$-shot network code with
	\begin{equation}
| \Im\,\theta_{\rm{\Rmnum{1}}}|=3\times2^{k-1}\quad\text{and}\quad|\Im\,\theta_{\rm{\Rmnum{2}}}|=3^{k}-1 .\end{equation}
By the definition of $\gamma_{3^{k}-1}$ (see~\eqref{eq:def:gamma_m}),  we immediately obtain that for $k\geq 1$,
	\begin{equation}
		\gamma_{3^{k}-1}\leq |\Im\, \theta_{\rm{\Rmnum{1}}}|= 3\times2^{k-1}.\label{ineq:gamma_3^k-1_low_bound}
	\end{equation}
Then we readily see that for the case $k>1$, $\gamma_{3^{k}-1}\leq 3\times2^{k-1}<4^{k}-2^{k-1}$, and hence
 this  is a contradiction to  \eqref{ineq:error_in_Wang}.

 % With the same argument above, if we construct an admissible $k$-shot network code $\{\theta_{\rm{\Rmnum{1}}},\theta_{\rm{\Rmnum{2}}};~\varphi\}$ below: for each pair $(\vec{x},\vec{y})\in\mathcal{A}^k\times\mathcal{Y}^k$,
% \begin{equation}
% \begin{cases}
% &\theta_{\rm{\Rmnum{1}}}(\vec{x},\vec{y})=\vec{x}+\vec{y};\quad\\
% &\theta_{\rm{\Rmnum{2}}}(\vec{y})=\vec{0};\\
%&\varphi\big( \theta_{\rm{\Rmnum{1}}}(\vec{x},\vec{y}),\theta_{\rm{\Rmnum{2}}}(\vec{y})\big)=\theta_{\rm{\Rmnum{1}}}(\vec{x},\vec{y}
%		)+\theta_{\rm{\Rmnum{2}}}(\vec{y}),
% \end{cases}
% \end{equation}
% by \eqref{eq:def:gamma_m}, we immediately obtain that
% \begin{equation}
% \gamma_1\leq |\Im\,\theta_{\rm{\Rmnum{1}}}|=4^k.
% \end{equation}
% Similarly, if we construct an admissible $k$-shot network code $\{\theta_{\rm{\Rmnum{1}}},\theta_{\rm{\Rmnum{2}}};~\varphi\}$ below: for each pair $(\vec{x},\vec{y})\in\mathcal{A}^k\times\mathcal{Y}^k$,
% \begin{equation}
% \begin{cases}
% &\theta_{\rm{\Rmnum{1}}}(\vec{x},\vec{y})=\vec{x};\quad\\
% &\theta_{\rm{\Rmnum{2}}}(\vec{y})=\vec{y};\\
%&\varphi\big( \theta_{\rm{\Rmnum{1}}}(\vec{x},\vec{y}),\theta_{\rm{\Rmnum{2}}}(\vec{y})\big)=\theta_{\rm{\Rmnum{1}}}(\vec{x},\vec{y}
%		)+\theta_{\rm{\Rmnum{2}}}(\vec{y}),
% \end{cases}
% \end{equation}
% by \eqref{eq:def:gamma_m}, we immediately obtain that
% \begin{equation}
% \gamma_{3^k}\leq |\Im\,\theta_{\rm{\Rmnum{1}}}|=2^k.
% \end{equation}
\end{example}

Next, we will  give an enhanced form of the results in Example~\ref{counterexample}.
\begin{lemma}\label{lem:gamma_3^k-1_value}
 For the case $k\geq 1$, $\gamma_{3^{k}-1}=3\times 2^{k-1}$.
  \end{lemma}

 Before proving the above lemma, we first present  the following results.

\begin{lemma}\label{lemma:pf-W-relat_code_color_set}
For a nonnegative integer $\ell$ with  $0\leq \ell  \leq 3^{k}$, let
\begin{align}
	M_{k}(\ell)=\min_{\substack{\;\;\mathcal{L}\subseteq\mathcal{Y}^{k}\;\text{s.t.}\;|\mathcal{L}|=\ell}}|\mathcal{A}^{k}+ \mathcal{L}|.\label{def:M_k(l)}
	\end{align}
Then, for any positive integer $m$ with $1\leq m\leq 3^{k}$,
\begin{equation}
		\gamma_{m}\geq M_
		{k}\Big(\Big\lceil\;\frac{3^{k}}{m}\;\Big\rceil\Big).\label{eq:pf-W-lemma:relat_code_color_set}
	\end{equation}
\end{lemma}

The proof of Lemma~\ref{lemma:pf-W-relat_code_color_set} is deferred to Appendix~\ref{appendix-pf-lem-W-relate_code_color-set}.

\begin{prop}\label{lemma:M_k(2)} $M_{k}(2)=3\times 2^{k-1}$.
\end{prop}
\begin{proof}
First, we let $\mathcal{L}=\{\vec{0},\vec{e}_{1}\}$, where $\vec{0}\in\mathcal{Y}^{k}$ is  a $k$-vector whose all components are equal to $0$, and  $\vec{e}_{1}\in\mathcal{Y}^{k}$ is a $k$-vector whose the first component is equal to $1$ while all other components are equal to $0$. Immediately, we have $|\mathcal{A}^{k}+\mathcal{L}|=3\times 2^{k-1}$ and thus $M_{k}(2)\leq 3\times 2^{k-1}.$

Next, we let $\mathcal{L}=\{\vec{y}_{1},\vec{y}_{2}\}\subseteq\mathcal{Y}^{k}$ be an arbitrary subset of size $2$.
Without loss of generality, we assume that $\vec{y}_{1}[1]\neq \vec{y}_{2}[1]$, where $\vec{y}_{i}[1]$ stands for the first component of $\vec{y}_{i}$, $i=1,2$. We consider
\begin{align}
|\mathcal{A}^{k}+\mathcal{L}|&=\big|\mathcal{A}^{k}+\{\vec{y}_{1},\vec{y}_{2}\}\big|\label{eq_pf_lemma:M_k(2),1}\\
&=\big|\mathcal{A}^{k}+\{\vec{y}_{1}\}\big|
+\big|\big(\mathcal{A}^{k}+\{\vec{y}_{2}\}\big)\setminus\big(\mathcal{A}^{k}+\{\vec{y}_{1}\}\big)\big|.\nonumber
\end{align}
 It follows from $\vec{y}_{1}[1]\neq \vec{y}_{2}[1]$ that there exists a symbol $ x^{*}\in\mathcal{A}$ such that  $x^{*}+\vec{y}_{2}[1]\notin\mathcal{A}+\{\vec{y}_{1}[1]\}$. Further, we let
 \begin{equation}
 \mathcal{A}^{*}=\{ \vec{x}+\vec{y}_{2}:\vec{x}\in\mathcal{A}^{k}\;\text{with}\;\vec{x}[1]=x^{*}\}.
  \end{equation}
  Then, we see that $$\mathcal{A}^{*}\subseteq\big(\mathcal{A}^{k}+\{\vec{y}_{2}\}\big)\setminus\big(\mathcal{A}^{k}+\{\vec{y}_{1}\}\big),$$ and thus \begin{equation}
  \big|\big(\mathcal{A}^{k}+\{\vec{y}_{2}\}\big)\setminus\big(\mathcal{A}^{k}+\{\vec{y}_{1}\}\big)\big|\geq|\mathcal{A}^{*}|=2^{k-1}.
  \end{equation}
   Together with
 $|\mathcal{A}^{k}+\{\vec{y}_{1}\}|=2^{k}$, we continue to consider \eqref{eq_pf_lemma:M_k(2),1}:
 \begin{equation}
 |\mathcal{A}^{k}+\mathcal{L}|\geq 2^{k}+2^{k-1}=3\times 2^{k-1}.\label{eq_pf_lemma:M_k(2),2}
 \end{equation}
 Note that the inequality \eqref{eq_pf_lemma:M_k(2),2} is true for any $\mathcal{L}\subseteq\mathcal{Y}^{k}$ with size $2$, and so $M_{k}(2)\geq 3\times 2^{k-1}.$ We thus have proved that $M_{k}(2)= 3\times 2^{k-1}$.  The lemma is proved.
\end{proof}

We immediately obtain the following consequence of Lemma~\ref{lemma:pf-W-relat_code_color_set} and Proposition~\ref{lemma:M_k(2)}.

\begin{cor}\label{cor:gamma_3^k-1_low_bound}
For the case $k\geq 1$, $\gamma_{3^{k}-1}\geq M_{k}(2)=3\times 2^{k-1}$.
\end{cor}

  Together Example~\ref{counterexample} with Corollary \ref{cor:gamma_3^k-1_low_bound} above, we immediately prove  Lemma~\ref{lem:gamma_3^k-1_value}.

\subsection{A valid upper bound}\label{subsec:corr_tech}

  In this subsection, we will obtain a valid upper bound on $\mathcal{C}(\widehat{\mathcal{N}},\widehat{f})$ by following the technique proposed by \cite{Wang-Tightness_upper_bound}.

For the special case of $k=1$, we already know that $\mathcal{C}(1)=\frac{1}{2}\log_{3}6\approx 0.815$ (cf.~\eqref{eq:C(1)-value}).
Next we study an upper bound on the optimal value $\mathcal{C}(k)$, $k\geq2$, by
relaxing the constraint in \eqref{eq:opti_problem_C_k}. From the discussion in Section~\ref{subsec:Wang_tech-intro} and Lemma~\ref{lem:gamma_3^k-1_value}, we obtain that $\gamma_{1}=4^k$, $\gamma_{3^k}=2^k$, and  $\gamma_{m}\geq\gamma_{3^{k}-1}= 3\times2^{k-1}$ for $m =  2,3,\cdots, 3^{k}-2$.
Then, we  reformulate \emph{a relaxed optimization problem} similarly by reducing $(\gamma_{m},m)$ to $(3\times2^{k-1},2)$ for $m=2,3,\cdots,3^{k}-2$:
 \begin{equation}
	\begin{split}
		\max_{\alpha_{1},\beta,\alpha_{3^{k}-1},\alpha_{3^{k}}}&k(\alpha_{1}+\beta+\alpha_{3^{k}-1}+\alpha_{3^{k}})\\
		\text{s.t.}\quad&2k\alpha_{1}+\beta\log(3\times2^{k-1})+\alpha_{3^{k}-1}\log(3\times2^{k-1})+k\alpha_{3^{k}}\leq 1;\\
		&\beta+\alpha_{3^{k}-1}\log(3^{k}-1)+\alpha_{3^{k}}\log3^{k}\leq 1; \quad\;\alpha_{1},\beta,\alpha_{3^{k}-1},\alpha_{3^{k}}\geq 0.
	\end{split}\label{eq:relax_opti_problem-correct}
\end{equation}
where $\beta=\sum_{m=2}^{3^{k}-2}\alpha_{m}$.
For the problem \eqref{eq:relax_opti_problem-correct},
 we can obtain  the  optimal value (by solving its dual problem),  denoted by $\mathcal{D}(k)$, $k\geq 2$,  as follows:
\begin{equation}\nonumber
	\mathcal{D}(k)=	1-\frac{(\log3-1)^{2}}{(k-1)\log3+(\log3)^{2}-1}.
\end{equation}
 Clearly,
  \begin{equation}
		\mathcal{C}(\widehat{\mathcal{N}},\widehat{f})=\sup\{ \mathcal{C}(k): k\geq 1\}\leq\sup\Big\{\frac{1}{2}\log_{3}6; ~\mathcal{D}(k): k\geq 2\Big\}=1.
\end{equation}
We therefore obtain \emph{a valid upper bound}
\begin{equation}
\mathcal{C}(\widehat{\mathcal{N}},\widehat{f})\leq 1,
 \end{equation}
 where the upper bound of Guang~\emph{et al.} on $\mathcal{C}(\widehat{\mathcal{N}},\widehat{f})$ is also equal to $1$.
 This also implies that the technique proposed by \cite{Wang-Tightness_upper_bound} does not give a better upper bound than  the upper bound of Guang~\emph{et al.}~\cite{Guang_Improved_upper_bound}.
}

\end{journalonly}

\section{Conclusion}\label{sec:concl}

In the paper, we put forward the model of zero-error distributed function compression system of two binary memoryless sources $(\boldsymbol{\mathrm{s}}_1\boldsymbol{\mathrm{s}}_2;C_{1}, C_{2}; f)$ depicted in Fig.\,\ref{fig:Gen-System_use}, where the function $f$ is the binary arithmetic sum, the two switches $\boldsymbol{\mathrm{s}}_1$ and $\boldsymbol{\mathrm{s}}_2$ take values $0$ or $1$ to represent them open or closed, and by symmetry, it is assumed without loss of generality that $C_1\geq C_2$ for the capacity constraints $C_1$ and $C_2$. The compression capacity for the model is defined as the maximum average number of times that the function $f$ can be compressed with zero error for one use of the system, which measures the efficiency of using the system. We  fully characterized the compression capacities for all the four cases of the model $(\boldsymbol{\mathrm{s}}_1\boldsymbol{\mathrm{s}}_2;C_{1}, C_{2}; f)$ for $\boldsymbol{\mathrm{s}}_1\boldsymbol{\mathrm{s}}_2=00,01,10,11$, where characterizing the compression capacity for the case $(01;C_1,C_2;f)$ with $C_1>C_2$ is highly nontrivial. Toward this end, we developed a novel graph coloring approach and proved a crucial aitch-function that plays a key role in estimating the minimum chromatic number of a conflict graph. Furthermore,
we presented an important application of the zero-error distributed function compression problem for the model $(01;C_1,C_2;f)$ to  network function computation. By considering the model of network function computation  that is transformed from $(01;C_1,C_2;f)$ with $C_1>C_2$,
 we gave the answer to the open problem that whether the best known upper bound of Guang~\emph{et al.} on computing capacity is in general tight, i.e., in general
this upper bound is  not tight.

The more general case of the model $(\boldsymbol{\mathrm{s}}_1\boldsymbol{\mathrm{s}}_2;C_{1}, C_{2}; f)$ is to consider arbitrary joint probability distributions of two sources beyond binary sources and arbitrary functions $f$. The characterization of the compression capacity for this general model will be investigated in the future.

\numberwithin{theorem}{section}
\appendices

\section{Proof of~Lemma~\ref{lemma:chi_m_low_bound}}\label{append-lemma:chi_m_low_bound}

Let $m$ be an arbitrary positive integer such that $1\leq m\leq 2^k$.
We first prove that
\begin{equation}\label{cor:chi_m=min_max_color}
	\chi_{m}\leq\min_{\substack{\text{ \rm{all partitions} }\\\text{$\{P_1,P_2,\cdots,P_m\}$ \rm{of} $\mathcal{A}^{k}$}}}\max_{1\leq i\leq m}\chi\big(G(P_i)\big).
\end{equation}
Let $\big\{P^*_1,P^*_2,\cdots,P^*_m\big\}$ be a partition of $\mathcal{A}^{k}$ such that
 \begin{equation}\label{cor:chi_m-min_max_color-eq1}
	\max_{1\leq i\leq m}\chi\big(G(P^*_i)\big)=\min_{\substack{\text{ \rm{all partitions} }\\\text{$\{P_1,P_2,\cdots,P_m\}$ \rm{of} $\mathcal{A}^{k}$}}}\max_{1\leq i\leq m}\chi\big(G(P_i)\big).
\end{equation}
For each block $P_i^*,$ $i=1,2,\cdots,m$, we let $c_i$ be a coloring of the conflict graph $G(P_i^*)$ of the
minimum chromatic number $\chi\big(G(P_i^*)\big)$, i.e., \rmnum{1})
for any two adjacent vertices
$(\vec{x},\vec{y})$ and $(\vec{x}\,',\vec{y}\,')$ in $G(P^*_i)$ (which satisfy $\vec{x}+\vec{y}\neq \vec{x}\,'+\vec{y}\,'$),
\begin{equation}
c_i(\vec{x},\vec{y})\neq c_i(\vec{x}\,',\vec{y}\,'), \label{G(P_i)-Color-c_i}
\end{equation}
and \rmnum{2}) $|\Im\,c_i|=\chi\big(G(P_i^*)\big)$.
Further, we let
 \begin{equation} \Im\,c_i=\big\{1,2,\cdots,\chi\big(G(P^*_i)\big)\big\},\quad i=1,2,\cdots,m.\label{cor:chi_m-min_max_color-eq2}
  \end{equation}
Based on the above colorings $c_i,$ $i=1,2,\cdots,m$, we give a $k$-shot code $\mbC=\big\{\phi_1,\phi_2;~\psi\big\}$ for the model $(01;2,1;f)$, where
\begin{itemize}
\item the encoding functions $\phi_1$ and $\phi_2$ are defined as follows: for each $i=1,2,\cdots,m$,
\begin{equation}
\begin{cases}
\phi_1(\vec{x},\vec{y})=c_i(\vec{x},\vec{y}),\quad&\forall~(\vec{x},\vec{y})\in\mathcal{A}^k\times P^*_i;\\
\phi_2(\vec{y})=i,\quad&\forall~\vec{y}\in P^*_i;
\end{cases}\label{phi-1,2-constr}
\end{equation}
\item the decoding function $\psi$ is defined by
\begin{equation*}
\psi\big(\phi_1(\vec{x},\vec{y}),\phi_2(\vec{y})\big)=\vec{x}+\vec{y},\quad\forall~(\vec{x},\vec{y})\in
\mathcal{A}^k\times\mathcal{A}^k.
\end{equation*}
\end{itemize}
In the following, we verify the admissibility of the code $\mbC$. Toward this end, we claim that
\begin{equation}
\big(\phi_{1}(\vec{x},\vec{y}), \phi_{2}(\vec{y})\big)\neq \big(\phi_{1}(\vec{x}\,',\vec{y}\,'),
\phi_{2}(\vec{y}\,')\big),\label{cor:chi_m-min_max_color-eq6}
\end{equation}
for any two pairs
$(\vec{x},\vec{y})$ and $(\vec{x}\,',\vec{y}\,')$ in $\mathcal{A}^k\times\mathcal{A}^k$ with $\vec{x}+\vec{y}\neq \vec{x}\,'+\vec{y}\,'$.
To be specific, if $\vec{y}$ and $\vec{y}\,'$ are in two different blocks, say $\vec{y}\in P_i^*$
and $\vec{y}\,'\in P_j^*$ with $i\neq j$, then
\begin{equation*}
\phi_2(\vec{y})=i\neq j=\phi_2(\vec{y}\,'),
\end{equation*}
which implies \eqref{cor:chi_m-min_max_color-eq6}. Otherwise, $\vec{y}$ and $\vec{y}\,'$ are in the same
block, say $P_i^*$, which implies
$\phi_2(\vec{y})=\phi_2(\vec{y}\,')=i$. By the definition of the coloring $c_i$ (cf.~\eqref{G(P_i)-Color-c_i}) and the definition of the encoding function $\phi_1$ (cf.~\eqref{phi-1,2-constr}), we have
\begin{equation*}
\phi_1(\vec{x},\vec{y})=c_i(\vec{x},\vec{y})\neq c_i(\vec{x}\,',\vec{y}\,')=\phi_1(\vec{x}\,',\vec{y}\,'),
\end{equation*}
which also implies \eqref{cor:chi_m-min_max_color-eq6}.

Now, for this admissible $k$-shot code $\mbC$, it follows from   \eqref{cor:chi_m-min_max_color-eq2} and \eqref{phi-1,2-constr} that
\begin{equation*}
|\Im\,\phi_1|=\max_{1\leq i\leq m}\chi\big(G(P^*_i)\big)\quad\text{and}\quad |\Im\,\phi_2|=m.
\end{equation*}
Together with the definition of $\chi_{m}$ (cf.~\eqref{eq:def:chi_m}) and \eqref{cor:chi_m-min_max_color-eq1}, we prove \eqref{cor:chi_m=min_max_color} by
\begin{equation*}
\chi_m\leq |\Im\,\phi_1|=\max_{1\leq i\leq m}\chi\big(G(P^*_i)\big)=\min_{\substack{\text{ \rm{all partitions} }\\\text{$\{P_1,P_2,\cdots,P_m\}$ \rm{of} $\mathcal{A}^{k}$}}}\max_{1\leq i\leq m}\chi\big(G(P_i)\big).\label{cor:chi_m<=min_max_color-5}
\end{equation*}

 Next, we prove that
 \begin{equation*}
\chi_m\geq\min_{\substack{\text{ \rm{all partitions} }\\\text{$\{P_1,P_2,\cdots,P_m\}$ \rm{of} $\mathcal{A}^{k}$}}}\max_{1\leq i\leq m}\chi\big(G(P_i)\big).
\end{equation*}
Let  $\big\{\phi_1,\phi_2;~\psi\big\}$ be an admissible $k$-shot  code such that $|\Im\,\phi_2|=m$ and $|\Im\,\phi_{1}|=\chi_m$ (cf.~\eqref{eq:def:chi_m}).
 For each $\gamma\in\Im\,\phi_{2}$, we let $\phi_{2}^{-1}(\gamma)$ be  the inverse
image of~$\gamma$ under $\phi_{2}$, i.e.,
\begin{equation*}
\phi_{2}^{-1}(\gamma)=\big\{\vec{y}\in\mathcal{A}^{k}:~\phi_{2}(\vec{y})=\gamma\big\}.
\end{equation*}
 Then, we readily see that
 \begin{equation}
 \mathcal{P}=\big\{P_{\gamma}\triangleq\phi_{2}^{-1}(\gamma):~\gamma\in\Im\,\phi_{2}\big\}\label{pf-pro1-eq-00}
  \end{equation}
  is a  partition of $\mathcal{A}^{k}$, and  all the $m$ blocks $P_{\gamma},~\gamma\in\Im\,\phi_{2}$ are non-empty.

 We  first prove that
 \begin{equation*}
\chi_m=|\Im\,\phi_{1}|\geq\max_{\gamma\in\Im\,\phi_{2}}\chi\big(G(P_{\gamma})\big),\label{pf-pro1-eq-1}
\end{equation*}
 or equivalently,
 \begin{equation}
 \chi_m=|\Im\,\phi_{1}|\geq\chi\big(G(P_{\gamma})\big)\label{pf-pro1-eq-0}
 \end{equation}
for each block $P_{\gamma}$. To see this, by the discussion immediately below the  proof of Lemma~\ref{def:conf-graph-G^k},
$(\phi_1,\phi_2)$ is a coloring of the conflict graph~$G$.
 For each subgraph $G(P_{\gamma})$, $(\phi_1,\phi_2)$ restricted to $G(P_{\gamma})$ is a coloring of $G(P_{\gamma})$. This implies that
\begin{equation}
\chi\big(G(P_{\gamma})\big)\leq\big|\big\{\big(\phi_{1}(\vec{x},\vec{y}),\phi_{2}(\vec{y})\big):~\vec{x}\in\mathcal{A}^{k}~\text{and}
~\vec{y}\in P_{\gamma}\big\}\big|.\label{pf-pro1-eq-2}
\end{equation}
We note that $\phi_{2}(\vec{y})=\gamma,~\forall~\vec{y}\in P_{\gamma}$ ($=\phi_{2}^{-1}(\gamma)$).
Continuing from \eqref{pf-pro1-eq-2}, we  obtain that
\begin{align}
\chi\big(G(P_{\gamma})\big)&\leq\big|\big\{\big(\phi_{1}(\vec{x},\vec{y}),\phi_{2}(\vec{y})\big):~\vec{x}\in\mathcal{A}^{k}~\text{and}
~\vec{y}\in P_{\gamma}\big\}\big|=\big|\big\{\phi_{1}(\vec{x},\vec{y}):~\vec{x}\in\mathcal{A}^{k}~\text{and}
~\vec{y}\in P_{\gamma}\big\}\big|\nonumber\\
& \leq\big|\big\{\phi_{1}(\vec{x},\vec{y}):~\vec{x}\in\mathcal{A}^{k}~\text{and}
~\vec{y}\in\mathcal{A}^{k}\big\}\big|=|\Im\,\phi_{1}|.\nonumber%\label{pf-pro1-eq-3}
\end{align}
We thus  prove \eqref{pf-pro1-eq-0} for each block $P_{\gamma}$.
 Finally,  we note that  $\mathcal{P}$ in \eqref{pf-pro1-eq-00} is a particular partition of $\mathcal{A}^{k}$ that includes $m$ non-empty blocks. So, we further obtain that
 \begin{equation}
 \chi_m=|\Im\,\phi_{1}|\geq\max_{\gamma\in\Im\,\phi_{2}}\chi\big(G(P_{\gamma})\big)\geq\min_{\substack{\text{ all partitions }\\\text{$\{P_1,P_2,\cdots,P_m\}$ of $\mathcal{A}^{k}$}}}\max_{1\leq i\leq m}\chi\big(G(P_i)\big).\nonumber%\label{cor:chi_m>=min_max_color-5}
 \end{equation}
Thus, we have proved Lemma~\ref{lemma:chi_m_low_bound}.

\section{Proof of~Lemma~\ref{lemma:h(l)-property}}\label{appendix-pf-lemma:h(l)-property}

By $h(0)=0$, we first note that the inequality~\eqref{lemma:eq-h(l)-property} is satisfied for $\ell=0$. Now, consider an arbitrary positive integer $\ell$, and thus we have $$h(\ell)>0,~\ell=1,2,\cdots.$$
Let  $\ell_a$ and $\ell_b$ be two nonnegative integers satisfying
 $\ell=\ell_a+\ell_b$ and $\ell_a\geq\ell_b$. To prove \eqref{lemma:eq-h(l)-property}, it is equivalent to proving
\begin{equation*}
\frac{h(\ell_a)}{h(\ell)}+\frac{1}{2}\cdot\frac{h(\ell_b)}{h(\ell)}\geq 1,\label{pflemma:eq-h(l)-property-1}
\end{equation*}
or equivalently,
\begin{equation}
\Big(\frac{\ell_a}{\ell}\Big)^{\log3-1}+\frac{1}{2}\cdot\Big(\frac{\ell_b}{\ell}\Big)^{\log3-1}\geq 1.\label{pflemma:eq-h(l)-property-2}
\end{equation}
With \eqref{pflemma:eq-h(l)-property-2}, we define the following function:
 \begin{equation*}
 q(x)= x^{\log3-1}+\frac{1}{2}\cdot (1-x)^{\log3-1},\quad x\in\Big[\dfrac{1}{2},~1\Big].\label{pflemma:eq-h(l)-property-3}
 \end{equation*}

 For the above function $q(x)$, we first easily calculate
 \begin{equation}
 q\Big(\frac{1}{2}\Big)=q(1)=1.\label{pflemma:eq-h(l)-property-4}
 \end{equation}
 Further, after a simple calculation, we obtain that  the second derivative
 \begin{align*}
 q''(x)<0,\quad\forall~x\in\Big[\frac{1}{2},~1\Big].
 \end{align*}
 This shows that $q(x)$ is concave on the interval $[1/2,~1]$. Together with \eqref{pflemma:eq-h(l)-property-4}, we prove that
 \begin{align}
 q(x)\geq 1,\quad\forall~x\in\Big[\frac{1}{2},~1\Big].\nonumber
 \end{align}
 We thus have proved the inequality~\eqref{pflemma:eq-h(l)-property-2}. Then, the lemma is proved.

\section{Guang~et al.'s Upper Bound on the Computing Capacity $\mathcal{C}(\mathcal{N},f)$}\label{appendix:GYYB-bound-intro_and_com}

We first present the general model of network function computation following from \cite{Guang_Improved_upper_bound}. Let $\mG=(\mV,\mE)$ be an arbitrary directed acyclic graph, where $\mV$ and $\mE$ are a finite set of nodes and a finite set of edges, respectively. We assume that a symbol taken from a finite alphabet $\mA$ can be reliably transmitted on each edge for each use. On the graph $\mathcal{G}$, there are $s$ \emph{source nodes} $\sigma_1, \sigma_2, \cdots, \sigma_s$, where we let $S=\{\sigma_1, \sigma_2, \cdots, \sigma_s\}$, and a single \emph{sink node} $\rho \in \mV \setminus S$. Each source node has no input edges and the single sink node  has no output edges. Further, we assume without loss of generality that there exists a directed path from every node $u\in \mathcal{V}\setminus \{\rho\}$ to $\rho$ in~$\mG$. The graph $\mG$, together with the set of the source nodes $S$ and the single sink node~$\rho$, forms a \emph{network} $\mN$, i.e., $\mN=(\mG, S, \rho)$.
Let $f:~\mA^s \to \mO$ be an arbitrary nonconstant function, called the \emph{target function}, where $\mO$ is viewed as the image set of $f$. The $i$th argument of the target function~$f$ is the source message generated by the $i$th source node $\sigma_i$ for $i=1,2,\cdots,s$, and the target function $f$ is required to  compute with zero error at the sink node $\rho$ multiple times. To be specific, let
$k$ be a positive integer. For each $i=1,2,\cdots,s$, the $i$th source node~$\sigma_i$ generates $k$ symbols $x_{i,1},x_{i,2},\cdots,x_{i,k}$  in $\mA$, written as a vector $\vec{x}_i$, i.e., $\vec{x}_i=(x_{i,1},x_{i,2},\cdots,x_{i,k})$.
We further let $\vec{x}_S=(\vec{x}_{1}, \vec{x}_{2}, \cdots, \vec{x}_{s})$. The $k$ values of the target function
\begin{align}%\label{function_f}
f(\vec{x}_S)\triangleq\big(f(x_{1,j},x_{2,j}, \cdots, x_{s,j}): j=1,2,\cdots,k \big)\nonumber
\end{align}
are required to compute with zero error at $\rho$. We use $(\mN,f)$ to denote the general model of network function computation.

We further follow from \cite{Guang_Improved_upper_bound} to present some graph-theoretic notations. For two nodes $u$ and $v$ in $\mathcal{V}$,  we write  $u\rightarrow v$ if there exists a (directed) path from~$u$ to $v$ on $\mathcal{G}$.
If there exists no such a path from~$u$ to $v$, we say that $v$ is \emph{separated from} $u$, denoted by $u\nrightarrow v$.
 Given an edge subset $C\subseteq\mathcal{E}$, we  define three subsets  of  source nodes as follows:
\begin{align*}
K_C &=  \big\{ \sigma\in S:\ \exists\ e\in C \text{ s.t. } \sigma\rightarrow\tail(e) \big\},\label{def_K_C}\\
I_C&= \big\{ \sigma\in S:\ \text{$\sigma\nrightarrow\rho$ upon deleting the edges in $C$ from $\mathcal{E}$} \big\},\\
J_C&=K_C\setminus I_C.
\end{align*}
An edge subset $C$ is said to be a \emph{cut set} if $I_{C}\neq\emptyset$, and let $\Lambda(\mathcal{N})$  be the collection of all the cut sets, i.e.,
$\Lambda(\mathcal{N}) = \{C\subseteq \mathcal{E} : I_{C}\neq\emptyset\} $.
For any subset $J\subseteq S$, we let $x_{J} = (x_{j}:\sigma_{j}\in J) $ and use
$\mathcal{A}^{J}$
(instead of $\mathcal{A}^{|J|}$
for simplicity) to denote the set of
all possible $|J|$-dimensional vectors taken by $x_{J}$. In particular, when $J=\emptyset$, we have
$\mathcal{A}^{J}=\mathcal{A}^{0}$, the
singleton that contains an empty vector of dimension  $0$. Furthermore,
for notational convenience, we suppose that argument of the target function $f$ with subscript $i$ always stands for the symbol generated by the $i$th source node $\sigma_{i}$, so that we can ignore the order of the arguments of~$f$.

Now, we  present the upper bound of Guang~\emph{et al.}  on $\mathcal{C}(\mathcal{N},f)$ as follows (cf. {\cite[Theorem~2]{Guang_Improved_upper_bound}}):
\begin{equation}
\mathcal{C}(\mathcal{N},f)\leq\min_{C\in\Lambda(\mathcal{N})}\frac{|C|}{\log_{|\mathcal{A}|} n_{C,f}},\label{Guang-upper-bound}
\end{equation}
where   the key parameter $n_{C,f}$, which depends on the cut set $C$ and the target function $f$, will be specified in the following discussions.

\begin{definition}[\!\!{\cite[Definition~1]{Guang_Improved_upper_bound}}]\label{def:I_a_j_equiv}
	Consider two disjoint sets $I,J\subseteq S$ and a fixed $a_{J}\in \mathcal{A}^{J}$. For any $b_{I}$ and $b'_{I}\in \mathcal{A}^{I}$, we say $b_{I}$ and $b'_{I}$ are $(I,a_{J})$-equivalent if
	\begin{equation}\nonumber
		f(b_{I},a_{J},d)=	f(b'_{I},a_{J},d),\quad \forall~ d\in\mathcal{A}^{S\setminus(I\cup J)}.
	\end{equation}
\end{definition}
The above relation has been proved to be an equivalence relation, which thus
partitions $\mathcal{A}^{ I}$ into \emph{$(I,a_{J} )$-equivalence classes.}

\begin{definition}[\!\!{\cite[Definition~2]{Guang_Improved_upper_bound}}]\label{def:strong_partition}
	Let $C\in\Lambda(\mathcal{N})$ be a cut set and $\mathcal{P}_{C} =\{C_{1}, C_{2},\cdots,C_{m}\}$
	be a partition of the cut set $C$. The partition $\mathcal{P}_{C}$ is said to be a strong
	partition of $C$ if the following two conditions are satisfied:
	\begin{enumerate}
		\item $I_{C_{\ell}}\neq\emptyset$,\quad $\forall~ 1\leq \ell\leq m;$
		\item $I_{C_{i}}\cap K_{C_{j}}=\emptyset$,\quad$\forall~ 1\leq i,j\leq m $ and $i\neq j$.\footnote{There is a typo in the original definition of strong partition {\cite[Definition~2]{Guang_Improved_upper_bound}}, where in 2), ``$I_{C_i}\cap I_{C_j}=\emptyset$'' in~{\cite[Definition~2]{Guang_Improved_upper_bound}} should be ``$I_{C_{i}}\cap K_{C_{j}}=\emptyset$'' as stated in Definition~\ref{def:strong_partition} in the current paper.}
		\end{enumerate}
\end{definition}

For a cut set $C$ in $\Lambda(\mathcal{N})$,  the partition $\{C\}$ is called the
\emph{trivial strong partition} of $C$.

\begin{definition}[\!\!{\cite[Definition~3]{Guang_Improved_upper_bound}}]\label{def:I_a_L_a_j_equiv}
	Let $I$ and $J$ be two disjoint subsets of $S$. Let
	$I_{\ell},~\ell= 1, 2,\cdots, m,$ be $m$ disjoint subsets of $I$ and let $L = I\setminus(\cup_{\ell=1}^{m}I_{\ell})$. For given $a_{J}\in\mathcal{A}^{J}$ and $a_{L}\in\mathcal{A}^{L}$, we say that
	$b_{I_{\ell}}$	 and $b_{I_{\ell}}'$	 in $\mathcal{A}^{I_{\ell}}$	 are $(I_{\ell}, a_{L}, a_{J} )$-equivalent for $1\leq \ell\leq m$,
	if for each $c_{I_{j}}\in\mathcal{A}^{I_{j}} $ with $1\leq j\leq m$ and $j\neq \ell$,
		$\big(b_{I_{\ell}},a_{L},\{c_{I_{j}},1\leq j\leq m, j\neq \ell\}\big)$
and
	$	\big(b_{I_{\ell}}',a_{L},\{c_{I_{j}},1\leq j\leq m, j\neq \ell\}\big)$ in $\mathcal{A}^{I}$
		are $(I, a_{J})$-equivalent.
	\end{definition}

It has also been proved that  the above $(I_{\ell}, a_{L}, a_{J} )$-equivalence relation for every $\ell$ is also an equivalence
relation and thus partitions $\mathcal{A}^{I_{\ell}}$
into \emph{$(I_{\ell}
, a_{L}, a_{J})$-equivalence
classes}.

Now, we specify $(\mathcal{N},f)$ to be the model of network function computation depicted in Fig.\,\ref{fig:N}. We will explicitly calculate the upper bound of Guang~\emph{et al.} for the model. Toward this end, we will calculate $n_{C,f}$ for every cut set $C\in\Lambda(\mathcal{N})$ by considering two cases below. We first let $C\in\Lambda(\mathcal{N})$
be an arbitrary cut set, and then  let $I=I_{C}$ and $J=J_{C}$ for notational simplicity.\\
\textbf{Case 1:} $|I|=1$, i.e., $I=\{\sigma_{1}\}$ or $\{\sigma_{2}\}$.

For this case, we note that $|C|\geq 2$ and  the trivial strong partition $\{C\}$ is the only strong partition of $C$.
Then, $n_{C,f}$ is the maximum number of all $(I,a_{J})$-equivalence classes over all $a_{J}\in\mathcal{A}^{J}$, i.e.,
\begin{equation}\label{eq:trivi-partition-n_cf-def}
 n_{C,f}=\max_{a_{J}\in\mathcal{A}^{J}}\big|\big\{\text{all $(I,a_{J})$-equivalence classes}\big\}\big|.
\end{equation}
With \eqref{eq:trivi-partition-n_cf-def}, we can readily calculate  $n_{C,f}=2$ by considering the following   two subcases:

\textbf{Case 1A:} $I=\{\sigma_{i}\}$, $J=\{\sigma_{j}\}$, and then $S\setminus(I\cup J)=\emptyset$ (e.g., $C=\{e_{1},e_{2}\}$);

\textbf{Case 1B:} $I=\{\sigma_{i}\}$, $J=\emptyset$, and then $S\setminus(I\cup J)=\{\sigma_{j}\}$ (e.g., $C=\{d_{1},d_{2}\}$);\\
where the index set $\{i,j\}=\{1,2\}$, i.e., either $j=2$ if $i=1$ or $j=1$ if $i=2$.\\
\textbf{Case 2:} $|I|=2$, i.e., $I=S=\{\sigma_{1},\sigma_{2}\}$.

For this case, we consider the following two subcases to calculate $n_{C,f}$.

\textbf{Case 2A:} The cut set $C$ only has a (trivial) strong partition $\{C\}$ (e.g., $C=\{e_{1},e_{2},e_{3}\}$).

 We note  that $|C|\geq 3$ and $n_{C,f}$ is the number of all possible function values, i.e., $n_{C,f}=3$.

\textbf{Case 2B:} The cut set $C$ has a non-trivial strong partition.

  We note that $|C|\geq 5$ and every non-trivial strong partition must be a two-partition $\{C_{1},C_{2}\}$ with $K_{C_1}= I_{C_{1}}=\{\sigma_{1}\}$ and $K_{C_2}= I_{C_{2}}=\{\sigma_{2}\}$.
  We let $I_i=I_{C_i},~i=1,2$ for notational simplicity.
Then, we  have $L=I\setminus(I_{1}\cup I_{2})=\emptyset$ and $J=\emptyset$ so that $a_{L}\in\mathcal{A}^L$ and $a_{J}\in\mathcal{A}^J$ are two empty vectors. For example, the cut set  $C=\{d_{1},d_{2},d_{3},d_{4},e_{3}\}$  has a
unique nontrivial strong partition $\mathcal{P}_{C}=\big\{C_{1}=\{d_{1},d_{2}\},C_{2}=\{d_{3},d_{4},e_{3}\}\big\}$ with $K_{C_1}=I_{1}=\{\sigma_{1}\}$ and $K_{C_2}=I_{2}=\{\sigma_{2}\}$.
By Definition~\ref{def:I_a_L_a_j_equiv}, for $\ell=1,2$,  the $(I_{\ell},a_{L},a_{J})$-equivalence relation is specified to be the one that  two symbols $b$ and $b'$ in $\mathcal{A}^{I_{\ell}}=\{0,1\}$ are $(I_{\ell},a_{L},a_{J})$-equivalent if
\begin{equation*}
f(b,c)=f(b',c),~\text{i.e., $b+c=b'+c$},\quad\forall~c\in \mathcal{A}^{I_{\bar{\ell}}}=\{0,1\},
\end{equation*}
where  $\bar{\ell}=2$ if $\ell=1$ or $\bar{\ell}=1$ if $\ell=2$.
Thus, we   see that $\{0\}$ and $\{1\}$ are two distinct $(I_{\ell},a_{L},a_{J})$-equivalence classes, $\ell=1,2$.
Now, we calculate $n_{C,f}$ as follows:
\begin{align*}
n_{C,f}&=\sum_{b=0}^{2}\big|\big\{(b_{1},b_{2}):\text{$b_{1}\in\mathcal{A}^{I_{1}}$, $b_{2}\in\mathcal{A}^{I_{2}}$ s.t. $f(b_{1},b_{2})=b_1+b_2=b$}\big\}\big|\\
&=\big|\big\{(0,0)\big\}\big|+\big|\big\{(0,1),(1,0)\big\}\big|+\big|\big\{(1,1)\big\}\big|=4.
\end{align*}

Combining all the above discussions, we summarize  that for each cut set $C\in\Lambda(\mathcal{N})$,
\begin{equation*}
  \begin{cases}
    \text{$|C|\geq 2$ and $n_{C,f}=2$, } & \text{if $|I_C|=1$}; \\
    \text{$|C|\geq 3$ and $n_{C,f}=3$, } & \text{if $|I_C|=2$ and $C$ only has  a (trivial) strong partition $\{C\}$};\\
    \text{$|C|\geq 5$ and $n_{C,f}=4$, } & \text{if $|I_C|=2$ and $C$ has  a non-trivial strong partition}.
  \end{cases}\label{eq:3-classes-cut}
\end{equation*}
Then, we calculate the upper bound \eqref{Guang-upper-bound} as
\begin{equation*}
\min_{C\in\Lambda(\mathcal{N})}\frac{|C|}{\log n_{C,f}}\geq\min\bigg\{\frac{2}{\log 2},~\frac{3}{\log 3},~\frac{5}{\log 4}\bigg\}=\frac{3}{\log 3}.
\end{equation*}
Further, we consider the  cut set $C=\{e_{1},e_{2},e_{3}\}$  and
 then we have $|C|=3$ and $n_{C,f}=3$ (see Case~2A).
 This immediately implies  that
\begin{equation*}
\min_{C\in\Lambda(\mathcal{N})}\frac{|C|}{\log n_{C,f}}=\frac{3}{\log 3}=\log_{3}8.\label{Guang-upper-bound-value}
\end{equation*}

\begin{journalonly}

\section{Proof~of~Lemma~\ref{lemma:pf-W-relat_code_color_set}}\label{appendix-pf-lem-W-relate_code_color-set}

{\color{blue}
Before proving the lemma, we first present the following proposition, which can be proved
by the same argument in the proof of Lemma~\ref{lemma:chi(G(P))-equiv-form}.
\begin{prop}
For a subset $\mathcal{L}\subseteq\mathcal{Y}^k$, consider  the  subgraph $S(\mathcal{L})$  of the conflict graph $S$. Then,
\begin{equation}
\gamma\big(S(\mathcal{L})\big)=|\mathcal{A}^k+\mathcal{L}|.\label{eq:gamma(S(P))-low-bound}
\end{equation}
\end{prop}

Now, we start to prove Lemma~\ref{lemma:pf-W-relat_code_color_set}.
Consider an arbitrary partition $\mathcal{P}=\{P_1,P_2,\cdots,P_m\}$ of  $\{0,1,2\}^{k}$ that  includes  $m$  blocks. We see that there exists a block in $\mathcal{P}$, say $P_1$, such that
\begin{equation}
|P_1|\geq \frac{3^k}{m},\label{pflemma:pf-W-N_relat_code_color_set-1}
\end{equation}
 because, otherwise, we have
 \begin{equation}
 \bigg|\bigcup_{i=1}^{m}P_i\bigg|=\sum_{i=1}^{m}|P_i|<m\cdot\frac{3^k}{m}=3^k,
 \end{equation}
a contradiction.
The inequality \eqref{pflemma:pf-W-N_relat_code_color_set-1} immediately implies that
$|P_1|\geq \big\lceil3^{k}/m\big\rceil$. Thus, we obtain that
\begin{equation}
\max_{1\leq i\leq m}\gamma\big(S(P_i)\big)\geq \gamma\big(S(P_1)\big)\geq \gamma\big(S(\widehat{P}_1)\big) \geq\min_{\text{$P\subseteq\mathcal{Y}^{k}$ s.t. $|P|=\big\lceil\frac{3^{k}}{m}\big\rceil$ }}\gamma\big(S(P)\big),\label{pflemma:pf-W-N_relat_code_color_set-2}
\end{equation}
where $\widehat{P}_1$ is a subset of $P_1$ with $|\widehat{P}_1|= \big\lceil3^{k}/m\big\rceil$ and
the inequality $\gamma\big(S(P_1)\big)\geq \gamma\big(S(\widehat{P}_1)\big)$ in \eqref{pflemma:pf-W-N_relat_code_color_set-2} follows from the fact that a coloring of $S(P_1)$ restricted to its subgraph $S(\widehat{P}_1)$ is also a coloring of $S(\widehat{P}_1)$.
By \eqref{eq:gamma_m=min_max_color},
\begin{align}
\gamma_{m}&=\min_{\substack{\text{ all  partitions }\\\text{$\{P_1,P_2,\cdots,P_m\}$ of $\mathcal{Y}^{k}$}}}\max_{1\leq i\leq m}\gamma\big(S(P_i)\big)\\
&\geq\min_{\text{$P\subseteq\mathcal{Y}^{k}$ s.t. $|P|=\big\lceil\frac{3^{k}}{m}\big\rceil$ }}\gamma\big(S(P)\big)\label{pflemma:pf-W-N_relat_code_color_set-3.1}\\
&=\min_{\text{$P\subseteq\mathcal{Y}^{k}$ s.t. $|P|=\big\lceil\frac{3^{k}}{m}\big\rceil$ }}|\mathcal{A}^k+P|\label{pflemma:pf-W-N_relat_code_color_set-3.2}\\
&=M_{k}\Big(\Big\lceil\;\frac{3^{k}}{m}\;\Big\rceil\Big),\label{pflemma:pf-W-N_relat_code_color_set-3.3}
\end{align}
where the inequality \eqref{pflemma:pf-W-N_relat_code_color_set-3.1} follows from the inequality \eqref{pflemma:pf-W-N_relat_code_color_set-2}, which holds for each partition $\mathcal{P}$ of $\mathcal{Y}^k$ that includes $m$ blocks;   the  equality \eqref{pflemma:pf-W-N_relat_code_color_set-3.2} follows from the equality~\eqref{eq:gamma(S(P))-low-bound} and the notation~\eqref{eq:pf-W-def-set-M+L}; and  the  equality \eqref{pflemma:pf-W-N_relat_code_color_set-3.3} follows from the definition~\eqref{def:M_k(l)}.
We thus have proved the lemma.
}
\end{journalonly}


\begin{thebibliography}{99}

\bibitem{Shannon48}
C.~E.~Shannon, ``A mathematical theory of communication,''
\textit{Bell Sys. Tech. Journal}, 27: 379-423, 623-656, 1948.


\bibitem{Slepian-Wolf-IT73}
D.~Slepian and J.~K.~Wolf, ``Noiseless coding of correlated information sources,'' \textit{IEEE Trans. Inf. Theory}, vol.~19, no.~4, pp.~471--480, July~1973.


%\bibitem{Ahlswede-Korner_dis_compr_for_spec_func} R. F. Ahlswede and J. K{\"o}rner, ``Source coding with side information and a converse for degraded broadcast channels,'' \textit{IEEE Trans. Inf. Theory}, vol. 21, no. 6, pp. 629-637, Nov. 1975.

\bibitem{Korner-Marton-IT73}
J.~K\"{o}rner and K.~Marton, ``How to encode the modulo-two sum of binary
sources (Corresp.)'', \textit{IEEE Trans. Inf. Theory}, vol.~25, no.~2, pp.~219--221,
Mar.~1979.


\bibitem{Doshi_fun_comp_graph_color_sch}
V. Doshi, D. Shah, M. M{\'e}dard, and M. Effros, ``Functional compression through graph coloring,''
    \textit{IEEE Trans. Inf. Theory}, vol.~56, no.~8, pp.~3901-3917, Aug.~2010.

\bibitem{Feizi-Medard}
S. Feizi and M. M{\'e}dard, ``On network functional compression,''
\textit{IEEE Trans. Inf. Theory}, vol.~60, no.~9, pp.~5387-5401, Sept.~2014.

\bibitem{Orlitsky-Roche_general_side_inf_model_rat_reg}
A. Orlitsky and J. R. Roche, ``Coding for computing,''
     \textit{IEEE Trans. Inf. Theory}, vol.~47, no.~3, pp.~903-917, Mar.~2001.









\bibitem{Witsenhausen-IT-76}
H.~Witsenhausen, ``The zero-error side information problem and chromatic numbers,'' \textit{IEEE Trans. Inf. Theory}, vol.~22, pp.~592--593, Sept. 1976.


\bibitem{Alon-Orlitsky_source_cod_graph_entropies}
N. Alon and A. Orlitsky, ``Source coding and graphs entropies,''
    \textit{IEEE Trans. Inf. Theory}, vol.~42, pp.~1329-1339, Sept.~1996.

\bibitem{Koulgi_zero-error-cod_cor_inf_sour}
P. Koulgi, E. Tuncel, S. Regunathan, and K. Rose, ``On zero-error coding of correlated sources,''
    \textit{IEEE Trans. Inf. Theory}, vol.~49, pp.~2856-2873, Nove.~2003.

\bibitem{Shannon_zero_cap_noisy_channel}
C. E. Shannon, ``The zero error capacity of a noisy channel,'' \textit{IEEE Trans. Inf. Theory}, vol.~2, no.~3, pp.~8-19, Sept.~1956.

\bibitem{Rosenfeld_frac_color_num_bound_shannon_cap}
M. Rosenfeld, ``On a problem of Shannon,'' \textit{proc. Amer. Math. Soc.}, vol.~18, pp.~315-319, 1967.

 \bibitem{Lovasz_shannon_cap_graph}
L. Lov\'{a}sz, ``On the Shannon capacity of a graph,'' \textit{IEEE Trans. Inf. Theory}, vol.~25, pp.~1-7, Jan.~1979.

\bibitem{Haemers_shannon_lovasz_problems}
W. Haemers, ``On some problems of Lov{\'a}sz concerning the Shannon capacity of a graph,''
 \textit{IEEE Trans. Inf. Theory}, vol.~25, pp.~231-232, Mar.~1979.

\bibitem{Appuswamy11}
R.~Appuswamy, M.~Franceschetti, N.~Karamchandani, and K.~Zeger, ``Network coding for computing: cut-set bounds,''
 \textit{IEEE Trans. Inf. Theory}, vol.~57, no.~2, pp.~1015-1030, Feb.~2011.

\bibitem{Appuswamy13}
R.~Appuswamy, M.~Franceschetti, N.~Karamchandani, and K.~Zeger, ``Linear codes, target function classes, and network computing capacity,''
\textit{IEEE Trans. Inf. Theory}, vol.~59, no.~9, pp.~5741-5753, Sept.~2013.


\bibitem{Ramamoorthy-Langberg-JSAC13-sum-networks}
A.~Ramamoorthy and M.~Langberg, ``Communicating the sum of sources over a network,''
\textit{IEEE J. Sel. Areas Commun.}, vol.~31, no.~4, pp.~655-665, April 2013.


\bibitem{Rai-Dey-TIT-2012}
B.~Rai and B.~Dey, ``On network coding for sum-networks,''
 \textit{IEEE Trans. Inf. Theory}, vol.~58, no.~1, pp.~50-63, Jan.~2012.


\bibitem{Tripathy-Ramamoorthy-IT18-sum-networks}
A.~Tripathy and A.~Ramamoorthy, ``Sum-networks from incidence structures: construction and capacity analysis,''
\textit{IEEE Trans. Inf. Theory}, vol.~64, no.~5, pp.~3461-3480, May~2018.


\bibitem{Huang_Comment_cut_set_bound}
 C. Huang, Z. Tan, S. Yang, and X. Guang, ``Comments on cut-set bounds on network function computation,''
 \textit{IEEE Trans. Inf. Theory}, vol.~64, no.~9, pp.~6454-6459, Sept.~2018.


\bibitem{Guang_Improved_upper_bound}
X. Guang, R. W. Yeung, S. Yang and C. Li, ``Improved upper bound on the network function computing capacity,''
\textit{IEEE Trans. Inf. Theory}, vol.~65, no.~6, pp.~3790-3811, June~2019.


\bibitem{Wang-Tightness_upper_bound}
J. Wang, S. Yang and C. Li, ``On the tightness of a cut-set bound on network function computation,''
in \textit{Proc. IEEE Int. Symp. Inf. Theory (ISIT)}, Vail, CO, USA, June 2018, pp.~1824-1828.






\bibitem{Ahlswede_net_flow}
R. Ahlswede, N. Cai, S.-Y. R. Li and R. W. Yeung, ``Network information flow,''
\textit{IEEE Trans. Inf. Theory}, vol.~46, no.~4, pp.~1204-1216, July~2000.


\bibitem{Li-Yeung-Cai-2003}
S.-Y. R. Li, R. W. Yeung, and N. Cai, ``Linear network coding,''
\textit{IEEE Trans. Inf. Theory}, vol.~49, no.~2, pp.~371-381, July~2003.

\bibitem{Koetter-Medard-algebraic}
R. Koetter and M. M\'{e}dard, ``An algebraic approach
to network coding,''
 \textit{IEEE/ACM Trans. Netw.}, vol.~11, no.~5,
pp.~782-795, Oct.~2003.

\bibitem{co-construction}
S. Jaggi, P. Sanders, P. A. Chou, M. Effros, S. Egner, K. Jain, and
L. M. G. M. Tolhuizen, ``Polynomial time algorithms for multicast
network code construction,''
\textit{IEEE Trans. Inf. Theory}, vol.~51, no.~6, pp.~1973-1982, June~2005.


\bibitem{Yeung-book}
R. W. Yeung,
\textit{Information Theory and Network Coding}. New York: Springer, 2008.

\bibitem{Zhang-book}
R. W. Yeung, S.-Y. R. Li, N. Cai, and Z. Zhang, ``Network coding theory,''
\textit{Foundations and Trends in Communications and Information Theory}, vol.~2, nos.4 and 5,
pp.~241-381, 2005.

\bibitem{Fragouli-book}
C. Fragouli and E. Soljanin, ``Network coding fundamentals,''
\textit{Foundations and Trends in Networking}, vol.~2, no.~1, pp.~1-133, 2007.

\bibitem{Fragouli-book-app}
C. Fragouli and E. Soljanin, ``Network coding applications,''
\textit{Foundations and Trends in Networking}, vol.~2, no.~2, pp.~135-269, 2007.

\bibitem{Ho-book}
T. Ho and D. S. Lun, \textit{Network Coding: An Introduction}. Cambridge, U.K.: Cambridge Univ. Press, 2008.


\bibitem{Wyner-Ziv_rat_distortion_fun_sid_inf} A. D. Wyner and J. Ziv, ``The rate-distortion function for source coding
with side information at the decoder,''
\textit{IEEE Trans. Inf. Theory}, vol.~22,
no.~1, pp.~1-10, Jan.~1976.

\bibitem{Yamamoto_rat_distortion_gener_fun_sid_inf} H. Yamamoto, ``Wyner-Ziv theory for a general function of the
correlated sources (corresp.),''
 \textit{IEEE Trans. Inf. Theory}, vol.~28, no.~5,
pp.~803-807, Sept.~1982.

\bibitem{Feng_rat_distortion_net_fun_sid_inf} H. Feng, M. Effros, and S. Savari, ``Functional source coding for
networks with receiver side information,''
in \textit{Proc. Allerton Conf.
Commun., Control, Comput.}, Sept.~2004, pp.~1419-1427.

\bibitem{Berger-Yeung_multi_source_cod_dist_cri} T. Berger and R. W. Yeung, ``Multiterminal source encoding with
one distortion criterion,''
 \textit{IEEE Trans. Inf. Theory}, vol.~35, no.~2, pp.~228-236, Mar.~1989.

\bibitem{Barros_cor_sour_rat_dist_reg} J. Barros and S. D. Servetto, ``On the rate-distortion region for separate
encoding of correlated sources,'' in \textit{Proc. IEEE Int. Symp. Inf. Theory (ISIT)}, Yokohama, Japan, June 2003, pp.~171.


\bibitem{Wagner_Rat_reg_Guassian_sou_cod} A. B. Wagner, S. Tavildar, and P. Viswanath, ``Rate region of the quadratic Gaussian two-encoder source-coding problem,''
    \textit{IEEE Trans. Inf. Theory}, vol.~54, no.~5, pp.~1938-1961, 2008.



%
%\bibitem{Korner-graph-entropy}
%J.~K\"{o}rner, ``Coding of an information source having ambiguous alphabet
%and the entropy of graphs,'' in \textit{Proc. 6th Prague Conf Inf. Theory}, 1973,
%pp.~411-425.
%
%\bibitem{Shayevitz-un-restr-distri-fun-comp}
%O.~Shayevitz, ``Distributed Computing and the Graph Entropy Region,'' in  \textit{IEEE Trans. Inf. Theory}, vol.~60, no.~6, pp.~3435-3449, June 2014.
%






	

	%\bibitem{Ahlswede_net_flow} R. Ahlswede, N. Cai, S.-Y. R. Li and R. W. Yeung, ``Network information flow", \textit{IEEE Trans. Inf. Theory}, vol. 46, no. 4, pp. 1204-1216, Jul. 2000.
%	\bibitem{Doshi_fun_compression_graph_color}V. Doshi, D. Shah, M\'{e}dard and M. Effros, ``Functional compression through graph coloring", \textit{IEEE Trans. Inf. Theory}, vol. 56, no. 8, pp. 3901-3917, Aug 2010.
%	\bibitem{Appuswamy_cut_set_bound}  R. Appuswamy, M. Franceschetti, N. Karamchandani and K. Zeger,  ``Network coding for computing: Cut-set bounds", \textit{IEEE Trans. Inf. Theory}, vol. 57, no. 2, pp. 1015-1030, Feb. 2011.
%	\bibitem { R. Appuswamy_lin_cod_fun_class_and_comp_cap}   R. Appuswamy, M. Franceschetti, N. Karamchandani, and K. Zeger, ``Linear codes, target function classes, and network computing capacity,'' \textit{IEEE Trans. Inf. Theory}, vol. 59, no. 9, pp. 5741-5753, Jan. 2013.
%	\bibitem   {R.Appuswamy_comp_lin_fun_liner_coding} R. Appuswamy and M. Franceschetti,  ``Computing linear functions by linear coding over networks,'' \textit{IEEE Trans. Inf. Theory}, vol. 60, no. 1, pp. 422-431, Jan. 2014.
%	\bibitem {Huang_Comment_cut_set_bound}   C. Huang, Z. Tan, S. Yang, and X. Guang, `` Comments on cut-set bounds on network function computation,'' \textit{IEEE Trans. Inf. Theory}, vol. 64, no. 9, pp. 6454-6459, Sep. 2018.
%	\bibitem{Huang_upper_bound} C. Huang, Z. Tan and S. Yang,  ``Upper bound on function computation in directed acyclic networks," \textit{Information Theory Workshop (ITW) 2015 IEEE}, Apr. 2015.
%	\bibitem {Guang_im_up_bound_ITW2016}  X. Guang, S.Yang, and C. Li, ``An improved upper bound on network function computation using cut-set partition,'' in \textit{IEEE Inf. Theory Workshop(ITW)}, Cambridge, U.K., Sep. 2016, pp. 11-15.
%	\bibitem {Guang_enh_bound_ISIT2018}      X. Guang, R. W. Yeung, S.Yang, and C. Li, ``An enhanced capacity bound for network function computation,'' in \textit{Proc. IEEE Int. Symp. Inf. Theory,} Cambridge, Vail, CO, USA, Jun. 2018, pp. 1819-1823.
%
%	\bibitem{Guang_Improved_upper_bound} X. Guang, R. W. Yeung, S. Yang and C. Li, ``Improved Upper Bound on the Network Function Computing Capacity,"  \textit{IEEE Trans. Inf. Theory}, vol. 65, no. 6, pp. 3790-3811, June 2019.
%

%
%   \bibitem{Zhang-book}
%R. W. Yeung, S.-Y. R. Li, N. Cai, and Z. Zhang, ``Network coding theory,''
%\textit{Foundations and Trends in Communications and Information Theory}, vol. 2, nos.4 and 5, pp. 241-381, 2005.
%
%   \bibitem{yeung08b}
%R.~W.~Yeung, \emph{Information Theory and Network Coding}.\hskip 1em plus 0.5em
%  minus 0.4em\relax Springer, 2008.
\end{thebibliography}
\end{document}